\documentclass[11pt,letterpaper]{article}

\usepackage[english]{babel}

\usepackage[margin=1in]{geometry}
\usepackage{amsfonts} 
\usepackage{xspace}
\usepackage{amsmath,amssymb,amsthm}
\usepackage{graphicx}
\usepackage{xcolor}
\usepackage[colorlinks=true, allcolors=blue]{hyperref}
\usepackage{cleveref}
\usepackage{todonotes}
\usepackage{mathtools}
\usepackage{enumitem}
\usepackage{algorithm}
\usepackage{algorithmic}
\usepackage{multicol}
\usepackage{multirow}
\usepackage{nicefrac}
\usepackage{tikz}
\usepackage{setspace}
\usetikzlibrary{calc}

\makeatletter
\renewcommand{\paragraph}{%
  \@startsection{paragraph}{4}%
  {\z@}{1ex \@plus 1ex \@minus .2ex}{-1em}%
  {\normalfont\normalsize\bfseries}%
}
\renewcommand{\subparagraph}{%
	\@startsection{subparagraph}{5}%
	{\parindent}{1ex \@plus 1ex \@minus .2ex}{-1em}%
	{\normalfont\normalsize\bfseries}%
}
\makeatother

\newtheorem{theorem}{Theorem}[section]
\makeatletter

\newtheorem{fact}[theorem]{Fact}

\newtheorem{lemma}[theorem]{Lemma}

\theoremstyle{definition}
\newtheorem{definition}[theorem]{Definition}
\newtheorem{remark}[theorem]{Remark}

\newcommand{\eps}{\varepsilon}
\renewcommand{\epsilon}{\varepsilon}

\newcommand{\eat}[1]{}

\newcommand{\R}{\mathbb{R}}
\DeclareMathOperator*{\E}{\ensuremath{\mathbb{E}}}
\newcommand{\Z}{\mathbb{Z}}

\DeclareMathOperator{\diam}{diam}

\newcommand{\VC}{\ensuremath{\mathrm{VC}}}
\newcommand{\far}{\ensuremath{\mathrm{far}}}

\newcommand{\inn}{{\mathrm{in}}}

\newcommand{\ProblemName}[1]{\textsc{#1}}
\newcommand{\kzC}{\ProblemName{$(k, z)$-Clustering}\xspace}
\newcommand{\kMedian}{\ProblemName{$k$-Medians}\xspace}

\newcommand{\kMeans}{\ProblemName{$k$-Means}\xspace}

\eat{

}

\makeatletter
\newcommand*{\rom}[1]{\expandafter\@slowromancap\romannumeral #1@}
\makeatother

\DeclareMathOperator{\poly}{poly}
\DeclareMathOperator{\polylog}{polylog}

\DeclareMathOperator{\dist}{dist}

\DeclareMathOperator{\Balls}{Balls}

\DeclareMathOperator{\Var}{Var}
\DeclareMathOperator*{\argmin}{arg\,min}

\DeclarePairedDelimiter{\abs}{\lvert}{\rvert}

\DeclarePairedDelimiter{\norm}{\lVert}{\rVert}
\DeclarePairedDelimiter{\ceil}{\lceil}{\rceil}
\DeclarePairedDelimiterX{\inner}[2]{\langle}{\rangle}{#1 , #2}
\DeclarePairedDelimiterXPP\normp[1]{}{\lVert}{\rVert}{_p}{#1}

\NewDocumentCommand\cost{o}{%
    \IfNoValueTF{#1}%
    {\operatorname{cost}}
    {\operatorname{cost}^{(#1)}}
}

\NewDocumentCommand\Ballh{o}{%
    \IfNoValueTF{#1}%
    {\operatorname{Ball}}
    {\operatorname{Ball}_{H}^{(#1)}}
}

\NewDocumentCommand\Ballsh{o}{%
    \IfNoValueTF{#1}%
    {\operatorname{Balls}}
    {\operatorname{Balls}_{H}^{(#1)}}
}

\allowdisplaybreaks

\title{Towards Tight Robust Coresets for $k$-Medians Clustering}
\author{Lingxiao Huang\thanks{
		Email: \texttt{huanglingxiao@nju.edu.cn}
	}\\
	Nanjing University
	\and
	Zhenyu Jiang\thanks{
		Email: \texttt{zhenyujiang@smail.nju.edu.cn}
	}\\
	Nanjing University
	\and
	Yi Li\thanks{ 
		Email: \texttt{yili@ntu.edu.sg}}\\
	Nanyang Technological University
	\and
	Xuan Wu\thanks{ 
		Email: \texttt{xuan.wu@ntu.edu.sg}}\\
	Nanyang Technological University
}
\date{}

\begin{document}

\pagenumbering{roman} 

\maketitle

\begin{abstract}
This paper considers coresets for the robust $k$-medians problem with $m$ outliers, and new constructions in various metric spaces are obtained.
Specifically, for metric spaces with a bounded VC or doubling dimension $d$, the coreset size is $O(m) + \tilde{O}(kd\varepsilon^{-2})$, which is optimal up to logarithmic factors. For Euclidean spaces, the coreset size is $O(m\varepsilon^{-1}) + \tilde{O}(\min\{k^{4/3}\varepsilon^{-2}, k\varepsilon^{-3}\})$, improving upon a recent result by Jiang and Lou (ICALP 2025). These results also extend to robust $(k,z)$-clustering, yielding, for VC and doubling dimension, a coreset size of $O(m) + \tilde{O}(kd\varepsilon^{-2z})$ with the optimal linear dependence on $m$. This extended result improves upon the earlier work of Huang et al. (SODA 2025). The techniques introduce novel dataset decompositions, enabling chaining arguments to be applied jointly across multiple components.
\end{abstract}

\tableofcontents

\newpage

\pagenumbering{arabic}

\section{Introduction}
Let $(M,\dist)$ be a metric space and let $X$ denote the dataset. 
The \kzC problem asks to solve the following optimization problem 
\[
	C^\ast := \argmin_{C\in \binom{M}{k}} \, \cost_z(X,C), \quad\text{ where }\quad  
    \cost_z(X,C) := \sum_{x\in X} \Big(\min_{c\in C}\dist(x,c)\Big)^z
\]
and $\binom{M}{k}$ denotes the set of all $k$-element subsets of $M$. The special case $z=1$, known as the \kMedian problem, is one of the most extensively studied clustering problems and the main focus of this paper. Accordingly, we omit the subscript in $\cost_1$ and simply write $\cost$.

In practice, the presence of outliers poses a significant challenge to the optimization problem above. For example, a small number of adversarially inserted outliers can drastically bias the solution by forcing the algorithm to designate them as centers, thus failing to find the true underlying clusters of the dataset. To address this issue, we consider the following robust version of \kMedian, proposed by~\cite{charikar2001algorithms}, which seeks to find a minimizer that tolerates (at most) $m$ outliers:
\begin{equation}\label{eqn:robust_cost_intro}
	C^\ast := \argmin_{C\in \binom{M}{k}} \, \cost^{(m)}(X,C), \quad\text{ where }\quad 
	\cost^{(m)}(X,C):=\min_{L\in \binom{X}{m}}\cost(X\setminus L,C).
\end{equation}

Solving robust \kMedian is considerably more challenging than solving its vanilla counterpart (i.e., without outliers). Existing constant-factor approximation algorithms \cite{chen2008constant,krishnaswamy2018constant,gupta2024structural} have a high-order polynomial running time. Moreover, while fixed-parameter tractable (FPT) algorithms for $(1+\eps)$-approximations are known~\cite{feng2019improved,agrawal2023clustering},
their running times remain far from near-linear in the dataset size. Hence, designing scalable algorithms for robust \kMedian continues to be an active area of research.

Among the various approaches, coresets have emerged as a powerful technique for scalable clustering. An $\eps$-coreset is a small weighted subset of the dataset that approximates the \kMedian cost within a relative error of $\eps$ for every set of $k$ centers.
By constructing such a coreset, one can solve robust \kMedian more efficiently, as the runtime depends on the coreset size rather than on the size of the full dataset.
Over the last two decades, extensive research has yielded nearly optimal coresets for various metric spaces \cite{DBLP:journals/dcg/Har-PeledK07,Chen09,langberg2010universal,DBLP:conf/stoc/FeldmanL11,cohen2021new,Cohen2022Towards,cohen2022improved,huang2024optimal,bansal2024sensitivity,Cohen25}. 
For metric spaces with a finite VC dimension\footnote{For brevity, we say a metric space has VC dimension $d_{\VC}$ when the VC dimension of its ball family is $d_{\VC}$.} $d_{\VC}$, the best known $\eps$-coreset for $\kMedian$ has a size of\footnote{Throughout this paper, $\tilde{O}(f) = O(f\polylog f)$.} $\tilde{O}(kd_{\VC}\eps^{-2})$~\cite{Cohen25}. Similarly, for metric spaces with a finite doubling dimension $d_D$, the coreset size is $\tilde{O}(kd_D\eps^{-2}$)~\cite{cohen2021new}. Both of these bounds are known to be tight up to logarithmic factors~\cite{Cohen2022Towards}.
However, the coreset size need not always depend on the dimension. In Euclidean space, a dimension-independent coreset of size $\tilde{O}(\min(k^{4/3} \eps^{-2}, k\eps^{-3}))$ can be achieved~\cite{Cohen2022Towards,cohen2022improved}, and this bound is near-tight for $\eps$ that is not too small~\cite{huang2024optimal}.

The success of vanilla coresets does not carry over to the robust setting, and our understanding of robust coresets lags significantly behind.
For over a decade, the best-known coreset constructions for robust \kMedian with $m$ outliers (see Definition~\ref{def:rbcoreset}) either had size exponential in $k+m$~\cite{feldman2012data} or achieved only bi-criteria guarantees~\cite{DBLP:conf/stoc/FeldmanL11,huang2018epsilon}.
A recent breakthrough by Braverman et al.~\cite{braverman2022power} introduced a general framework for clustering with constraints, which led to substantial progress on this problem.
In particular, the first polynomial-size coreset for robust \kMedian in Euclidean space had size $\tilde{O}(m+k^3\eps^{-5})$~\cite{huang2022near}.
Subsequent work~\cite{huang2025coresets} improved this bound to $\tilde{O}(m+k^2\eps^{-4})$ and extended the construction to metric spaces with finite VC or doubling dimension, yielding coreset sizes of $\tilde{O}(m+k^2 d_{\VC}\eps^{-4})$ and $\tilde{O}(m+k^2 d_{D}\eps^{-2})$, respectively.
While these results achieve near-linear dependence on the number of outliers $m$, their $\poly(k/\eps)$ dependence remains substantially worse than the optimal vanilla coreset size $Q$. 

Another recent work~\cite{jiang2025coresets} adopts an alternative approach to construct robust coresets in Euclidean space, achieving a size of $\tilde{O}\big(m\eps^{-1}\cdot \min(k,\eps^{-1}) + Q\big)$.
This approach can be generalized to metric spaces with finite VC or doubling dimension, resulting in the same robust coreset size but with $Q$ now denoting the vanilla coreset size specific to the respective metric space.
While this result matches the vanilla coreset size in its $Q$ term, the other ``outlier term'' is much worse than near-linear in $m$ due to additional multiplicative factors.

On the other hand, we know that any robust \kMedian coreset must have size at least $\Omega(m+Q)$ (see Theorem~\ref{thm:lower} in Appendix \ref{sec:lb}). This leads to the following question.

\begin{quote}
\textit{%
In which family of metric spaces can an $\eps$-coreset for \kMedian with $m$ outliers be constructed with size $\tilde{O}(m+Q)$, where $Q$ denotes the current optimal size of vanilla $\eps$-coreset in the same metric space?
}
\end{quote}

Tight vanilla coreset sizes are often achieved through chaining arguments~\cite{bansal2024sensitivity,Cohen25}. 
In contrast, all the robust coreset results discussed so far do not use chaining arguments and rely on geometric decomposition techniques, which partition the dataset into significantly more pieces with desirable properties than are needed for vanilla coresets. 
This methodological gap exists because applying vanilla geometric decomposition and chaining to the robust setting faces a fundamental difficulty: even a minor perturbation to the center set can drastically change the outliers. 
Such instability can either invalidate the correctness of the chaining argument or, if addressed with a simple fix, lead to an ineffective, large coreset size.

In this paper, we address this difficulty by introducing a new form of dataset decomposition which allows for a careful adaptation of the chaining arguments.
As a result, we obtain robust \kMedian coresets whose sizes nearly match the best known vanilla bounds for metric spaces with finite VC or doubling dimension.
Moreover, for Euclidean space, where vanilla coresets admit dimension-independent size bounds, we obtain a coreset of size $\tilde{O}(m\eps^{-1}+Q)$.

\subsection{Our results}
\label{sec:result}

\newcommand{\specialcell}[2][c]{%

\begin{tabular}[#1]{@{}c@{}}#2\end{tabular}}
  
    \begin{table}[t]
		\centering
		\footnotesize
		
		\begin{tabular}{|c|c|c|c|}
			\hline
			Metric space & Prior results & Our results \\ \hline
			VC  & \specialcell{$O(m) + \tilde{O}(k^2 d_{\VC} \eps^{-2z-2})$ \cite{huang2025coresets} \\ $\tilde{O}(\min\{km\eps^{-1}, m\eps^{-2z}\}) + \text{Vanilla size}$ \cite{jiang2025coresets} \\ $\Omega(m+k d_{\VC} \eps^{-2})$ ($d_{\VC} = \log n$) \cite{huang2022near,Cohen2022Towards} } & \specialcell{ $O(m)+\tilde{O}(k d_{\VC} \eps^{-2z})$ $\ast$ \\ (Theorems \ref{thm_gen} and \ref{thm:kzC})}
                \\ \hline
			Doubling & \specialcell{$O(m) + \tilde{O}(k^2 d_D \eps^{-2z})$ \cite{huang2025coresets} \\ $\tilde{O}(\min\{km\eps^{-1}, m\eps^{-2z}\} ) + \text{Vanilla size}$ \cite{jiang2025coresets} \\ $\Omega(m +k d_D \eps^{-2})$ \cite{huang2022near,Cohen2022Towards} \\ $\Omega(k d_D \eps^{-\max\{z,2\}}/\log k)$ ($\eps = \Omega(k^{-1/(z+2)})$) \cite{huang2024optimal} } & \specialcell{$O(m) + \tilde{O}(k d_{D} \eps^{-2z})$ $\ast$ \\ (Theorems \ref{thm_doubling} and \ref{thm:kzC}) } \\ \hline
			Euclidean & \specialcell{$O(m) + \tilde{O}(k^2 \eps^{-2z-2})$ \cite{huang2025coresets} \\ $\tilde{O}(\min\{km\eps^{-1}, m\eps^{-2z}\} ) + \text{Vanilla size}$ \cite{jiang2025coresets} \\ $\Omega(m +k \eps^{-2})$ \cite{huang2022near,Cohen2022Towards} \\ $\Omega(k \eps^{-z-2})$ ($\eps = \Omega(k^{-1/(z+2)})$) \cite{huang2024optimal} } & \specialcell{$O(m \eps^{-z})+\tilde{O}(\min\{k^{\frac{2z+2}{z+2}}\eps^{-2}, k\eps^{-z-2}\})$ \\ (Theorems \ref{thm_Euc} and \ref{thm_Euc_kzC}) \\ $O(m) + \tilde{O}(k \eps^{-2z-2})$ \\ (Corollary of Theorems \ref{thm_doubling} and \ref{thm:kzC}) } \\ \hline
		\end{tabular}
        
\caption{Comparison of state-of-the-art coreset sizes and our results for robust \kzC. 
Results marked with $\ast$ indicate near-optimal worst-case bounds for robust \kMedian ($z=1$).
Here, \( d_{\VC} \) and \( d_D \) denote the VC and doubling dimensions, respectively.
All results assume that $z\geq 1$ is a constant and omit factors of $2^{O(z)}$ or $z^{O(z)}$ in the coreset sizes.
For Euclidean results derived as corollaries of Theorems \ref{thm_doubling} and \ref{thm:kzC}, see the discussion preceding Theorem \ref{thm_Euc}.
Note that the work \cite{jiang2025coresets} provides a reduction from robust to vanilla coreset sizes.
}

\label{tab:result}
\end{table}

We now present our results in detail, focusing on robust \kMedian.
The extension to general robust \kzC, defined similarly to \eqref{eqn:robust_cost_intro} but using 
$\cost_z(\cdot,\cdot)$ in place of $\cost(\cdot,\cdot)$, is deferred to Appendix~\ref{sec:kzC}. A summary of our results can be found in Table~\ref{tab:result}. 
%
%

Throughout this paper, we assume the availability of a distance oracle that returns \( \dist(x,y) \) in constant time for any queried pair \( (x,y) \).

First, we have the following theorem for coresets in metric spaces with a finite VC dimension.

\begin{theorem}[Finite VC dimension] \label{thm_gen} 
Let $(M,\dist)$ be a metric space with VC dimension $d_{\VC}\geq 1$. 
There exists an algorithm that, given a dataset $X\subseteq M$ of size $n\geq 1$, constructs an $(\eps,m)$-robust coreset of 
$X$ for robust \kMedian with size $O(m)+\tilde{O}(k\cdot d_{\VC}\cdot \eps^{-2})$ in $O(nk)$ time.
\end{theorem}

In the special case where \( m = 0 \), this coreset size matches that of the vanilla \( k \)-median coreset \cite{Cohen25}, which is worst-case optimal. 
This follows from the fact that in a metric space with $n$ points and VC dimension \( \log n \), the tight coreset size for vanilla \( k \)-median is precisely \( \tilde{\Theta}(k \cdot \log n \cdot \eps^{-2}) \) \cite{cohen2021new,Cohen2022Towards}.  
Moreover, the \( O(m) \) factor aligns with the coreset lower bound \( \Omega(m) \) established in \cite{huang2022near}.  
Thus, we obtain the first tight VC-dimension-based coreset for robust \( k \)-median.

Next, we give the coreset result in doubling spaces.

\begin{theorem}[Finite doubling dimension] 
\label{thm_doubling} 
Let $(M,\dist)$ be a metric space with doubling dimension $d_D\geq 1$.
There exists an algorithm that, given a dataset $X\subseteq M$ of size $n\geq 1$, constructs an $(\eps,m)$-robust coreset of 
$X$ for robust \kMedian with size $O(m)+\tilde{O}(k\cdot d_D\cdot \eps^{-2})$ in $O(nk)$ time.
\end{theorem}

Similarly to the VC-dimension case, this doubling-dimension-based result extends the vanilla result from \cite{Cohen25} by introducing an additive \( O(m) \) term.  
This size is optimal, as it matches the lower bound of the coreset size,  \( \Omega(m + k \cdot d_D \cdot \eps^{-2}) \), established in \cite{huang2022near,Cohen2022Towards}.  

Finally, we consider Euclidean spaces.
Since the Euclidean metric \( (\mathbb{R}^d, \|\cdot\|_2) \) has doubling dimension at most \( d+1 \), Theorem~\ref{thm_doubling} and its generalization in Theorem~\ref{thm:kzC} imply the existence of a Euclidean coreset of size \( O(m) + \tilde{O}(k  d  \eps^{-2}) \).  
By applying the standard iterative size reduction and terminal embedding technique from \cite{BJKW21,NN19}, which reduces the ambient dimension $d$ to $\tilde{O}(\eps^{-2})$, this coreset size can be further reduced to \( O(m) + \tilde{O}(k\eps^{-4} )\). 
This already improves upon the prior bound of \( O(m) + \tilde{O}(k^2 \eps^{-4}) \) from \cite{huang2025coresets}, achieving an improvement by a factor of \( k \).  
However, when \( m = 0 \), this size exceeds the vanilla coreset bound \( \tilde{O}(\min\{k^{4/3}\eps^{-2}, k\eps^{-3}\}) \) from \cite{cohen2021new,cohen2022improved,huang2024optimal}.
This constrasts with the VC and the doubling cases, in which our results extend the previous vanilla results.
This difference stems from the special geometric structure of Euclidean metrics, which enables further reductions in coreset size but also complicates obtaining tight bounds.
Indeed, even for the vanilla setting, the optimal coreset size in Euclidean space remains partially open~\cite{huang2024optimal}.
To address the added complexity of robustness, we develop a new algorithm specifically tailored for robust Euclidean \kMedian, yielding the following theorem.

\begin{theorem}[Euclidean spaces] \label{thm_Euc}
Let $(M,\dist)=(\mathbb{R}^d,\|\cdot\|_2)$. 
There exists an algorithm that, given a dataset $X\subset \mathbb{R}^d$ of size $n\geq 1$, constructs an $(\eps,m)$-robust coreset of 
$X$ for the robust \kMedian with size $O(m \eps^{-1})+\tilde{O}(\min\{k^{4/3}\eps^{-2}, k\eps^{-3}\})$ in $O(nk)$ time.
\end{theorem}

This Euclidean coreset size extends the vanilla result from \cite{cohen2021new,cohen2022improved,huang2024optimal} by introducing an additive \( O(m \eps^{-1}) \) term.  
An interesting open question is whether this additional term can be further improved to \( O(m) \), matching the $\Omega(m)$ lower bound \cite{huang2022near}.  
Moreover, Theorem~\ref{thm_Euc} improves the recent bound \( \tilde{O}(\min\{m\eps^{-2}, km \eps^{-1}\} + \min\{k^{4/3}\eps^{-2}, k\eps^{-3}\}) \) from \cite{jiang2025coresets} by replacing the term \(\min\{m\eps^{-2}, km \eps^{-1}\}\) with a smaller \( m \eps^{-1} \).

\paragraph{Extension to general $z\geq 1$}
We further extend our results to general robust \((k,z)\)-clustering under various metrics.
For metric spaces with a finite VC or doubling dimension, this extension is formalized in Theorem~\ref{thm:kzC} (see Appendix~\ref{sec:kzC}).
The resulting coreset size is \( O(m) + \tilde{O}(k \cdot \min\{d_{\VC}, d_D\} \cdot \eps^{-2z}) \).  
Compared to the previous bound of \( O(m) + \tilde{O}(k^2 \cdot \min\{ d_{\VC} \cdot \eps^{-2z-2}, d_D \cdot \eps^{-2z} \}) \) from \cite{huang2025coresets}, our result significantly reduces the dependence on \( k \) from \( k^2 \) to \( k \).  
Notably, our coreset achieves a linear dependence on \( k \), which is optimal.

In Euclidean spaces, we also extend Theorem~\ref{thm_Euc} to general \( z \geq 1 \), as stated in Theorem~\ref{thm_Euc_kzC}. 
The resulting coreset size is \( O(m \eps^{-z}) + \tilde{O}(\min\{k^{\frac{2z+2}{z+2}} \eps^{-2}, k \eps^{-z-2}\}) \), introducing an additive \( O(m \eps^{-z}) \) term to the vanilla size of \cite{huang2024optimal}.  
Compared to the recent bound \( \tilde{O}(\min\{m \eps^{-2z},\; km \eps^{-1}\} + \min\{k^{\frac{2z+2}{z+2}} \eps^{-2}, k \eps^{-z-2}\}) \) from \cite{jiang2025coresets}, our result improves the \( m \eps^{-2z} \) term to \( m \eps^{-z} \).

\subsection{Other related work}
\label{sec:related}

\paragraph{Coresets for constrained \kMedian}
Recently, there has been significant interest in developing coresets for constrained variants of the \kMedian problem. 
For capacitated \kMedian, where centers have capacity constraints, 
a coreset of size $\tilde{O}(k^2\eps^{-3}\log^2 n)$ was first constructed for instances in Euclidean space~\cite{cohen2022fixed}. 
Subsequent research yielded improvements, first to a size of $\tilde{O}(k^3\eps^{-6})$~\cite{braverman2022power}, which notably is independent of both the dataset size and dimension, and more recently to $\tilde{O}(k^2\eps^{-4})$~\cite{huang2025coresets}.
A closely related variant is the fair \kMedian problem.
The first coreset developed for this variant had a size of $\tilde{O}(k\eps^{-d}\log n)$~\cite{Schmidt19}. 
This was followed by coresets of size $\tilde{O}(k^2\eps^{-d})$~\cite{HJV19} and $\tilde{O}(k^2\eps^{-3}d\log n)$~\cite{DBLP:conf/icalp/BandyapadhyayFS21}.
These bounds were later improved to $\tilde{O}(k^3\eps^{-6})$~\cite{braverman2022power} and eventually to $\tilde{O}(k^2\eps^{-4})$~\cite{huang2025coresets}. 
Other variants of constraints have also been studied, including coresets for clustering with missing values~\cite{NEURIPS2021_90fd4f88} and fault-tolerant coresets~\cite{huang2025coresets}.
All these results can be extended to general \kzC.

\paragraph{Coresets for other problems}
Beyond clustering, coresets have been applied to a wide range of optimization and machine learning problems.
These include low-rank approximation~\cite{cohen2017input}, principal component analysis~\cite{feldman2020turning} and mixture models~\cite{huang2020coresets,lucic2018training}
and various regression tasks~\cite{chhaya2020coresets,MMWY22,WY23:sensitivity,MO24:sensitivity,AP24,gajjar:COLT,LT25}. Coresets have also been used for other robust optimization problems. For example, Wang et al.~\cite{wang2021robust} proposed a framework to compute robust local coresets for continuous-and-bounded learning, and Huang et al.~\cite{huang2022coresets} introduced a unified framework to construct a coreset for general Wasserstein distributionally robust optimization problem.


\subsection{Paper organization}
In Section \ref{sec:pre}, we review fundamental concepts, including VC and doubling dimensions, and formally define the robust clustering problem. Section \ref{sec:overview} provides a high-level technical overview.

Section \ref{sec:three_coresets} introduces three coreset constructions for the key data structure: \( (r,k) \)-instances.  %
Section \ref{sec:proof_thm} presents algorithms that integrate these coreset constructions and provides proofs for Theorems \ref{thm_gen}--\ref{thm_Euc}.  

The analysis of the second coreset construction is fully presented in Section~\ref{sec:three_coresets}, while the first and third coreset constructions involve highly technical arguments, with major details deferred to later sections. Specifically, Sections \ref{sec:VC} and \ref{sec:Doubling} establish the guarantees of the first coreset for $(r,k)$-instances with finite VC dimension or bounded doubling dimension, respectively. Finally, 
Section \ref{sec_euc} analyses the third coreset for $(r,k)$-instances in Euclidean spaces.
\section{Preliminaries}
\label{sec:pre}

Throughout this paper, let $(M,\dist)$ denote a metric space. For notational convenience, for a set $S$, we denote by $\binom{S}{k}$ the set of all $k$-element subsets of a set $S$. For a set $A\subseteq M$, the diameter of $A$ is denoted by $\mathrm{diam}(A) = \sup_{a,b\in A} \dist(a,b)$, and the distance from a point $x$ to $A$ is defined as $\dist(x,A) = \inf_{y\in A} \dist(x,y)$.

\paragraph{Ball range space and VC dimension} For a point $x\in M$ and a radius $r>0$, let $B(x,r)=\{y\in M\mid \dist(x,y)\leq r\}$ denote the ball centered at $x$ with radius $r$. The ball range space of $(M,\dist)$ is the collection of all balls, denoted by $\Balls=\{B(x,r)\mid x\in M,r>0\}$. For a $k$-point set $C\in \binom{M}{k}$, let $B(C,r)=\bigcup_{c\in C} B(c,r)$. The $k$-$\Balls$ range space of $M$ is denoted by $\Balls_k=\{B(C,r)\mid C\in \binom{M}{k},r>0\}$.

A finite set $X\subset M$ is said to be shattered by $\Balls$ if $|X\cap \Balls|=2^{|X|}$, where $X\cap \Balls=\{X\cap B(x,r)\mid B(x,r)\in \Balls\}$.
The VC dimension of $(M,\dist)$ is the maximum size of a subset of $M$ that can be shattered by $\Balls$, or $+\infty$ if no such maximum exists.

\paragraph{Doubling dimension}
We say that a metric space $(M,\dist)$ has doubling dimension at most $t$, if for any ball $B(x,r)$, there exists a set $C \in \binom{M}{q}$ with $q\leq 2^t$ so that the ball $B(x,r) \subseteq B(C,\frac{r}{2})$, i.e., any ball can be covered by at most $2^t$ balls of half radius.

\paragraph{Robust $(k,z)$-clustering} For a dataset $X\subset M$, the vanilla \kzC problem aims to compute a $k$-center set $C\in \binom{M}{k}$ that minimizes
\[
\cost_z(X,C):= \sum_{x\in X} \dist^z(x,C).
\]
This problem captures several variants, including vanilla \kMedian ($z = 1$) and vanilla \kMeans ($z = 2$).

A well-known variant, called \kzC with $m$ outliers, or the robust\kzC problem, is to find $C\in \binom{M}{k}$ that minimizes
\[
\cost^{(m)}_z(X,C):=\min_{Y\subseteq X,|X\setminus Y|\leq m} \cost_z(Y,C).
\] 
Namely, $\cost^{(m)}(X,C)$ aggregates all but the largest $m$ distances to $C$.
Similar to the vanilla case, we call the problem robust $k$-median when $z = 1$ and robust $k$-means when $z=2$.

\paragraph{Weighted Set} Let $(Y,w_Y)$ be a weighted set where each point $y\in Y$ has a weight $w_Y(y)\geq 0$. We use $w(y)$ to denote $w_Y(y)$ if the weighted set is clear in the context. Moreover, for a weighted set $Y$, we denote by $\norm{Y}_0$ the number of points in $Y$ and by $\norm{Y}_1$ the total weight of $Y$.
The $k$-median objective on $Y$ is defined as $$
\cost_z(Y,C):=\sum_{y\in Y} w(y)\dist^z(y,C).
$$

\paragraph{Weighted Outlier} Let $(Y, w_Y)$ and $(Z, w_Z)$ be two weighted sets. We call $Z$  a valid $m$-weighted outlier of $Y$ if (i) $Z\subseteq Y$, (ii) $w_Z(z)\leq w_Y(z)$ for all $z\in Z$, and (iii) $w_Z(Z) = m$. Also denote by $Y\setminus Z$ the weighted set $(Y, w_Y-w_Z)$.

For a weighted set $(Y, w_Y)$, we use $L_Y^{(m)}$ to denote the set of all valid $m$-weighted outliers of $Y$. 
The robust $(k,z)$-clustering objective on $Y$ is defined as 
$$
\cost^{(m)}_z(Y,C):=\min_{Z\in L_Y^{(m)}} \cost_z(Y\setminus Z,C).
$$

\begin{definition}[Robust coresets] \label{def:rbcoreset}
	A weighted subset $D$ of $X$ is an $(\eps,m,\Delta)$-\emph{robust coreset} of $X$ for robust $(k,z)$-clustering if for every $t=0,\dots,m$, and every $C\in \binom{M}{k}$, it holds that
	$$
	|\cost[t]_z(X,C)-\cost[t]_z(D,C)|\leq \eps\cdot\cost[t]_z(X,C)+ \Delta.
	$$
	An $(\eps,m,0)$-robust coreset is also called an $(\eps,m)$-\emph{robust coreset}.
\end{definition}

This robust coreset preserves the clustering cost for every outlier budget $t\le m$ (rather than only for $t=m$), which yields the following mergeability property.

\begin{fact}[Mergeability of Robust Coresets, \cite{huang2022near}]\label{fact_merge}
	Suppose that $X_1\cap X_2=\emptyset$. If $D_1$ is an $(\eps,m,\Delta_1)$-robust coreset for $X_1$ and $D_2$ is an $(\eps,m,\Delta_2)$-robust coreset for $X_2$, then their union $D_1\cup D_2$ forms an $(\eps,m,\Delta_1+\Delta_2)$-robust coreset for $X_1\cup X_2$.
\end{fact}

Given a real number $a \in \mathbb{R}$, we define $a^+ = \max\{a, 0\}$.
We shall repeatedly use a basic fact about the maximum of Gaussian variables.

\begin{fact}[Maximum of Gaussians, {\cite[p79]{LT91}}] \label{fact_max_gauss}
	Let $g_1,\dots,g_n$ be mean-zero Gaussian variables (which are not necessarily independent), then
	$$
	\E \max_i |g_i|\lesssim \sqrt{\log n} \sqrt{\max_i \E g_i^2}.
	$$
\end{fact}


\section{Technical overview}
\label{sec:overview}

Our approach is built upon a key conceptual idea: the \emph{\((r,k)\)-instance} (Definition~\ref{def:rk}), a structure consisting of \(k\) balls all having the same radius \(r\).  
This uniform-radius constraint is, to the best of our knowledge, a new contribution to the coreset literature, serving as an architectural foundation for the following three major methodological innovations. 
%

\begin{enumerate}[label=(\roman*),nosep,leftmargin=*]
	\item A new dataset decomposition that partitions the dataset into only \(\poly\!\log (k/\eps)\) pieces, each forming an \((r,k)\)-instance for some radius \(r\).  
    This is a significant reduction from the $\tilde{O}(k)$ pieces required by previous methods~\cite{huang2022near,huang2025coresets} and is  
    %
    essential for eliminating a $k$ factor from the final coreset size.
	
	\item A robust coreset algorithm (Algorithm \ref{alg_coreset}) for an \((r,k)\)-instance, accompanied by a novel analysis that achieves the following two advances simultaneously.
    \begin{itemize}[nosep,leftmargin=*]
        \item We introduce the first application of the chaining technique to the robust setting, overcoming its prior confinement to the vanilla case \cite{bansal2024sensitivity,Cohen25}. Our analysis crucially relies on the uniformity of the radii, which enables control over the placement of outliers. 
        \item We develop a multi-ring range-space argument that extends the prior single-ring analyses~\cite{huang2022near,huang2025coresets} by treating the rings collectively as a single \((r,k)\)-instance. 
    \end{itemize}
    
	\item Two algorithms (Algorithms \ref{alg_coreset2} and \ref{alg_coreset3}) for \((r,k)\)-instances with small radii, which employ new reductions to the vanilla coreset framework. 
	Specifically in Euclidean spaces, our joint reduction across all components of an \((r,k)\)-instance yields improved coreset sizes compared to the component-wise vanilla reduction used in~\cite{jiang2025coresets} (see Remark~\ref{remark:comparison}).

\end{enumerate}

Below, we begin by presenting these novelties in the context of Theorem \ref{thm_gen} for metric spaces with a finite VC dimension \(d_{\VC}\). 
We then discuss how to adapt these results to doubling metrics, highlight the additional technical novelty required for Euclidean metrics, and extend the techniques to general robust \kzC.

\paragraph{Review of prior results}
We begin by reviewing the techniques of \cite{huang2022near}, which obtains the first polynomial-sized robust coresets. 
Given a dataset, the authors first find a constant-factor approximation $C^\ast$ for robust \kMedian and explicitly add the $m$ outliers from this solution to the coreset. 
Then they apply the ring-decomposition framework of \cite{braverman2022power}, which builds upon \cite{Chen09}, to reduce the remaining data points (called inliers) to $\tilde{O}(k^2\eps^{-1})$ \emph{rings}, where points in each ring are approximately equidistant from the center. 
A robust coreset is then constructed for each ring using uniform sampling, and the mergeability of robust coresets (Fact~\ref{fact_merge}) implies a robust coreset for the whole dataset. 
Given a ring $X$ of radius $r$, the analysis of uniform sampling relies on the following integral representation of the cost function:
\[
    \cost[m](X,C)=\int_{0}^{\infty} (|X\setminus B(C,u)|-m)^+ \; du.
\]
Suppose that $D$ is a uniform sample (reweighted by $\frac{|X|}{|D|}$) from $X$, one immediately obtains 
\begin{equation} \label{tech:int}
|\cost[m](X,C)-\cost[m](D,C)|\leq \int_{0}^{\infty} \big|\|D\setminus B(C,u)\|_1-|X\setminus B(C,u)|\big| \; du.
\end{equation} 
Now, if the sample size on each ring is $\tilde{O}(kd_{\VC}\eps^{-2})$, it follows from the classical VC theory that $\|D\setminus B(C,u)\|_1$ approximates $|X\cap B(C,u)|$ up to an additive error $\eps\cdot |X|$ for every $C$ and $u$, i.e., $D$ is an $\eps$-range space approximation of $X$. 
Since $X$ is a ring of diameter $O(r)$, the integrand of (\ref{tech:int}) vanishes outside an interval of length $O(r)$ and, consequently, $D$ is an $(\eps,r,\eps |X| r)$-robust coreset of $X$. 
Applying this to $\tilde{O}(k^2\eps^{-1})$ rings results in a robust coreset of size $m+\tilde{O}(k^2\eps^{-1}\cdot kd_{\VC}\eps^{-2}) = m+\tilde{O}(k^3d_{\VC}\eps^{-3})$, as concluded in~\cite{huang2022near}.
A subsequent work~\cite{huang2025coresets} essentially reduces the number of ``effective rings'' to \( \tilde{O}(k) \) by adaptively selecting the sample size for each ring, achieving a smaller coreset size of \( m + \tilde{O}(k^2 d_{\VC} \eps^{-2}) \).

Our goal is to further reduce the coreset size to \( m + Q(\eps) \), where \( Q(\eps) \) denotes the size of a vanilla \( \eps \)-coreset.  
The state-of-the-art result is \( Q(\eps) = \tilde{O}(k d_{\VC} \eps^{-2}) \) \cite{Cohen25}, improving on \cite{huang2022near} by a factor of \( k \).  
This improvement is achieved by grouping the rings into \( \poly\log (k/\eps) \) collections, each containing at most \( k \) rings, and performing uniform sampling at the group level rather than for individual rings.  
A straightforward idea would be to apply the vanilla coreset construction directly to the union of all rings, i.e.\ the set of all inliers.
However, this fails to guarantee a robust coreset, since points within the rings can become outliers with respect to center sets different from \( C^\ast \). 
Therefore, new technical ideas are required. 

\paragraph{Grouping rings: $(r,k)$-instances}
A natural strategy is to leverage the core principle of vanilla coreset construction: perform uniform sampling across multiple rings and analyse the total error collectively, rather than summing the errors from individual rings.
However, a key obstacle is that the error induced by outliers depends on the chosen center set and can vary by orders of magnitude across different rings, making it difficult to control the total error uniformly over the entire space of center sets.
Addressing this motivates the key definition of an \((r,k)\)-instance (see Definition~\ref{def:rk}), which is essentially a set of points that can be covered by \(k\) balls of equal radius \(r\) (i.e.\ contained in \(B(A, r)\) for some $k$-point set $A$). 
The idea of grouping rings has been instrumental in recent advances towards constructing near-optimal vanilla coresets \cite{cohen2021new,Cohen2022Towards,cohen2022improved,huang2024optimal,Cohen25}.
These approaches typically group rings of varying radii from different clusters, whereas in our \((r,k)\)-instance, all rings must have a uniform radius \(r\).
While such uniformity is unnecessary for vanilla coreset constructions, it is crucial for robust coresets, as it ensures ``comparable'' outlier-induced errors across rings within the same group.

Now we provide a high-level overview of our algorithm, which decomposes the input dataset $X$ into a collection of $(r,k)$-instances (see Lemma \ref{thm_decomp} and Figure \ref{fig:VC_decomp} for illustration).
First, we add $m + \eps^{-1}$ points farthest from $C^\ast$ to the coreset. 
Adding these additional $\eps^{-1}$ points guarantees that all remaining points in $X$ are within a distance of $\eps\cdot\cost[m](X,C^\ast)$ from $C^\ast$. 
Then we collect all points within distance $r_{\inn}:=\frac{\eps}{mk}\cdot \cost[m](X,C^\ast)$ from $C^\ast$, forming an $(r_{\inn},k)$-instance.
Now, all remaining points have distances to $C^\ast$ within the interval $[r_{\inn}, mk r_{\inn}]$, allowing them to be decomposed into $\log (mk)$ layers of $(r,k)$-instances with geometrically increasing radii $r$. 
We then construct a robust coreset for each of these \((r,k)\)-instances (with carefully chosen error guarantees) and merge them into a final robust coreset of \(X\) using Fact \ref{fact_merge}. 
The main technical challenge is to construct robust coresets of size \(\tilde{O}(k d_{\VC} \epsilon^{-2})\) for these \((r,k)\)-instances, i.e., the layers and the inner instance; we illustrate the key ideas for each scenario below.

\paragraph{Handling a layer of $(r,k)$-instance} 
Consider a layer of \((r,k)\)-instance \(X = X_1 \cup \cdots \cup X_k \subset B(A,r)\), where each \(X_i \subset B(a_i,r)\) for some \(a_i \in A\). 
We develop Algorithm~\ref{alg_coreset} to construct an \((\epsilon,m,\epsilon \abs{X} r)\)-robust coreset for \(X\); roughly speaking, we take \(\tilde{O}(k d_{\VC} \epsilon^{-2})\) uniform samples from \(X\) (Theorem \ref{thm_main}). 
Without loss of generality, assume that all \(X_i\) have equal cost \(\cost(X_i,a_i)\). 
This type of instance is also referred to as an \((r,k)\)-regular instance (Definition \ref{def:rbcoreset}); we show in Lemma \ref{lemma_mainring} how to remove this assumption and in Figure \ref{fig:first_coreset} for illustration.

A natural idea is to extend the single-ring analysis (shown in \eqref{tech:int}) to an \((r,k)\)-instance. 
With \(\tilde{O}(k d_{\VC} \epsilon^{-2})\) points in $D$, we can still guarantee that \(\bigl|\|D \setminus B(C,u)\|_1 - |X \setminus B(C,u)|\bigr| \le \epsilon |X|\) for every \(C\) and \(u\). 
However, applying this bound directly would yield an additive error of \(\epsilon k |X|r\) (arising from integrating over \(k\) intervals of length \(O(r)\), one for each of $k$ clusters), which is a $k$-factor larger than our target error of \(\epsilon \abs{X} r\). 
Hence, additional technical insights are needed for a tighter control of the error.

\subparagraph{\em Chaining argument}
The first idea is to borrow insights from the vanilla coreset construction, where the chaining argument is central.
An essential step there is bounding the net size by \(\exp(\tilde{O}(k))\), which counts the number of distinct distance vectors \(\bigl(\dist(x,C)\bigr)_{x \in D}\) for all possible choices of \(C\), under a certain level of discretization. 
In the robust setting, however, a small perturbation to $C$ can drastically change the set of outliers and thus the distance vectors. 
A natural idea is to fix the set of outliers, which then requires taking a union bound over all $\binom{|X|}{m}$ possible sets of $m$ outliers, causing an additional $\exp(m)$ factor in the net size.
This results in a net size far exceeding the desired \(\exp(\tilde{O}(k))\), which will blow up the error bound we aim to control. 

To address this issue, we apply the chaining argument \emph{exclusively} to those clusters whose points are all inliers with respect to \(C\). 
Specifically, we partition both \(X\) and \(D\) into the same three collections of clusters with respect to \(C\): \(A_1\) (clusters of inliers), \(A_2\) (clusters of outliers), and \(A_3\) (clusters whose distances to \(C\) lie within an interval of length \(O(r)\)); see Lemma~\ref{proof_fact} and Figure \ref{fig:subset}.
Such a decomposition is possible because every cluster in an \((r,k)\)-instance has the same radius, and an additional realignment procedure (Lines 4–6 in Algorithm \ref{alg_coreset}) ensures that each cluster \(X_i\) has a total sample weight of  \(|X_i|\) (Definition \ref{def_cap}).
Notably, the realignment procedure ensures that the sample set \(D\) has the same partition as \(X\), a property absent in vanilla coreset constructions.
We then apply the chaining argument solely to \(A_1\), obtaining an \emph{indexed-subset cost approximation} that bounds the error introduced by \(A_1\) (Definition~\ref{def_subset}). 
Since \(A_1\) is a union of clusters and there are at most \(2^k\) such unions, the net size only increases by a factor of \(\exp(k)\), which remains manageable (Lemma~\ref{lemma_subset}).

\subparagraph{\em Range space argument} 
Since \(A_2\) contains only outliers, we may safely ignore it. 
It remains to control the error contributed by \(A_3\). 
Since the distances of all points in \(A_3\) from \(C\) lie within an interval of length \(O(r)\), we anticipate extending the single-ring analysis in \eqref{tech:int} to \(A_3\) and obtaining an additive error of \(O(\epsilon \lvert X\rvert r)\) (Lemma \ref{lemma_event}). 
A key challenge is that \(A_3\) varies with the choice of \(C\), increasing the number of relevant ranges to consider.
To mitigate this, we extend the notion of \emph{range space approximation} to encompass all \(2^k\) possible sets of \(A_3\) (Definition \ref{def_range}). 
While this extension increases the number of ranges by a factor of \(2^k\), we can safely ignore this factor because the original number of ranges is already on the order of \(2^{k d_{\VC}}\).
Consequently, the range space approximation guarantee of the sample set \(D\) remains valid, meaning \(D\) induces at most an \(\epsilon |X|\) error for each of these ranges. 
Now, the integral in \eqref{tech:int} is effectively over the distance interval for \(A_3\), which has length \(O(r)\), and the overall error for \(A_3\) is thus within \(\epsilon |X| \cdot O(r)\) (Lemma \ref{lemma_range}).

\paragraph{Handling inner $(r_{in},k)$-instance} 
Consider the inner \((r_{\inn},k)\)-instance \(X = X_1 \cup \cdots \cup X_k\), where each \(X_i \subseteq B(a_i, r_{\inn})\). 
We shall construct an \((\epsilon, m, mk\,r_{\inn})\)-robust coreset for \(X\) (Lemma \ref{lemma_innergroup}), where the additive error is at most $mk\,r_{\inn} \leq \epsilon \cdot \cost[m](X, C^\ast)$, which is well within acceptable bounds.

We propose Algorithm \ref{alg_coreset2} to achieve this goal. 
For each \(X_i\), we first select an arbitrary subset \(H_i\subseteq X_i\) of \(m\) points and move these points to the center \(a_i\). 
The coreset then consists of two components: a weighted point \((a_i, |H_i|)\) to represent the collection \(H_i\), and a vanilla coreset of the remaining points \(X \setminus \bigl(\bigcup_{i \in [k]} H_i\bigr)\). 
See also Figure \ref{fig:second_coreset} for illustration.
The algorithm’s correctness follows from the fact that the total moving distance between \(X\) and \(D\) is at most \(m k\,r_{\inn} \), which ensures the same level of additive error.

Interestingly, this robust coreset construction for the inner \((r_{\inn},k)\)-instance also leads to another coreset of size \(O(km \epsilon^{-1}) + \tilde{O}(k d_{\VC} \epsilon^{-2})\) for the full dataset (rather than just for the inner instance). 
The algorithm maintains this robust coreset \(D\) along with all points outside the inner \((r_{\inn},k)\)-instance.
Since each outside point \(x\) satisfies \(\dist(x, C^\ast) \geq r_{\inn} = \frac{\epsilon}{mk} \cdot \cost[m](X, C^\ast)\), there can be at most \(km \epsilon^{-1}\) such points, yielding the stated coreset size. 
We note that this matches the size obtained in a recent paper \cite{jiang2025coresets}, although we use different reductions to the vanilla setting. 
Unlike our requirement that each ball in the inner \((r_{\inn},k)\)-instance has a small radius, their approach requires every ball to be \emph{dense}, containing at least \(m \epsilon^{-1}\) points.

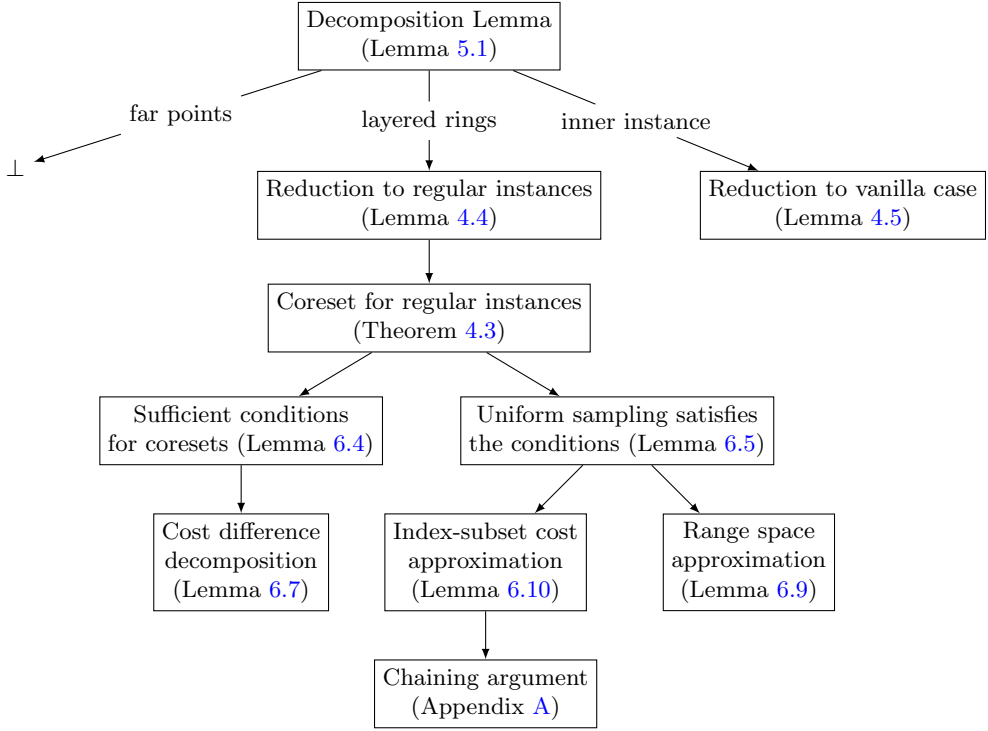
\begin{figure}[t]
\centering
\footnotesize
\begin{tikzpicture}[every text node part/.style={align=center}]
\node[draw] (root) at (2,3.75) {Decomposition Lemma\\ (\Cref{thm_decomp})};
\node[draw] (ringroot) at (2,1.5) {Reduction to regular instances\\ (\Cref{lemma_mainring})};
\node[draw] (innerroot) at (7.5,1.5) {Reduction to vanilla case\\ (\Cref{lemma_innergroup})};
\node[draw] (ringregular) at (2,0) {Coreset for regular instances\\ (\Cref{thm_main})};
\node[draw] (ringsufficient) at (-0.5,-1.5) {Sufficient conditions\\ for coresets (\Cref{lemma_main_tech})};
\node[draw] (costdiff) at (-0.5,-3.25) {Cost difference\\ decomposition \\ (\Cref{lemma_inq})};
\node[draw] (sufficient) at (4.5,-1.5) {Uniform sampling satisfies\\ the conditions (\Cref{alg_guarantee})};
\node[draw] (indexsubset) at (2.75,-3.25) {Index-subset cost \\ approximation \\ (\Cref{lemma_subset})};
\node[draw] (chainingroot) at (2.75,-5) {Chaining argument \\ (\Cref{sec_proof_subset_appendix})};
\node[draw] (rangespace) at (6.25,-3.25) {Range space \\ approximation \\ (\Cref{lemma_range})};
\node (far) at (-3.5,2) {$\perp$};
\node (v6) at ($(root)!.6!(far)$) {far points};
\node (v7) at ($(root)!.5!(ringroot)$) {layered rings};
\node (v8) at ($(root)!.5!(innerroot)$) {inner instance};
\draw [-latex] (root) -- (v7) -- (ringroot);
\draw [-latex] (ringroot) -- (ringregular);
\draw [-latex] (ringregular) -- (ringsufficient);
\draw [-latex] (ringregular) -- (sufficient);
\draw [-latex] (root) -- (v8) -- (innerroot);
\draw [-latex] (root) -- (v6) -- (far);
\draw [-latex] (sufficient) -- (indexsubset);
\draw [-latex] (sufficient) -- (rangespace);
\draw [-latex] (indexsubset) -- (chainingroot);
\draw [-latex] (ringsufficient) -- (costdiff);
\end{tikzpicture}
\caption{Structural diagram of the proof for VC instances (Theorem \ref{thm_gen})}\label{fig:flowchart}
\end{figure}
This concludes our proof overview for the VC case. A structural diagram is shown in Figure~\ref{fig:flowchart}. %

\paragraph{Adapting to doubling metrics}
The analysis for the doubling metric in Theorem \ref{thm_doubling} is almost identical to that for the VC case.
We employ the same point set decomposition into $(r,k)$-instances (Lemma \ref{thm_decomp}) and the same algorithms (Algorithms \ref{alg_coreset} and \ref{alg_coreset2}) for constructing robust coresets,
with the only difference being the substitution of $d_D$ for $d_{\VC}$ in the coreset size.
%
The key difference is that, in a doubling space, the VC dimension of the set of balls can be unbounded \cite{huang2018epsilon}, preventing a direct application of the same range space argument used in the VC case. 
To overcome this, we adopt the technique from \cite{huang2018epsilon}, which shows that by (randomly) distorting the distance function \(\dist\) slightly (Definition \ref{def:eps_smooth}), one can bound a ``probabilistic'' version of the VC dimension by \(\tilde{O}(d_D)\) (Lemma \ref{lemma:randomdist}). 
This property enables a distorted variant of the range space argument (Lemma \ref{lemma:eps_range}), and we further show that the resulting error for \(A_3\) under this distortion remains well-controlled (Lemma \ref{lemma:diffrerence:doubling}).

\paragraph{Euclidean spaces}
As discussed above, our algorithm 
for the inner $(r_{\inn}, k)$-instance (Algorithm~\ref{alg_coreset2}) also produces a Euclidean coreset of size $O(km \eps^{-1}) + Q(\eps)$, where $Q(\eps)$ now denotes the size of the vanilla $\eps$-coreset in Euclidean space.
Note that the total number of moved points in Algorithm~\ref{alg_coreset2} is $m k$, while the number of outliers is only $m$, wasting a $ k$ factor. 
This observation motivates Algorithm~\ref{alg_coreset3}, which further reduces the Euclidean coreset size to $O(m \eps^{-1}) + Q(\eps)$.
The key idea is to increase the radius threshold $r_{\inn}$ from $\frac{\epsilon}{mk} \cdot \cost[m](X, C^\ast)$ to $\frac{\epsilon}{m} \cdot \cost[m](X, C^\ast)$.
See Figure \ref{fig:euc_decomp} for illustration.

In the analysis of this algorithm, a major challenge arises when the inner \((r_{\inn},k)\)-instance contains points that are outliers with respect to some center set \(C\). 
As noted earlier, a na\"ive adaptation of the vanilla chaining argument can blow up the net size by a factor of \(\exp(m)\). 
To address this, we propose a novel decomposition of the robust clustering cost by strategically reassigning outliers across clusters, ensuring that only one cluster contains both inliers and outliers, while all other clusters contain exclusively inliers or exclusively outliers (Lemma~\ref{lemma_Euc_tech}).
This decomposition relies on the uniform radius assumption, which guarantees that moving each outlier introduces only a small additive error $O(r_{\inn})$ to the overall clustering cost.
Consequently, the net size in the chaining argument grows by at most a factor of \(2^k\), which remains manageable as before.

\paragraph{Extension to general $z\geq 1$}
To extend our approach to general robust \kzC, the key idea is to use the generalized triangle inequality (Lemma \ref{lm:triangle}). 
In both the VC and doubling cases, we further require the sample set \(D\) to be an \(\epsilon^z\)-range space approximation, rather than \(\epsilon\) as in the robust \kMedian\ case. This change results in an \(\epsilon^{-2z}\) term in the coreset size (Theorem \ref{thm:kzC}). 
We note that if the range space argument could be further refined, it would be possible to achieve the desired coreset size of \(m + Q(\eps)\) for any \(z \ge 1\) (Remark \ref{remark:open}), which remains an intriguing open problem.
\section{Three Coreset Constructions for $(r,k)$-Instances}
\label{sec:three_coresets}

A common approach to constructing a robust coreset for a point set is to decompose the set into distinct parts with useful geometric properties and then construct a coreset for each part separately.  
This paper follows this strategy by introducing the notion of \( (r,k) \)-instances and \( (r,k) \)-regular instances.  
%

\begin{definition}[$(r,k)$-instance] \label{def:rk}
A dataset $X\subset M$ is called an $(r,k)$-instance if
$X\subset B(A,r)$ for some $A\in \binom{M}{k}$.
\end{definition}

Intuitively, an \( (r,k) \)-instance \( X \) consists of \( k \) clusters, each contained within a ball of radius \( r \).  
As a result, we have \( \cost(X,A) \leq |X| \cdot r \).  

We further define a specific subclass of \( (r,k) \)-instances below.

\begin{definition}[$(r,k)$-regular instance] A point set $X=X_1\cup\cdots\cup X_{k}\subset M$ is called an $(r,k)$-\emph{regular instance} if 
	\begin{enumerate}[label=(\roman*),itemsep=0pt]
		\item $\forall i,j\in [k]$, $i\neq j$, $X_i\cap X_j=\emptyset$;
		\item $\forall i\in [k]$, $X_i\subseteq B(a_i,r)$ for some center $a_i\in M$;
		\item $\forall i,j\in [k], i\neq j$, $\abs{X_i} \leq 2\abs{X_j}$.
	\end{enumerate} 
\end{definition}

An \( (r,k) \)-regular instance further imposes the condition that the number of points in each cluster \( X_i \) is of the same order.  
Specifically, we have \( \frac{|X|}{2k} \leq |X_i| \leq \frac{2|X|}{k} \) for all \( i \in [k] \).  
We note that the prior work \cite{cohen2021new,Cohen25} introduced a similar structure called a \emph{group}, which also consists of \( k \) clusters.
The key distinction between our data structure and a group is that we enforce strict alignment on the radius of different clusters, a property that is crucial for our analysis.

In the three subsections below, we shall give three different coreset constructions for $(r,k)$-instances. 

\subsection{First coreset construction}
\label{sec_mainring}

The first coreset construction applies to both VC-dimension-based and doubling-dimension-based instances.  
As a preliminary step, we present a robust coreset construction for \( (r,k_0) \)-regular instances (\( k_0 \leq k \)) in Algorithm~\ref{alg_coreset}, with its theoretical guarantee stated in Theorem~\ref{thm_main}.
Algorithm~\ref{alg_coreset} takes a uniform sample $U$ from $X$ (Line 2) and then constructs a coreset $D$ by assigning weights to points in $U$ (Lines 3-7).

\begin{algorithm}[t]
	\caption{$\mathrm{Coreset1}(X,r,k_0,k,\eps, m)$}
	\label{alg_coreset}
	\begin{algorithmic}[1]
		\REQUIRE An $(r,k_0)$-regular instance $X=X_1\cup\cdots\cup X_k$, parameters $\eps$, $m$
		\ENSURE Coreset $D$
		\STATE $d\leftarrow \min\{d_{\mathrm{VC}}, d_D\}$, $s \gets \frac{kd}{\eps^2}\cdot \log^{O(1)} (\frac{kd}{\epsilon})\cdot \log^3 m$
		\STATE $U\gets $ a uniform sample with replacement of size $s$ from $X$, with each sampled point weighted by $\frac{n}{s}$
		\STATE $D\gets \emptyset$
		\FOR{$i=1,\dots,k_0$}
		\STATE $u_i\gets \norm{U\cap X_i}_1$
		\STATE for each $x\in U\cap X_i$, add $x$ to $D$ with weight $w_D(x)=\frac{n_i}{u_i}\cdot \frac{n}{s}$
		\ENDFOR
		\RETURN $D$
	\end{algorithmic}
\end{algorithm}

\begin{theorem}[VC and doubling $(r,k_0)$-regular instances]
\label{thm_main} 
Let $(M,\dist)$ be a metric space with VC dimension $d_{\mathrm{VC}}$ and doubling dimension $d_D$.
Let $d = \min\{d_{\mathrm{VC}}, d_D\}$.
Assume that $X\subseteq M$ is an $(r,k_0)$-regular instance with $|X|=n$.
With probability $1-O\big(\frac{1}{\log(k/\eps)\cdot \log (mk)}\big)$, Algorithm~\ref{alg_coreset} returns an $(\eps,m,\eps n r)$-robust coreset of $X$ for robust \kMedian with size $kd\eps^{-2}\cdot \log^{O(1)}(kd\eps^{-1})\cdot \log^3 m$.
\end{theorem}

Theorem~\ref{thm_main} is a major technical result and its proof is deferred to Section~\ref{sec:VC}. 
Based on the theorem, we can now prove our first robust coreset construction for $(r,k)$-instances by splitting an $(r,k)$-instance into $\log \frac{k}{\eps}$  regular instances. An illustration is given in Figure~\ref{fig:first_coreset}.

\begin{figure}[t]
    \centering
    \includegraphics[width=0.8\linewidth]{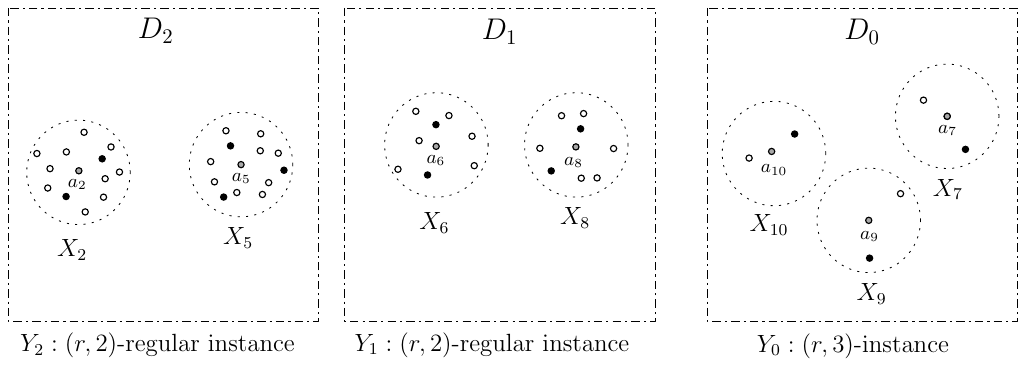}
    \vspace{-0.3cm}
    \caption{An illustrated example of the first coreset construction as per Lemma~\ref{lemma_mainring}. Black dots represent coreset points. Clusters are grouped based on their sizes.}
    \label{fig:first_coreset}
\end{figure}

\begin{lemma}[VC and doubling $(r,k)$-instances] 
\label{lemma_mainring}
Let $(M,\dist)$ be a metric space with VC dimension $d_{\mathrm{VC}}$ and doubling dimension $d_D$.
Let $d = \min\{d_{\mathrm{VC}}, d_D\}$.
There is a randomized algorithm which, given an $(r,k)$-instance $X$, computes with probability at least $1- O(\frac{1}{\log (mk)})
	$  an $(\eps,m,\eps r |X|)$-robust coreset with size $\tilde{O}(kd\eps^{-2})\cdot \log^3 m$ of $X$ for robust \kMedian.
\end{lemma}
\begin{proof}
	Let $X=X_1\cup\cdots\cup X_k$ be an $(r,k)$-instance where $X_i\subset B(a_i,r)$ for some $a_i\in M$, and suppose that $|X|=n$. We let $P_0=\{i\in [k] \mid 0 < |X_i|\leq \frac{\eps n}{k}\}$ and $P_j=\{i\in [k]\mid 2^{j-1}\cdot \frac{\eps n}{k}<|X_i|\leq 2^j\cdot \frac{\eps n}{k}\}$ for $j=1,\dots,\lceil\log \frac{k}{\eps}\rceil$. Observe that $P_j$'s form a partition of $[k]$.
	
	Let $Y_j=\bigcup_{i\in P_j} X_i$ for each $j=0,\dots,\lceil\log \frac{k}{\eps}\rceil$.
	
	Now we construct our coreset $D$. Observe that $Y_j$ is an $(r,|P_j|)$-regular instance for each $j\geq 1$. Applying Theorem~\ref{thm_main} to $Y_j$ yields an $(\eps,m,\eps r |Y_j|)$-robust coreset $D_j$ for robust \kMedian with probability at least $1-O\big(\frac{1}{\log \frac{k}{\eps}\cdot \log (mk)}\big)$, with $|D_j| = \tilde{O}(\frac{kd}{\eps^2})\cdot \log^3 m$. We include all these coresets $D_j$ in $D$.
	
	Next, we handle $Y_0$. For each $i\in P_0$, select an arbitrary point $x_i\in X_i$ and define $D_0=\{(x_i,|X_i|)\mid i\in P_0 \}$. By the triangle inequality, we know that $D_0$ is an $(\eps,m,2\eps n r)$-coreset of $Y_0$. We add $D_0$ to $D$. 
	
	By the mergeability of robust coresets, we conclude that $D$ is an $(\eps, m, 3\eps n r)$-robust coreset of $X$ for robust \kMedian. Furthermore, the size of $D$ is
	$$
	\norm{D}_0 = k+ \left\lceil\log \frac{k}{\eps}\right\rceil\cdot \tilde{O}\left(\frac{kd}{\eps^{2}}\right)\cdot \log^3 m=\tilde{O}\left(\frac{kd}{\eps^{2}}\right)\cdot \log^3 m.
	$$
	The overall failure probability is at most $\lceil\log\frac{k}{\eps}\rceil\cdot O\big(\frac{1}{\log \frac{k}{\eps}\cdot \log (mk)}\big) = O(\frac{1}{\log (mk)})$. Rescaling $\eps$ completes the proof.
\end{proof}

\subsection{Second coreset construction}
\label{sec_innergroup}

Suppose that $X$ is an $(r,k)$-instance and $X\subset B(A,r)$ for some $A=\{a_1,\dots,a_k\}\in \binom{M}{k}$. 
Let $X_i=\{x\in X\mid i=\mathrm{argmin}_{j\in [k]} \dist(x,a_j)\}$, where ties in $\mathrm{argmin}$ are broken arbitrarily. 
We present in Algorithm~\ref{alg_coreset2} our second coreset construction, which includes sufficient many points $H$ as possible outlier surrogates (Lines 1--6) and a vanilla $\eps$-coreset $D'$ on the remaining points (Lines 7--8). An illustration is given in Figure~\ref{fig:second_coreset}.

\begin{figure}
    \centering
    \includegraphics[width=0.8\linewidth]{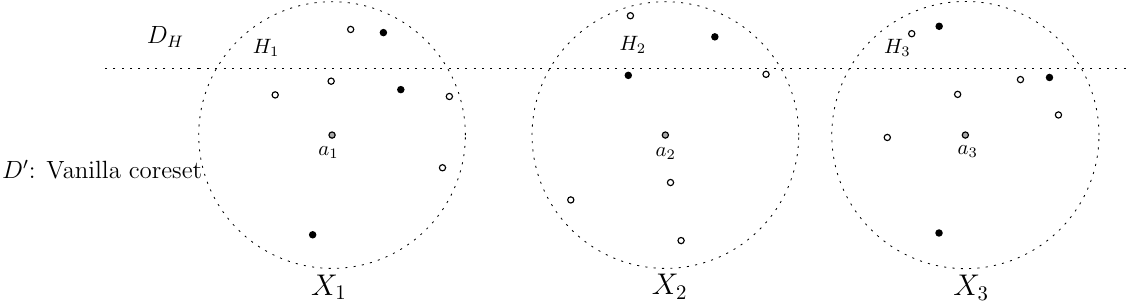}
    \vspace*{-0.3cm}
    \caption{An illustrated example of the second coreset construction as per Lemma~\ref{lemma_innergroup}. Black dots represent coreset points.}
    \label{fig:second_coreset}
\end{figure}

\begin{algorithm}[t]
	\caption{$\mathrm{Coreset}2(X,r,k,\eps, m)$}
	\label{alg_coreset2}
	\begin{algorithmic}[1]
		\REQUIRE An $(r,k)$-instance $X=X_1\cup\cdots\cup X_k$, parameters $\eps$, $m$, and a vanilla $\eps$-coreset oracle $\mathcal{B}$
		\ENSURE Coreset $D$
		\FOR{$i=1,\dots,k$}
		\STATE
		$H_i\leftarrow$ a subset of $\min(m,|X_i|)$ arbitrary points from $X_i$
		\STATE 
		$x_i\leftarrow$ an arbitrary point from $H_i$
		\STATE add $x_i$ with weight $w_D(x_i)=|H_i|$ to $D$
		\ENDFOR
		\STATE
		$H\leftarrow \bigcup_{i=1}^k H_i$
		\STATE
		$X'\leftarrow X\setminus H$
		\STATE
		$D'\leftarrow \mathcal{B}(X',\eps)$
		\STATE
		add $D'$ to $D$
		\RETURN $D$
	\end{algorithmic}
\end{algorithm}

\begin{lemma}[Reduction from vanilla coreset to robust coreset] 
\label{lemma_innergroup} Suppose there exists a randomized algorithm $\mathcal{B}$ that on every \kMedian instance, with probability at least $0.9$, computes a vanilla $\eps$-coreset with size $Q(\eps)$.
Given an $(r,k)$-instance $X$, with probability at least $0.9$, Algorithm~\ref{alg_coreset2} computes an $(\eps,m,mkr)$-robust coreset with size $k+Q(\eps)$ of $X$ for robust \kMedian.
\end{lemma}
\begin{proof} 
	
Let $D_H=\{(x_i,|H_i|)\mid i\in [k]\}$ denote the set of $k$ weighted points added into $D$ in the for loop. 
We prove that $D=D_H\cup D'$ is an $(\eps,m,O(mkr))$-robust coreset of $X=H\cup X'$.
	

Let $J^\ast$ denote the set of the outlier points in $\cost[t](X, C)$.
For each $i\in [k]$, let $t_i = |J^\ast \cap X_i|$. We then have $t_i \leq \min\{\abs{X_i}, m\} = \abs{H_i}$. Define $E^\ast_i$ as the set of outlier points in $\cost[t_i](D\cap X_i, C)$. Since $\norm{D\cap X_i}_1\geq \abs{H_i}\geq t_i$, it must hold that $\norm{E^\ast_i}_1 = t_i$. Now, let $E^\ast = \bigcup_{i\in [k]} E_i^\ast$. Noting that $t = \sum_{i\in [k]} t_i$, we obtain 
\[
\cost(D\setminus E^\ast, C) = \sum_{i\in [k]} \cost[t_i](D\cap X_i, C) \geq \cost[t](D, C).
\]
Since $t_i \leq \abs{H_i}$, we can find a subset $J'\subseteq H\setminus J^\ast$ such that $\abs{J'\cap H_i} = \abs{(J^\ast\setminus H)\cap X_i}$ for all $i\in [k]$. Define $J = (J^\ast\cap H)\cup J'$. 
Note that $|J\cap H_i| = |J^\ast\cap X_i| = t_i$ for all $i\in [k]$.
Similarly, define $E = (E^\ast \cap D_H)\cup E'$, where $E'\subseteq D_H\setminus E^\ast$ is chosen such that $\norm{E'\cap \{x_i\}}_1 = \norm{E_i^\ast\setminus D_H}_1$. Thus, we have $\norm{E\cap \{x_i\}}_1 = \norm{E^\ast_i}_1 = t_i$.
It follows that
	\begin{align*}
		&\quad\ \cost[t](D,C) - \cost[t](X,C) \\
		&\leq \cost(D\setminus E^\ast,C) - \cost(X\setminus J^\ast,C) \\
		&\leq \cost(D\setminus E,C) - \cost(X\setminus J,C)  + \abs{ \cost(E,C) - \cost(E^\ast,C) } + \abs{ \cost(J,C) - \cost(J^\ast,C) }\\
		&=  \cost(D\setminus E,C) - \cost(X\setminus J,C)  + \abs{ \cost(E',C) - \cost(E^\ast\setminus D_H,C) } \\
		&\hspace*{8cm} + \abs{ \cost(J',C) - \cost(J^\ast\setminus H,C) } \\
		&\leq \abs{ \cost(D\setminus E,C) - \cost(X\setminus J,C) } + 2mr + 2mr\\
		&= \abs{\cost(D_H\setminus E, C) - \cost(H\setminus J,C)} + \abs{\cost(D', C) - \cost(X',C)} + 4mr.
	\end{align*}
	We shall bound the first two terms separately.
	
	For the first term, by the triangle inequality,
	\begin{align*}
	\abs{\cost(D_H\setminus E, C) - \cost(H\setminus J, C)}
	\leq & \ \sum_{i\in [k]} \abs{(|H_i|-t_i)\dist(x_i, C) - \cost(H_i\setminus J, C)} \\
    \leq & \ \sum_{i\in [k]} \sum_{x\in H_i\setminus J} \dist(x,x_i)
	\leq 2mkr.
	\end{align*}
	where the last inequality follows from the fact that $\sum_i \abs{H_i} \leq mk$ due to the construction of $H$.
	
	For the second term, since $D'$ is a vanilla $\eps$-coreset of $X'$, we have
	\begin{align*}
		\abs{\cost(D', C) - \cost(X',C)} 
		&\leq \eps \cost(X', C) \\
		&\leq \eps \cost(X\setminus J, C) \\
		&\leq \eps \cost(X\setminus J^\ast, C) + \eps\abs{ \cost(J, C) - \cost(J^\ast, C) } \\
		&\leq \eps \cost(X\setminus J^\ast, C) + 2\eps m r \\
		&=\eps \cost[t](X,C) + 2\eps m r.
	\end{align*}
	Therefore,
    \begin{equation}\label{eqn:second_coreset_bound_1}
	\begin{aligned}
		\cost[t](D,C) - \cost[t](X,C)  
		&\leq \eps \cost^{(t)}(X, C) + (2mk+4m+2\eps m) r\\
		&\leq \eps\cost^{(t)}(X,C)+8mkr.
	\end{aligned}
    \end{equation}
    By a similar argument, we can also obtain that
    \[
        \cost[t](X,C) - \cost[t](D,C) \leq \eps\cost^{(t)}(D,C) + 8mkr,
    \]
    which implies that
    \begin{equation}\label{eqn:second_coreset_bound_2}
        \cost[t](X,C) - \cost[t](D,C) \leq \frac{\eps}{1+\eps}\cost^{(t)}(X,C) + \frac{8mkr}{1+\eps}.
    \end{equation}
	Combining \eqref{eqn:second_coreset_bound_1} and \eqref{eqn:second_coreset_bound_2}, we see that $D$ is an $(\eps,m,8mkr)$-robust coreset of $X$.
\end{proof}

\subsection{Third coreset construction}

Let $X$ be an $(r,k)$-instance in $\mathbb{R}^d$ and $X\subset B(A,r)$. Let $A=\{a_1,\dots,a_k\}$ and $X_i=\{x\in X\mid i = \argmin_{j\in [k]}\|x-a_j\|_2\}$, where the ties are broken arbitrarily, so $X_i\subset B(a_i,r)$.
For each $x\in X$, define $\pi(x)$ to be the unique index $i\in [k]$ such that $x\in X_i$. For each $i\in [k]$, let $\Delta_i := \frac{\cost(X_i,A)}{|X_i|}$ denote the average cost of $X_i$.

Our third coreset construction is for Euclidean spaces, given in Algorithm~\ref{alg_coreset3}. 
It is a modification of an importance sampling algorithm proposed by \cite{bansal2024sensitivity}. 
The only change occurs in Line~\ref{alg3_scale}, where we scale the weights to ensure that $\norm{D\cap X_i}_1 = \abs{X_i}$ for every $i\in [k]$ (this is called capacity-respecting, see Definition~\ref{def_cap}).

\begin{figure}[t]
    \centering
    \includegraphics[width=0.8\linewidth]{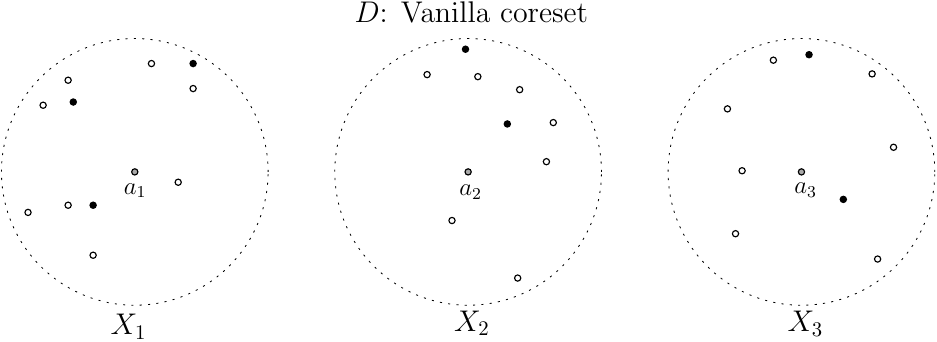}
    \vspace*{-0.3cm}
    \caption{An illustrated example of the third coreset construction as per Lemma~\ref{lemma_Euc}. The construction of $D$ includes an additional capacity-respecting step (Lines 6-10 in Algorithm \ref{alg_coreset3}) compared to prior vanilla construction, which is essential for handling robust clustering cost.}
    \label{fig:third_coreset}
\end{figure}

\begin{algorithm}[t]
	\caption{$\mathrm{Coreset}3(X,A,r,k,\eps)$}
	\label{alg_coreset3}
	\begin{algorithmic}[1]
		\REQUIRE An $(r,k)$-instance $X$, $A\in \binom{\mathbb{R}^d}{k}$,  and $\eps$.
		\ENSURE Coreset $D$
		\STATE $s\gets \min(k^{4/3} \eps^{-2}, k \eps^{-3})\cdot \log^{O(1)}(k/\eps)$ \label{alg_coreset3_size}
		\FOR{$x\in X$}
		\STATE $p(x)\gets \frac{1}{4}\big(\frac{\cost(x,A)+\Delta_{\pi(x)}}{\cost(X,A)}+\frac{1}{k|X_{\pi(x)}|}+\frac{\cost(x,A)}{k\cost(X_{\pi(x)},A)}\big)$
		\ENDFOR
		\STATE
		Draw with replacement a sample $U$ of $s$ points from $X$ where $x\in X$ is drawn w.p. $p(x)$, and weighted by $w_U(x)=\frac{1}{s\cdot p(x)}$ if drawn
		\FOR{$i=1,\dots,k$}
		\STATE $u_i\gets \|U\cap X_i\|_1$
		\ENDFOR
		\FOR{$x\in U$}
		\STATE add $x$ into $D$ with weight $w_D(x)=\frac{|X_{\pi(x)}|}{u_{\pi(x)}}\cdot w_U(x)$ 
		\label{alg3_scale}
		\ENDFOR
		\RETURN $D$
	\end{algorithmic}
\end{algorithm}

\begin{lemma}[Euclidean $(r,k)$-instances] 
\label{lemma_Euc}
    Given an $(r,k)$-instance $X$ in $\mathbb{R}^d$ such that $X\subset B(A,r)$ with $A\in \binom{\mathbb{R}^d}{k}$, Algorithm~\ref{alg_coreset3} computes in $O(nkd)$ time, with probability at least $0.9$, an $(\eps,m,6mr+\eps\cost(X,A))$-robust coreset of $X$, with size $\tilde{O}(\min(k^{4/3} \eps^{-2},k \eps^{-3}))$. 
\end{lemma}

This lemma reduces the use of a vanilla coreset to a robust coreset for Euclidean instances.
Its proof is highly technical and postponed to Section~\ref{sec_euc}.

\section{Proofs of Main Theorems: Point Set Decompositions}
\label{sec:proof_thm}

We prove Theorems \ref{thm_gen}, \ref{thm_doubling} and \ref{thm_Euc} for $z = 1$ via three coreset constructions in Section \ref{sec:three_coresets}.

\subsection{Proof of Theorems \ref{thm_gen} and \ref{thm_doubling}: VC instances and doubling instances}

Let $(M,\dist)$ be a metric space with VC dimension $d_{\mathrm{VC}}$ and doubling dimension $d_D$.
Let $d = \min\{d_\VC, d_D\}$.
The following lemma provides a decomposition of a dataset $X\subseteq M$, which is essential for coreset construction. An illustration is given in Figure~\ref{fig:VC_decomp}.

\begin{figure}
    \centering
    \includegraphics[width=\linewidth]{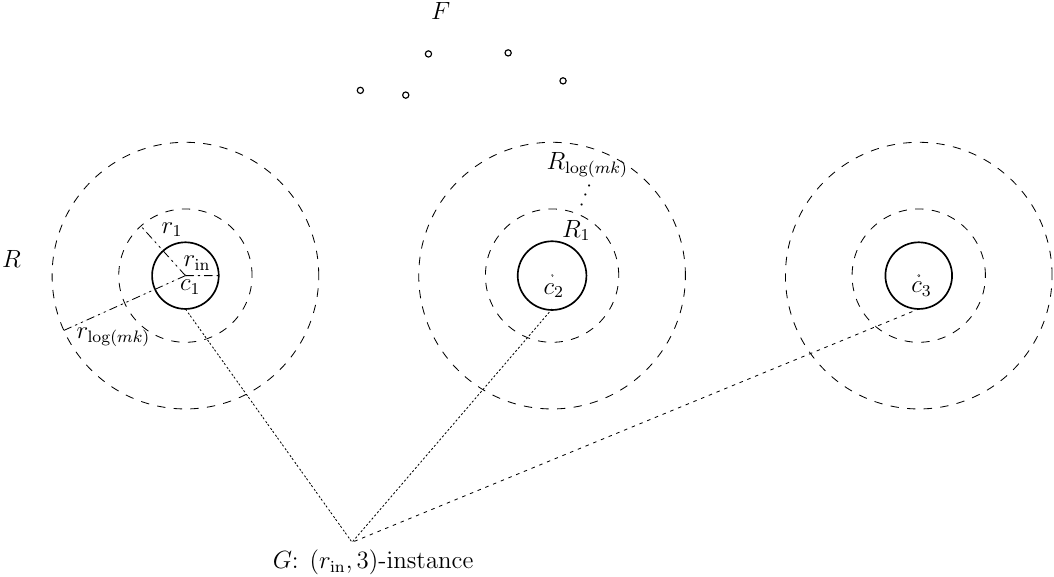}
    \vspace*{-0.7cm}
    \caption{An illustrated example of the point set decomposition for VC and doubling instances in Lemma~\ref{thm_decomp}. Here, 
    $r_{\inn} = \eps/(mk)\cdot \cost^{(m)}(X,C^\ast)$.}
    \label{fig:VC_decomp}
\end{figure}

\begin{lemma}[First point set decomposition] 
\label{thm_decomp}
	There is an $O(nk)$ time algorithm that, given any constant-factor approximation $C^\ast$ for robust $k$-median on $X$, decomposes $X$ as
	\[
		X=F\cup R\cup G,
	\]
	where
	\begin{itemize}[itemsep=0pt,topsep=0pt]
	\item $F$ is a finite subset of $X$ with $|F| = O(m+\eps^{-1})$, 
	\item $R$ is the union $R=R_1\cup\cdots\cup R_l$ with $l=O(\log (mk))$, each $R_i$ being an $(r_i,k)$-instance satisfying $\sum_{i\in [l]} r_i|R_i| = O(\cost^{(m)}(X,C^\ast))$, 
	\item $G$ is an $(r_{\inn},k)$-instance, where $r_{\inn} = \eps/(mk)\cdot \cost^{(m)}(X,C^\ast)$.
	\end{itemize}
\end{lemma}

\begin{proof}
	Let $L$ be the set of the $m+\lceil \eps^{-1} \rceil$ furthest points to $C^\ast$ in $X$ and define $Y=X\setminus L$. Let $r = \cost(Y,C^\ast)/|Y|$ denote the average cost of $Y$. For each $i\in[k]$, let $Y_i = \{y\in Y\mid i=\mathrm{argmin}_{i\in [k]} \dist(y,c_i^\ast)\}$ denote the set of inliers whose closest center is $c_i^\ast$, where ties in $\mathrm{argmin}$ are broken arbitrarily. Thus, $Y_1,\dots,Y_k$ form a partition of $Y$.
	
	\emph{Balls} around $C^\ast$ are defined as $B_i=\{y\in Y\mid \dist(y,C^\ast)\leq 2^i\cdot\eps r\}$. We then define the \emph{rings} such that $R_0=B_0$ and $R_i=B_i\setminus B_{i-1}$ for $i=1,2,\dots,T$, where $T$ is the largest integer such that $R_T\neq\emptyset$. This implies that $T\leq \log(|Y|/\eps)$ and $R_0,\dots,R_T$ form a partition of $Y$. 
	
	We are ready to define the decomposition. Specifically, we let
	
	\begin{itemize}[itemsep=0pt,topsep=0pt]
		\item $F = L$;
		\item $R = \bigcup_{j=T-s+1}^T R_j$ for $s=\min\{\lceil 1 + \log (mk) \rceil,T\}$;
		\item $G = \bigcup_{j=0}^{T-s} R_j$. 
	\end{itemize}
	
	We have $|F|= O(m+\eps^{-1})$ by definition. For each $j\in \{T-s+1,\dots,T\}$, $R_j$ is contained in the ball $B_j$. So $R_j$ is an $(r_j ,k)$-instance with $r_j=2^j \cdot \eps r$. So we have,
	$$
	\sum_{j=T-s+1}^T r_j|R_j|=\sum_{j=T-s+1}^T 2^j |R_j| \eps r\leq \sum_{j=T-s+1}^T 2\cost(R_j,C^\ast)\leq 2\cost(Y,C^\ast)\leq 2\cost^{(m)}(X,C^\ast).
	$$
	
	It remains to prove that $G$ is an $(\frac{\eps}{mk}\cdot \cost^{(m)}(X,C^\ast),k)$-instance. Let $G_i=G\cap Y_i$ for each $i\in [k]$. $G_i$'s are disjoint since $Y_i$'s are disjoint. It suffices to prove $\mathrm{diam}(G_i)\leq \frac{\eps}{mk}\cdot \cost^{(m)}(X,C^\ast)$ for each $i\in [k]$.
	
	Recall that $X=Y\cup L$ and $|L|=m+\lceil \eps^{-1}\rceil$, meaning that $\cost[m](X,C^\ast)$ aggregates at least $\lceil \eps^{-1}\rceil$ points outside the ball $B_T$. This implies that
	\[
	\eps\cdot \cost[m](X,C^\ast)\geq  \frac{\cost[m](X,C^\ast)}{\lceil \eps^{-1}\rceil}\geq 2^T\cdot \eps r.
	\]
	Hence, for each $i\in [k]$,
	\[
	\diam(G_i) \leq 2\cdot 2^{T-s}\cdot \eps r\leq 2\cdot \frac{1}{2mk}\cdot 2^T\cdot \eps r\leq \frac{\eps}{mk}\cdot \cost[m](X,C^\ast). \qedhere
	\]

\end{proof}

We are ready to prove Theorems \ref{thm_gen} and \ref{thm_doubling}.

\begin{proof}[Proof of Theorems \ref{thm_gen} and \ref{thm_doubling}]
We first find a constant-factor approximation $C^\ast$ of $X$, i.e.,\linebreak $\cost^{(m)}(X,C^\ast) \lesssim \min_{C\in \binom{M}{k}} \cost^{(m)}(X,C)$.\footnote{To improve the runtime, we can instead construct a tri-criteria approximation solution $C^\ast$ in near-linear time~\cite{bhaskara2019greedy}, which relaxes the outlier and center number constraints by constant factors. This construction only affects the coreset size by a constant factor, and hence, we usually assume $C^\ast$ to be a constant-factor approximation for simplicity. A more detailed discussion can be found in, e.g., Appendix A in \cite{huang2025coresets}.}
Apply Lemma~\ref{thm_decomp} to decompose $X$ into $X=F\cup R\cup G$, where $R=R_1\cup\cdots\cup R_l$ for $l=O(\log mk)$ and each $R_i$ is an $(r_i,k)$-instance, and $G$ is an $(\frac{\eps}{km}\cost^{(m)}(X,C^\ast),k)$-instance. We construct our coreset $D$ as follows.
	
	\begin{itemize}[itemsep=0pt]
		\item We add $F$ identically into $D$, i.e.\ each $x\in F$ is added into $D$ with unit weight.
		\item For each $R_i$ ($i\in [l]$), we apply Lemma~\ref{lemma_mainring} to construct an $(\eps,m,\eps r_i |R_i|)$-robust coreset $D_i$ of $R_i$. We add all $D_i$'s into $D$.
		\item For $G$, we remark that $G$ is an $(\frac{\eps}{km}\cost^{(m)}(X,C^\ast),k)$-instance, so we can apply Lemma~\ref{lemma_innergroup} to construct an $(\eps,m,\eps\cdot \cost^{(m)}(X,C^\ast))$-coreset $D_{\inn}$ of $G$. We add $D_{\inn}$ into $G$.
	\end{itemize}
	
	We first bound the size of $D$. Note that $\|D\|_0=\|F\|_0+\sum_{i=1}^l \|R_i\|_0+\|D_{\inn}\|_0$. We have $\|F\|_0 = O(m+\eps^{-1})$ by Lemma~\ref{thm_decomp}, $\|R_i\|_0 = \tilde{O}(kd\eps^{-2})\cdot \log^3 m$ by Lemma~\ref{lemma_mainring}, and $\|D_{\inn}\|_0\leq k+Q(\eps)$. By applying the state-of-the-art vanilla coreset \cite{Cohen25}, we have $Q(\eps) = \tilde{O}(kd\eps^{-2})$. Thus we have $\|D\|_0 = O(m)+\tilde{O}(kd\eps^{-2})\cdot \log^4 m = O(m) + \tilde{O}(kd\eps^{-2})$.\footnote{Here we use the fact that $B\log^4 A = B(\log B + \log\frac{A}{B})^4\lesssim B\log^4 B + A$ for $A,B\geq 1$.}
	
	Now we claim $D$ is an $(O(\eps),m)$-coreset of $X$ and the proof will be complete by rescaling $\eps$. Recall that 
	\[
	D=F\cup \left(\bigcup_{i=1}^l D_i\right)\cup D_{\inn},
	\]
	where $F$ is an $(\eps,m,0)$-robust coreset for $F$, each $D_i$ is an $(\eps,m,\eps r_i |R_i|)$-robust coreset for $R_i$ and $D_{\inn}$ is an $(\eps,m,\eps\cdot \cost^{(m)}(X,C^\ast))$-robust coreset. 
	Since $\sum_{i=1}^l r_i|R_i|= O(\cost^{(m)}(X,C^\ast))$,
	the mergeability of robust coresets (Fact~\ref{fact_merge}) implies that $D$ is an $(\eps,m,O(\eps\cdot \cost^{(m)}(X,C^\ast)))$-robust coreset of $X$. Recall that $\cost^{(m)}(X,C^\ast)\lesssim \cost^{(m)}(X,C)$ for every $C\in \binom{M}{k}$. Therefore, $D$ is also an $(O(\eps),m)$-robust coreset of $X$.
	
	By a union bound, we see that the failure probability is at most a constant.
\end{proof}

\subsection{Proof of Theorem \ref{thm_Euc}: Euclidean instances}
\label{proof_Euc}

For preparation, we have the following decomposition for every Euclidean instance. An illustration is given in Figure~\ref{fig:euc_decomp}.

\begin{figure} 
    \centering
    \includegraphics[width=0.8\linewidth]{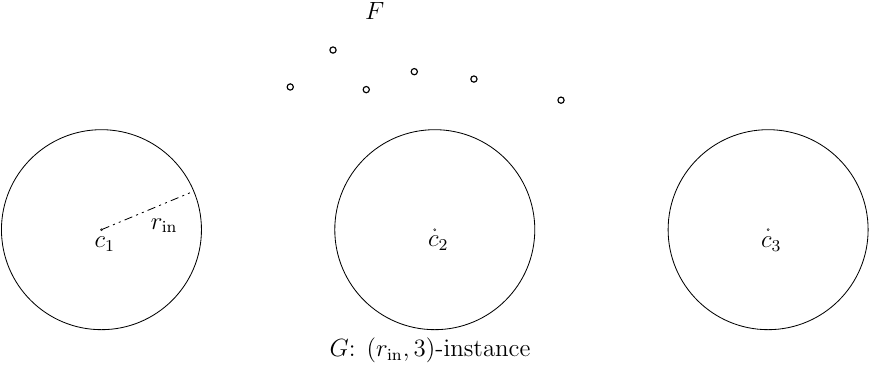}
    \vspace*{-0.3cm}
    \caption{An illustrated example of the decomposition for Euclidean instances in Lemma~\ref{thm_euc_decomp}. Note that $r_{\inn} = (\eps/m) \cost[m](X,C^\ast)$ in this figure, which is larger than that in Figure \ref{fig:VC_decomp} by a multiplicative factor of $k$.}
    \label{fig:euc_decomp}
\end{figure}

\begin{lemma}[Second point set decomposition]
\label{thm_euc_decomp}
	Suppose that $X\subset \mathbb{R}^d$ and let $C^\ast$ be an $O(1)$-approximation for robust $k$-median on $X$. Then, $X$ can be decomposed in $O(nk)$ time as
	$$
	X=F\cup G,
	$$ 
	where $|F|= O(m\eps^{-1})$ and $G$ is an $(r_{\inn},k)$-instance with $G\subset B(C^\ast, r_{\inn})$ for $r_{\inn} = \frac{\eps}{m} \cost[m](X,C^\ast)$.
\end{lemma}

\begin{proof}
	Define 
	$$
		F = \left\{x\in X \;\middle\vert\; \dist(x,C^\ast)> \frac{\eps}{m}\cdot \cost(X,C^\ast)\right\}.
	$$
	We argue that $|F|\leq m+m\eps^{-1}$. To see this, we note that except for $m$ outliers and at most $m\eps^{-1}$ inliers, each point $x\in X$ satisfies $\dist(x,C^\ast)\leq r_{\inn} := (\eps/m) \cdot \cost(X,C^\ast)$. Otherwise, the total cost over $X$ would exceed $\cost(X,C^\ast)$. Hence, letting $G:=X\setminus F$, we have that $G$ is an $(r_{\inn},k)$-instance and $G\subset B(C^\ast,r_{\inn})$.
\end{proof}

We are now ready to prove Theorem~\ref{thm_Euc}.

\begin{proof}[Proof of Theorem~\ref{thm_Euc}] Our coreset $D$ is constructed as follows. 
	First, we use Lemma~\ref{thm_euc_decomp} to compute a decomposition $X=F\cup G$ and add $F$ identically to $D$. Then we use Lemma~\ref{lemma_Euc} to compute an $(\eps,m,7\eps\cdot\cost^{(m)}(X,C^\ast))$-robust coreset $D_0$ of $G$ and add $D_0$ into $D$. By mergeability of robust coresets, we know that $D=F\cup D_0$ is an $(\eps,m,7\eps\cost[m](X,C^\ast))$-robust coreset of $X=F\cup G$. But $C^\ast$ is an $O(1)$-approximation of $X$ for $(k,m)$-robust $k$-median, hence $D$ is an $(O(\eps),m)$-robust coreset of $X$. Rescaling $\eps$ completes the proof.
\end{proof}

\section{Proof of Theorem \ref{thm_main}: First Coreset for VC Instances}
\label{sec:VC}

In this section, we prove Theorem \ref{thm_main} for the case that $d = d_{\VC}\leq d_D$, i.e., the coreset size is proportional to $d_{\VC}$. Throughout this section, we let $n:=|X|$ and $n_i:=|X_i|$ for each $i\in [k_0]$. We also let $D_i = D\cap X_i$ for each $i\in [k_0]$.
 For $I\subseteq [k_0]$, we define $X_I := \bigcup_{i\in I} X_i$, $D_I := D\cap X_I$ and $U_I := U\cap X_I$.

At a high level, we bound the induced error from the clusters containing only inliers and from the clusters containing both inliers and outliers separately.
To capture this idea, we introduce the following definitions.
The first definition is range space approximation, which is commonly used for metric spaces with bounded VC dimension; see, e.g., \cite{vershynin2018high}.

\begin{definition}[Range space approximation] \label{def_range}
We say $D$ is an $\eps$-range space approximation of $X$ if for every $I\subseteq [k_0]$, $C\in \binom{M}{k}$, and $r>0$,
$$
\abs[\big]{ |X_I\cap B(C,r)| - \norm{D_I\cap B(C,r)}_1 }\leq \eps n.
$$ 
\end{definition}

It requires that for each range $B(C,r)$ and each subset $I\subseteq [k_0]$, the total weight of points in $D_I \cap B(C,r)$ is a good approximation of the original size $|X_I\cap B(C,r)|$.
This definition is useful for bounding the induced error of $D$ restricted to those clusters that contain both inliers and outliers with respect to $C$.

Next, we propose a definition that bounds certain cost differences between $X$ and $D$.

\begin{definition}[Indexed-subset cost approximation] \label{def_subset}
We say $D$ is an $\eps$-indexed-subset cost approximation of $X$ if for every $I\subseteq [k_0]$, and every $C\in \binom{M}{k}$, 
$$
|\cost(X_I,C)-\cost(D_I,C)|\leq \eps \cost(X_I,C)+\eps n r.
$$
\end{definition}

This definition requires that $D$ preserves clustering cost for every subset $I\subseteq [k_0]$ of clusters and every center set $C$.
Different from the vanilla case that only needs to consider the induced cost difference of the whole $D$, this definition introduces an additional complexity by considering subsets $I$ of clusters.
This is useful for bounding the induced cost difference of those clusters that are completely inliers with respect to $C$.

Finally, we give the following definition that requires the weight of $D$ to align with $X$ for each cluster.
This property ensures that we can consider the same number of outliers within each cluster for $X$ and $D$.

\begin{definition}[Capacity-respecting] \label{def_cap}
We say $D$ is \emph{capacity-respecting} with respect to $X$ if for every $i\in [k_0]$, $\norm{D\cap X_i}_1 = \abs{X_i}$.
\end{definition}

Next, we give two key lemmata.
The first one provides a sufficient condition for $D$ being a robust coreset for $(r,k_0)$-regular instance, based on the above definitions. Its proof is postponed to Section~\ref{sec:lemma_main_tech}.

\begin{lemma}[Sufficient condition for robust coreset on bounded VC instances] 
\label{lemma_main_tech}
Let $X$ be an $(r,k_0)$-regular instance of size $n$ and $D$ a weighted subset of $X$.
If $D$ is simultaneously capacity-respecting, an $\eps$-range space approximation and an $\eps$-indexed-subset cost approximation of $X$, then $D$ is an $(O(\eps),m,O(\eps n r ))$-robust coreset of $X$.
\end{lemma}

Intuitively, this sufficient condition lets us bound the induced error from purely inlier clusters and from clusters that mix inliers and outliers separately, matching our high-level approach.

The second lemma shows that Algorithm~\ref{alg_coreset} indeed outputs a coreset $D$ satisfying the desired properties.

\begin{lemma}[Performance analysis of Algorithm~\ref{alg_coreset} for VC instances] \label{alg_guarantee}
With probability at least $1-O\big(\frac{1}{\log (k/\eps)\cdot \log (mk)}\big)$,  Algorithm~\ref{alg_coreset} outputs a weighted set $D$ which
is simultaneously an $\eps$-range space approximation of $X$, an $\eps$-indexed-subset cost approximation of $X$ and capacity-respecting with respect to $X$.
\end{lemma}

By Algorithm~\ref{alg_coreset}, $D$ is always capacity-respecting with respect to $X$, since
\[
    w_D(D\cap X_i)=\|U\cap X_i\|_0\cdot \frac{n_i}{u_i}\cdot \frac{n}{s}=\frac{su_i}{n}\cdot \frac{n_i}{u_i}\cdot \frac{n}{s}=n_i.
\]
Thus, it suffices to prove the other two properties of $D$.
We shall prove that $D$ is an $\eps$-range space approximation of $X$ in Section~\ref{sec:alg_guarantee_proof} and that $D$ is an $\eps$-indexed-subset cost approximation of $X$ in Section \ref{sec_proof_subset}.

Equipped with the preceding two lemmata, the proof of Theorem~\ref{thm_main} is immediate.

\begin{proof}[Proof of Theorem~\ref{thm_main} (VC instances)]
    We can combine Lemmata~\ref{lemma_main_tech} and~\ref{alg_guarantee} to obtain directly the coreset guarantee in Theorem~\ref{thm_main}, with a rescaling of $\eps$. The overall failure probability is $O\big(\frac{1}{\log(k/\eps)\cdot \log (mk)}\big)$ by Lemma~\ref{alg_guarantee}.
    %
\end{proof}


\subsection{Proof of Lemma~\ref{lemma_main_tech}: Sufficient condition on bounded VC instances}\label{sec:lemma_main_tech}

In the following, we assume that $D$ is simultaneously capacity-respecting, an $\eps$-range space approximation and an $\eps$-indexed-subset cost approximation of $X$.

Fix a $C\in \binom{M}{k}$ and an integer $t\in \{0,1,\cdots,m\}$. Observe that $|B(C,r)\cap X|$ is monotone in $r$. Let $r^\ast = \inf \{r>0\;\big|\; |B(C,r)\cap X|\geq n-t\}$. We partition $X_1,\dots,X_{k_0}$ into three groups. Specifically, let
\begin{itemize}[itemsep=0pt]
\item $A_1=\{i\in [k_0]\mid X_i\subseteq B(C,r^\ast)\}$,
\item $A_2=\{i\in [k_0]\mid X_i\cap B(C,r^\ast)=\emptyset\}$,
\item $A_3=[k_0]\setminus (A_1\cup A_2)$.
\end{itemize}

Let $X_{A_i}=\bigcup_{j\in A_i} X_j$ and $D_{A_i}=\bigcup_{j\in A_i} D_j$ for $i\in \{1,2,3\}$. 
Figure~\ref{fig:subset} shows an illustration of such a partition. We have the following observations. 

\begin{lemma}[Geometric properties of $D$] \label{proof_fact} The following statements hold.
\begin{enumerate}[label=(\roman*),itemsep=0pt]
\item $|X_{A_i}| = \norm{D_{A_i}}_1$ for every $i\in \{1,2,3\}$. \label{fact0}
\item All points in $X_{A_1}$ are inliers of $\cost[t](X,C)$ and all points in $X_{A_2}$ are outliers of $\cost[t](X,C)$.\label{fact1}
\item All points in $D_{A_1}$ are inliers of $\cost[t](D,C)$ and all points in $D_{A_2}$ are outliers of $\cost[t](D,C)$.\label{fact2}
\item $X_{A_3}\subseteq B(C,r^\ast+2r)\setminus B(C,(r^\ast-2r)^+)$. \label{proof_fact3}
\item $D_{A_3}\subseteq B(C,r^\ast+2r)\setminus B(C,(r^\ast-2r)^+)$. \label{proof_fact4}
\end{enumerate}
\end{lemma}

\begin{figure}[t]
    \centering
    \includegraphics[width=1\linewidth]{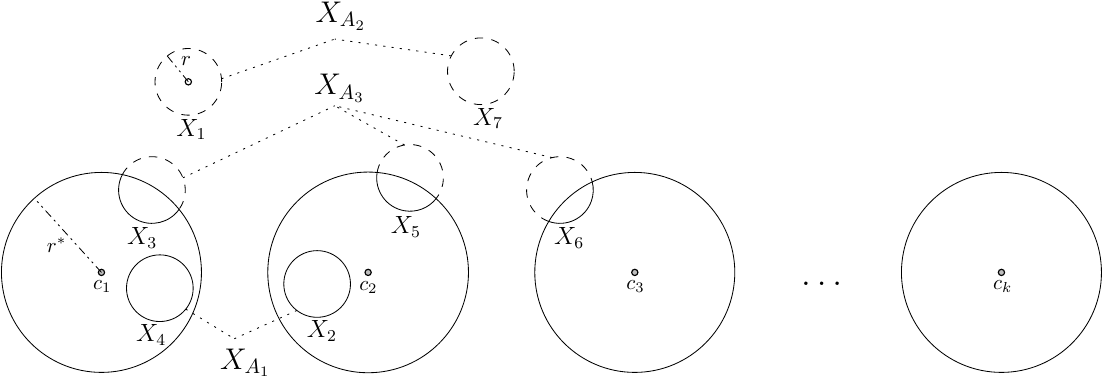}
    \vspace*{-0.7cm}
    \caption{An illustrated example of the partitioning of clusters in Section~\ref{sec:lemma_main_tech} }
    \label{fig:subset}
\end{figure}

\begin{proof}
Part \ref{fact0} follows from the capacity-respecting property of $D$. Part \ref{fact1} follows from the definition of $A_i$ ($i\in \{1,2,3\}$).

To prove \ref{fact2}, we note that by the capacity-respecting property of $D$ and the definition of $A_1$,  
\[
    \norm{D_{A_1}}_1 = \sum_{j\in A_1} \norm{D_j}_1 = \sum_{j\in A_1} |X_j|\leq n-t.
\]
Since $D_{A_1}\subseteq  X_{A_1}\subseteq B(C,r^\ast)$ and $\norm{D_{A_1}}_1 \leq n-t$, points in $D_{A_1}$ must all be inliers. Similarly, we know that $D_{A_2}\cap B(C,r^\ast)=\emptyset$ and $\norm{D_{A_2}}_1 \leq t$, so points in $D_{A_2}$ must all be outliers. So we have proved \ref{fact2}.

Next we prove \ref{proof_fact3} and \ref{proof_fact4}. Consider a point $x\in X_{A_3}$, by the definition of $A_3$, there exists $j\in A_3$ such that $x\in X_j$ and $X_j\cap B(C,r^\ast)$ is neither $X_j$ nor the empty set. Thus, there exists $y,z\in X_j$ such that $y\in B(C,r^\ast)$ and $z\not\in B(C,r^\ast)$, which implies, by triangle inequality, that $\dist(x,C)\leq \dist(x,y)+\dist(y,C)\leq r^\ast+2r$ and $\dist(x,C)\geq \dist(z,C)-\dist(x,z)\geq r^\ast-2r$. So we have proved Part \ref{proof_fact3}. Part \ref{proof_fact4} can be proved identically.
\end{proof}

Using Lemma~\ref{proof_fact}, we can obtain the following inequality that decomposes the cost difference between $X$ and $D$ into two parts.

\begin{lemma}[Cost difference decomposition] 
\label{lemma_inq}
It holds that
\begin{multline*}
|\cost[t](X,C)-\cost[t](D,C)|\leq |\cost(X_{A_1},C)-\cost(D_{A_1},C)| \\ +\int_{(r^\ast-2r)^+}^{r^\ast+2r} \abs[\Big]{ |B(C,u)\cap X_{A_3}| - \norm*{B(C,u)\cap D_{A_3}}_1 } du.
\end{multline*}
\end{lemma}

\begin{proof}

Let $g=|X_{A_2}\cap B(C,r^\ast)|=\sum_{i\in A_2} n_i$. By the definition of $r^\ast$, we know that $g\leq t$. Since points in $X_{A_1}$ are all in-liers and points in $X_{A_2}$ are all outliers, it holds that
\begin{align*}
\cost[t](X,C) &= \cost(X_{A_1},C)+\cost^{(t-g)}(X_{A_3},C)\\
&= \cost(X_{A_1},C)+\int_{0}^{\infty} \big( |X_{A_3}|-(t-g)-|B(C,u)\cap X_{A_3}|\big)^+du
\end{align*}
where the second equality follows from integrating by parts.

Similarly, we have 
\[
\cost[t](D,C) = \cost(D_{A_1},C) + \int_{0}^{\infty} \big( \norm{D_{A_3}}_1 - (t-g) -\norm*{B(C,u)\cap D_{A_3}}_1 \big)^+du
\]
Since the function $x\mapsto (x)^+$ is $1$-Lipschitz, 
\begin{align*}
& |\cost[t](X,C)-\cost[t](D,C)| \\
\leq{} & |\cost(X_{A_1},C)-\cost(D_{A_1},C)|+\int_{0}^{\infty} \bigg||B(C,u)\cap X_{A_3}|-\norm*{B(C,u)\cap D_{A_3}}_1 \bigg| du \\
={} & |\cost(X_{A_1},C)-\cost(D_{A_1},C)|+\int_{(r^\ast-2r)^+}^{r^\ast+2r} \bigg||B(C,u)\cap X_{A_3}|-\norm*{B(C,u)\cap D_{A_3}}_1 \bigg| du,
\end{align*}
where the inequality is due to $|X_{A_3}| = \norm*{D_{A_3}}_1$ and the equality is due to Lemma~\ref{proof_fact}\ref{proof_fact3}\ref{proof_fact4}.
\end{proof}

Now we are ready to prove Lemma~\ref{lemma_main_tech}.

\begin{proof}[Proof of Lemma~\ref{lemma_main_tech}]

    Combining Lemma~\ref{lemma_inq}, Definition~\ref{def_range} and Definition~\ref{def_subset}, we obtain that $$
|\cost[t](X,C)-\cost[t](D,C)|\leq \eps\cdot \cost[t](X,C) + O(\eps n r)
$$
which completes the proof of Lemma~\ref{lemma_main_tech}.
\end{proof}


\subsection{Proof of Lemma~\ref{alg_guarantee}: Range space approximation}\label{sec:alg_guarantee_proof}

In this section, we prove that $D$ is an $\eps$-range space approximation of $X$.
Recall that $D$ is derived by shifting the point weights of the original sample set $U$; see Algorithm \ref{alg_coreset}.
We first define the following event which implies that $D$ is ``sufficiently close to'' $U$:
\begin{eqnarray} \label{eq_event}
    \mathcal{E} = ``\text{$u_i\in (1\pm \eps)\cdot n_i$ for all $i\in [k_0]$''}.
\end{eqnarray}

We have the following lemma showing that event $\mathcal{E}$ happens with high probability.
The proof of the following lemma is essentially an adaption of Lemma 19 of \cite{Cohen2022Towards}, which we include for completeness.

\begin{lemma}\label{lemma_event} 
It holds that
$$
\mathbb{P}[\mathcal{E}]\geq 1 - O\left(\frac{\epsilon}{sk\log m}\right).
$$
\end{lemma}

\begin{proof}

Define for each $i\in [k_0]$ and $u\in U$ a random variable $v_{i,u} = \frac{n}{s}\mathbf{1}_{\{u\in X_i\}}$. Then $u_i = \sum_{u\in U} v_{i,u}$ and we can calculate that
\begin{gather*}
\E\sum_{u\in U} v_{i,u}=|X_i|=n_i, \quad \E\sum_{u\in U} v_{i,u}^2 \leq \frac{n\cdot n_i}{s}\leq \frac{2k_0n_i^2}{s}.
\shortintertext{Also note that}
\norm{v_{i,u}}_{\infty}\leq \frac{n}{s}\leq \frac{2k_0n_i}{s}.
\end{gather*}
It follows from Bernstein's inequality that
\[
\mathbb{P} \left\{\abs*{\sum_{u\in U} v_{i,u} - \E \sum_{u\in U} v_{i,u}}\geq \eps n_i\right\}\\
\leq \exp\left(-\frac{\eps^2 n_i^2}{\frac{4k_0n_i^2}{s}+\frac{2}{3}\cdot \frac{2k_0 n_i}{s}\cdot \eps n_i}\right)\\
\lesssim \frac{\epsilon}{k_0 k s\cdot \log m},
\] 
provided that $s\gtrsim \frac{k}{\eps^2} \log \frac{kd_{\VC}}{\epsilon}\cdot \log m$. Taking a union bound yields that $\mathbb{P}[\mathcal{E}]\geq 1-O(\frac{\epsilon}{sk\cdot \log m})$.
\end{proof}

We first prove that $U$ is an $\eps$-range space approximation of $X$. 
Then, we show that $D$ is an $O(\eps)$-range space approximation when conditioned on $\mathcal{E}$.

\begin{lemma}[Range space approximation for $U$] \label{lemma_range} 
With probability at least $1-O\big(\frac{1}{\log \frac{k}{\eps}\cdot \log (mk)}\big)$, the following inequality holds for every $J\subseteq [k_0]$, $C\in \binom{M}{k}$, and each $r>0$:
\begin{equation} \label{inq_range}
\abs[\bigg]{ |B(C,r)\cap X_J| - \norm*{B(C,r)\cap U_J}_1 }\leq \eps n.
\end{equation}
\end{lemma}

\begin{proof}

We begin by defining an indicator function $I_{J,C,r}: X \rightarrow \{0,1\}$ for each subset $J\subseteq [k_0]$, $C\in \binom{M}{k}$,  and each $r > 0$. This function is defined as $I_{J,C,r}(x) = 1$ if and only if $x \in X_J\cap B(C,r)$.
We observe that 
\begin{align*}
|B(C,r)\cap U_J|-|B(C,r)\cap X_J| &= \sum_{y\in U} \frac{n}{s}\cdot I_{J,C,r}(y)-\sum_{x\in X} I_{J,C,r}(x) \\
&=\sum_{y\in U} \frac{n}{s}\cdot I_{J,C,r}(y) - \E\left[ \sum_{y\in U} \frac{n}{s}\cdot I_{J,C,r}(y)\right].
\end{align*}

Our goal is to prove that with probability at least $1-O(\frac{1}{\log mk})$,
\begin{equation} \label{lemma1_eq1}
\sup_{C,r} \abs[\bigg]{ \norm*{B(C,r)\cap U_J}_1 - |B(C,r)\cap X_J| } \leq \eps n.
\end{equation}

We shall establish an upper bound for the expectation of equation \eqref{lemma1_eq1} and subsequently apply McDiarmid's inequality to derive a concentration inequality.
 We bound the expectation through the symmetrization trick (see, e.g., \cite[p153]{LT91}) with the maximum of Gaussian variables.


By the standard symmetrization trick,
\begin{align*}
\E\sup_{J,C,r} \abs*{ \sum_{y\in U} \frac{n}{s}\cdot I_{J,C,r}(y)-\E \left[ \sum_{y\in U} \frac{n}{s}\cdot I_{J,C,r}(y)\right] } 
&\leq 2\E\sup_{J,C,r} \abs*{ \sum_{y\in U} b_y\cdot \frac{n}{s}\cdot I_{J,C,r}(y) } \\
&\leq 2\cdot \sqrt{\frac{\pi}{2}} \cdot \E\sup_{J,C,r} \abs*{ \sum_{y\in U} g_y\cdot \frac{n}{s}\cdot I_{J,C,r}(y) }
\end{align*}
where $\{b_y\}_{y\in U}$ are independent Rademacher variables and $\{g_y\}_{y\in D}$ are independent standard Gaussians.


Next, we prove that 
\begin{equation}\label{eqn:supremum_gaussian_sum_range_space}
\E \sup_{J,C,r} \abs*{ \sum_{i=1}^s  g_i\cdot I_{J,C,r}(y_i) } \lesssim \sqrt{kds \log s}.
\end{equation}
To see this, recall that the VC dimension of the ball range space of $(M,\dist)$ is $d_{\VC}$, and since $|U|=s$, there are at most $s
^{O(d_{\VC})}$ essentially different ball range spaces over $X$ (in the sense of considering $X\cap B(c,r)$). Consequently, the number of distinct $k$-ball range spaces $B(C,r)\cap X$ over all choices of $C$ and $r$ is at most $N := s^{O(d_{\VC}k)}$. 
Thus, for each $J\subseteq [k_0]$, by considering $B(C,r)\cap X_J$,  the number of distinct sets 
$I_{J,C,r}$ is at most $2^{k_0}N=s^{O(d_{\VC}k)}$. This implies that the left-hand side of \eqref{eqn:supremum_gaussian_sum_range_space} is an expected supremum of $2^k N$ gaussian variables, each with a variance at most $s$. Applying Fact~\ref{fact_max_gauss}, the left-hand side of \eqref{eqn:supremum_gaussian_sum_range_space} is therefore $O(\sqrt{s}\cdot \sqrt{\log (2^k N)})=O(\sqrt{kd_{\VC}s\log s})$.

It follows that
\begin{gather*}
    \E \sup_{J,C,r} \abs[\Bigg]{ \norm{B(C,r)\cap U_J}_1 - |B(C,r)\cap X_J| } \leq n\cdot \sqrt{\frac{kd\log s}{s}}\leq \eps n
\shortintertext{when}
    s \gtrsim \frac{kd_{\VC}}{\eps^2}\cdot \log \frac{kd_{\VC}}{\eps}.
\end{gather*}
It remains to control the probability
$$
\mathbb{P}\left\{ \sup_{J,C,r} \abs[\Big]{ \norm{B(C,r)\cap U_J}_1 - |B(C,r)\cap X_J| } - \E \sup_{J,C,r} \abs[\Big]{ \norm{B(C,r)\cap U_J}_1 - |B(C,r)\cap X_J| } > \eps n \right\}.
$$

We use McDiarmid's inequality. Observe that for the expression $\sup_{J,C,r} \abs[\big]{ \norm{B(C,r)\cap U_j}_1 - |B(C,r)\cap X_J| }$, changing a single sample in $U$ results in a maximum change of $2\cdot \frac{n}{s}$. Thus, the total variance is bounded by $s\cdot (\frac{2n}{s})^2=\frac{4n^2}{s}$. It follows from McDiarmid's inequality that
\begin{multline*}
\mathbb{P} \left\{ \sup_{J,C,r} \abs[\Big]{ \norm{B(C,r)\cap U_J}_1 - |B(C,r)\cap X_J| } - \E \sup_{J,C,r} \abs[\Big]{ \norm{B(C,r)\cap U_J}_1 - |B(C,r)\cap X_J| } > \eps n \right\}\\
\leq  \exp\left(-\frac{2s (\eps n)^2}{4n^2}\right) 
=   \exp(-2\eps^2 s) 
\lesssim \frac{1}{\log mk},
\end{multline*}
provided that 
\[
    s \gtrsim \frac{1}{\eps^2} \cdot \log \frac{k}{\eps}\cdot \log (mk).
\]
Rescaling $\eps$ completes the proof.
\end{proof}

We are now ready to prove Lemma~\ref{alg_guarantee} for $D$ being a range space approximation.

\begin{proof}[Proof of Lemma~\ref{alg_guarantee}: $D$ is a range space approximation]

By a union bound, with probability at least $1-O\big(\frac{1}{\log(k/\eps)\cdot \log (mk)}\big)$, event $\mathcal{E}$  in \eqref{eq_event}, the event of Lemma~\ref{lemma_range} and the event of Lemma~\ref{lemma_subset} all happen. We condition on these events in the remainder of the proof.

We note that the supports of $U$ and $D$ are the same by construction. We first claim that under $\mathcal{E}$, $w_D(x)\in (1\pm O(\eps))w_U(x)$ for every $x\in U$. Indeed, by the definition of $\mathcal{E}$, we have for each $i\in [k_0]$ that
\begin{eqnarray} \label{eq_weight}
w_D(x)=\frac{n_i}{u_i}\cdot \frac{n}{s}=(1\pm O(\eps))\cdot \frac{n}{s}=(1\pm O(\eps)) \cdot w_U(x).
\end{eqnarray}
This implies for every $J\subseteq [k_0]$, $C\in \binom{M}{k}$ and $r>0$,
\begin{eqnarray} \label{inq_weight}
w_D(D_J\cap B(C,r))\in (1\pm O(\eps)) \cdot w_D(U_J\cap B(C,r)).
\end{eqnarray}
Plugging \eqref{inq_weight} into \eqref{inq_range} gives that
\[
    |B(C,r)\cap X| - O(\eps n)\leq w_D(D_J\cap B(C,r)) \leq |B(C,r)\cap X| + O(\eps n)
\] 
for every $J,C,r$. Rescaling $\eps$ shows that $D$ satisfies Definition~\ref{def_range}.
\end{proof}

\subsection{Proof of Lemma~\ref{alg_guarantee}: Indexed-subset cost approximation}
\label{sec_proof_subset}

In this section, we prove that $D$ is an $\eps$-indexed-subset cost approximation of $X$.
Similar to Section~\ref{sec:alg_guarantee_proof}, we first prove that $U$ is an $\eps$-indexed-subset cost approximation of $X$ and then translate to $D$ conditioned on event $\mathcal{E}$.
We have the following key lemma for $U$.

\begin{lemma}[Indexed-subset cost approximation for $U$] \label{lemma_subset}  
With probability $1-O\big(\frac{1}{\log(k/\eps)\cdot \log (mk)}\big)$, $U$ is an $\eps$-indexed-subset cost approximation, i.e., it holds for every $J\subseteq [k_0]$ and every $C\in \binom{M}{k}$ that
\begin{equation} \label{inq_subset}
    \abs{ \cost(U_J,C)-\cost(X_J,C) } \leq \eps\cdot \cost(X_J,C)+\eps n r.
\end{equation}
\end{lemma}

The proof employs a chaining argument developed in recent works \cite{fefferman2016testing,Cohen2022Towards,Cohen25}.  
In the remaining text, we provide the proof structure and highlight our novelties.
The complete version can be found in Appendix \ref{sec_proof_subset_appendix}.

Recall that our instance is an $(r,k_0)$-regular instance, given by $X=X_1\cup\cdots\cup X_{k_0}$, where $X_i\subset B(a_i,r)$ for some $a_i\in M$. For every $x\in X$, let $\pi(x)$ denote the unique index $j\in [k_0]$ such that $x\in X_j$. Then $\dist(x,a_{\pi(x)})\leq r$.

Let $T=\lceil \log 1/\eps\rceil$. We define $A_0=[0,2r]$, and $A_i=(2^i r,2^{i+1}r]$ for each $i\in [T]$. Let $I_{i,C}=\{j\in [k_0]\mid \dist(a_j,C)\in A_i\}$ for each $i\in \{0,\dots,T\}$.

For every $i\in \{0,\dots,T\}$, $J\subseteq [k_0]$, $C\in \binom{M}{k}$, and every subset $S\subseteq X$, the distance vector $v_{i,J,C}^S:S\rightarrow \mathbb{R}_{\geq 0}$ is defined as
\begin{equation}
\label{eq:vS}
v^S_{i,J,C}(x)=
\begin{cases}
    \dist(x,C), & \pi(x)\in J\cap I_{i,C} \\
    0, & \text{otherwise}.
\end{cases}
\end{equation}
Compared to the vanilla case, we have an additional subscript $J$ for distance vectors, which enables the handling of the robust case. 
This notion ensures that if $\dist(x,C) \neq 0$, we must have that $\pi(x)\in J$ and $\dist(a_{\pi(x)}, C)\in A_i$.
Since $i\leq T$, we know that $\dist(a_{\pi(x)}, C) \leq 4r/\eps$.
This inequality motivates our choice of \( T \). Specifically, if \( \dist(a_{\pi(x)}, C) > 4r/\eps \), then all points \( y \in B(a_{\pi(x)}, r) \) are sufficiently far from \( C \) such that $d(y, C) \in (1\pm \eps)\cdot d(x,C)$.
Thus, this cluster $B(a_{\pi(x)}, r)$ only induces a multiplicative error and can be safely ignored.

Define the level-$i$ distance vector set as $V^S_{i}=\{v^S_{i,J,C}\mid J\subseteq [k_0],C\in \binom{M}{k}\}$. The following lemma establishes the relationship between indexed-subset cost approximation and $V_i^U$.

We have the following lemma that reduces Lemma \ref{lemma_subset} to bounding chaining errors with respect to $v^S_{i,J,C}$.

\begin{lemma}[Sufficient chaining error for Lemma \ref{lemma_subset}]
\label{lemma_relate_main} 
If for every $i\in \{0,\dots,T\}$, 
\begin{equation} \label{inq_chaining_main}
\sup_{J,C} \frac{n}{s}\cdot \abs*{ \frac{\sum_{x\in U} v^U_{i,J,C}(x)-\E \sum_{x\in U} v^U_{i,J,C}(x)}{(1+2^i|J_{i,C}|/k_0)\cdot nr} }
\leq \frac{\eps}{\log(1/\eps)},
\end{equation}
where $J_{i,C}=J\cap I_{i,C}$,
and event $\mathcal{E}$ (defined in \eqref{eq_event}) happens, then $U$ is an $\eps$-indexed-subset cost approximation. 
\end{lemma}

To employ a chaining argument in proving inequality (\ref{inq_chaining_main}), we need to control the net size of $V_i^S$.
A subset $N\subseteq V^S_i$ is called an $\alpha$-net of $V^S_i$ if for every $i,J,C$, there exists a vector $v\in N$ such that
\[
\begin{cases}
|v(x)-v^S_{i,J,C}(x)|\leq \alpha\cdot 2^{i+2}\cdot r & v^S_{i,J,C}(x)>0\\
0 & v^S_{i,J,C}(x)=0
\end{cases}
\]
Note that our net covers cost vectors with respect to all subsets $J$ instead of only $J = [k]$.

We present the following lemma, which provides an upper bound on the net size.  
This result is essentially an adaptation and extension of Lemma~5.3 from \cite{huang2025coresets} and Lemma~3.6 from \cite{Cohen25}, with an additional consideration for the subscript \( J \).  
Fortunately, the inclusion of \( J \) increases the net size by at most a factor of \( 2^k \), which does not affect its overall order.

\begin{lemma}[Bounding net size on VC instances] 
\label{lemma_netsize_main}
For every $i\in \{0,\dots,T\}$, $\alpha\in (0,1)$ and every subset $S\subseteq X$ with $\abs{S}\geq 2$, there exists an $\alpha$-net $N^S_{\alpha,i}$ of $V^S_i$ such that
\[
\log \abs[big]{N^S_{\alpha,i}} \lesssim \alpha^{-1}d k\log k \cdot \log\abs{S}.
\]
Furthermore, if $\alpha\geq 2^{-i}$, it is possible to ensure that $N_{\alpha,i}^S$ consists of piecewise constant vectors. That is, every $v\in N_{\alpha,i}^S$ satisfies that $v(x)$ is constant for all $x\in X_j$ for each $j\in [k_0]$.
\end{lemma}

Using a standard chaining argument, we can use this net size to derive Inequality (\ref{inq_chaining_main}) (see Lemma \ref{lemma_chaining}).
\section{Proof of Theorem \ref{thm_main}: First Coreset for Doubling Instances} \label{sec:Doubling}

In this section, we prove Theorem~\ref{thm_main} for the case \( d = d_D \leq d_{\VC} \), where the coreset size is proportional to \( d \).  
Notably, the VC dimension of a doubling instance can be unbounded \cite{huang2018epsilon}, preventing us from directly applying the proof structure from Section~\ref{sec:VC}, as \( D \) is not necessarily a range space approximation.  
To address this challenge, we adopt the approach from \cite{huang2018epsilon} to effectively handle the range space of doubling instances.

We first introduce the following concept that slightly perturbs the distance function. 

\begin{definition}[Smoothed distance function]
\label{def:eps_smooth}
Let $(M,\dist)$ be a metric space. A function $\delta_\eps: M\times M \rightarrow [0,\infty)$ is called an \emph{$\eps$-smoothed distance function} of $(M,\dist)$ if it holds for all $x,y \in M$ that
\[
    (1-\eps ) \dist (x,y) \leq \delta_\eps(x,y)\leq (1+\eps)\dist(x,y).
\]
\end{definition}

When $\eps$ is clear in the context, we may omit the subscript in $\delta_\eps$ and write $\delta$ only.
We denote $B^\delta(x,r)=\{y\mid \delta(x,y)\leq r\} $ as the ball centered at $x$ with respect to the smoothed distance function $\delta$ and by $\cost^{\delta}(X,C)$ the cost function under $\delta$. 
Similarly, we use $\cost^{(t,\delta)}$ to define the cost function with $t$ outliers based on the smoothed distance function $\delta$. 
The same applies to other related definitions, where replacing the original distance function with $\delta$ yields the corresponding modified versions.

We shall use the following variant of Definitions~\ref{def_range} with respect to an $\eps$-smoothed distance function $\delta$. 

\begin{definition}[Range space approximation with respect to smoothed distance function] \label{def:eps_smoothed_range}
    We say $D$ is an $\eps$-range space approximation of $X$ with respect to a smoothed distance function $\delta$ if for every $I\subseteq[k_0]$, $C\in \binom{M}{k}$ and $r>0$,
    \[
        \abs[\big]{ X_I\cap B^\delta (C,r) - \norm{ D_I \cap B^\delta(C,r) }_1 
        }\leq \eps n.
    \]
\end{definition}




%

%
Similar to Section~\ref{sec:VC}, we present two key lemmata below. 
The first provides another sufficient condition for $D$ to be a robust coreset for $(r,k_0)$-regular instance, which is useful for doubling spaces. 
Its proof is deferred to Section~\ref{sec:lemma:main:doubling}.

\begin{lemma}[Sufficient condition for robust coreset on bounded doubling dimension] \label{lemma:main:doubling}
    Let $X$ be an $(r,k_0)$-regular instance of size $n$ and $D$ is a weighted subset of $X$. If $D$ is simultaneously capacity-respecting, an $\eps$-range space approximation with respect to an $\eps$-smoothed distance function, and an $\eps$-indexed-subset cost approximation, then $D$ is an $(O(\eps), m, O(\eps n r))$-robust coreset of $X$.
\end{lemma}

The second lemma shows that Algorithm \ref{alg_coreset} indeed outputs a coreset $D$ satisfying the aforementioned properties, whose proof can be found in Section \ref{sec:lemma_main:range:doubling}.

\begin{lemma}[Performance analysis of Algorithm~\ref{alg_coreset} for doubling instances]\label{lem:alg:doubling}
    With probability $1-O(\frac{1}{\log(k/\eps)\log{(km)}})$, Algorithm~\ref{alg_coreset} outputs a weighted set which is simultaneously capacity-respecting w.r.t.\ $X$, an $\eps$-range space approximation with respect to some (random) $\eps$-smoothed distance function, and an $\eps$-indexed-subset cost approximation of $X$.
\end{lemma}

Now we are ready to prove Theorem~\ref{thm_main} for doubling instances.

\begin{proof}[Proof of Theorem~\ref{thm_main} (doubling instances)]
    Combining Lemmata~\ref{lem:alg:doubling} and~\ref{lemma:main:doubling} yields the coreset guarantee in Theorem~\ref{thm_main}. The overall failure probability is $O(\frac{1}{\log(k/\eps)\cdot \log{mk}})$ by Lemma~\ref{lem:alg:doubling}.
\end{proof}

\subsection{Proof of Lemma~\ref{lemma:main:doubling}: Sufficient condition on bounded doubling instances} \label{sec:lemma:main:doubling}

We suppose that $D$ is capacity-respecting, an $\eps$-range space approximation with respect to an $\eps$-smoothed distance function $\delta$, and an $\eps$-smoothed indexed-subset cost approximation of $X$. 
Let $C\in \binom{M}{k}$ be a fixed center set and $t\in [m]$. 
Let $D_i=D\cap X_i$ and $r^\ast=\inf \{r>0\;\big|\; |B(C,r)\cap X|\geq n-t\}$. 
Recall that we can divide the $X_1\dots,X_{k_0}$ into the following three parts $A_1,A_2,A_3$ as in Section~\ref{sec:VC}. 
\begin{itemize}[itemsep=0pt]
\item $A_1=\{i\in [k_0]\mid X_i\subseteq B 
(C,r^\ast)\}$,
\item $A_2=\{i\in [k_0]\mid X_i\cap B 
(C,r^\ast)=\emptyset\}$,
\item $A_3=[k_0]\setminus (A_1\cup A_2)$.
\end{itemize}

Based on Lemma~\ref{proof_fact}, we can obtain the following inequality for cost under an $\eps$-smoothed distance function $\delta$, again decomposing the difference between $X$ and $D$ into several parts.

\begin{lemma}[Cost difference decomposition for $\delta$]\label{lemma:diffrerence:doubling}
If $D$ is capacity-respecting, it holds that 
\begin{equation} \label{inq:cost:doubling}
	\begin{aligned}
	    \abs{ \cost[t](X,C)-\cost[t](D,C) } &\leq \abs{ \cost(X_{A_1},C)-\cost(D_{A_1},C) } \\
	    &\quad + \int_{(r^\ast-2r)^+}^{r^\ast+2r} \abs*{ |B^\delta(C,u)\cap X_{A_3}|-\norm[\big]{B^\delta(C,u)\cap D_{A_3}}_1 } du \\
	    &\quad +\eps  \cost^{(t)}(X,C) +\eps \cost[t](D,C) + O(\eps n r).
	\end{aligned}
\end{equation}
\end{lemma}

\begin{proof} \label{proof:diffrerence:doubling}
    Let $g=\abs{ X_{A_2}\cap B(C,r^\ast) }$, then $g\leq t$. The same argument as in the proof of Lemma~\ref{lemma_inq} gives
    \begin{gather*}
    	\cost[t](X,C) = \cost(X_{A_1},C)+\cost[t-g](X_{A_3},C) \\
    	\cost[t](D,C) = \cost(X_{D_1},C)+\cost[t-g](D_{A_3},C).
    \end{gather*}
    Next we shall replace $\cost[t-g](X_{A_3},C)$ and $\cost[t-g](D_{A_3},C)$ with their distance-$\delta$ counterparts, $\cost[t-g,\delta](X_{A_3},C)$ and $\cost[t-g,\delta](D_{A_3},C)$, respectively. We claim that
    \begin{gather}
    	(1-\eps)\cost[t-g](X_{A_3},C) \leq \cost[t-g,\delta](X_{A_3},C) \leq (1+\eps)\cost[t-g](X_{A_3},C), \label{eqn:cost_under_delta_X_A_3} \\
    	(1-\eps)\cost[t-g](D_{A_3},C) \leq \cost[t-g,\delta](D_{A_3},C) \leq (1+\eps)\cost[t-g](D_{A_3},C). 
    	\label{eqn:cost_under_delta_D_A_3}
    \end{gather}
    Now we prove \eqref{eqn:cost_under_delta_X_A_3}.     
    Suppose $Y\subseteq X_{A_3}$ is the subset of points retained in $\cost[t-g](X_{A_3},C)$ and $Y^\delta$ be the subset retained in $\cost[t-g,\delta](X_{A_3},C)$; that is, $\cost[t-g](X_{A_3},C) = \cost(Y,C)$ and $\cost[t-g,\delta](X_{A_3},C) = \cost(Y^\delta,C)$.
	By the definition of the $\eps$-smoothed distance function, it follows that
	\begin{gather*}
		\cost[\delta] (Y^\delta,C)\leq \cost[\delta](Y,C) \leq (1+\eps)\cost(Y,C)
	\shortintertext{and}
		\cost[\delta](Y^\delta,C) \geq (1-\eps)\cost(Y^\delta,C) \geq (1-\eps)\cost(Y,C)
	\end{gather*}
	A similar argument would prove \eqref{eqn:cost_under_delta_D_A_3}. We have completed the proof of our claim.

	So far we have obtained
	\begin{equation}\label{eqn:cost_diff_doubling}
	\begin{aligned}
		 & \abs{\cost[t](X,C)-\cost[t](D,C)}  \\
		\leq{}& \abs{\cost(X_{A_1},C)-\cost(D_{A_1},C)} + \abs{ \cost[t-g,\delta](X_{A_3},C)-\cost[t-g,\delta](D_{A_3},C)} \\
		&\quad + \eps(\cost[t-g](X_{A_3},C) + \cost[t-g](D_{A_3},C)).
	\end{aligned}
	\end{equation}
    Next we deal with the middle term $\abs{ \cost[t-g,\delta](X_{A_3},C)-\cost[t-g,\delta](D_{A_3},C)}$.
    The same argument as in the proof of Lemma~\ref{lemma_inq} allows us to write $\cost[t-g,\delta](X_{A_3},C)$ and $\cost[t-g,\delta](D_{A_3},C)$ as the following integrals.
	\begin{align*}
    	\cost[t-g,\delta](X_{A_3},C)&= \int_{0}^{\infty} (\abs{X_{A_3}} - (t-g)-\vert B^\delta(C,u)\cap X_{A_3} \vert )^+ du, \\
	    \cost[t-g,\delta] (D_{A_3},C )&=\int_0^\infty (\norm{D_{A_3}}_1 - (t-g) - \norm{B^\delta(C,u)\cap D_{A_3}}_1)^+ du.
	\end{align*}
	It follows from Lemma~\ref{proof_fact}\ref{proof_fact3} and~\ref{proof_fact4} along with the fact that $\norm{D_{A_3}}_1 = \abs{X_{A_3}}$ that
    
	\begin{equation}\label{eqn:integral_split}
    \begin{aligned}
		 & \cost[t-g,\delta](X_{A_3},C) - \cost[t-g,\delta](D_{A_3},C) \\
        \leq{} & 
        \int_{(1-\eps)(r^\ast-2r)^+}^{(1+\eps)(r^\ast+2r)} \abs*{(\abs{X_{A_3}} - (t-g) - \abs{B^\delta (C,u)\cap X_{A_3}})^+ - 
		  (\norm{D_{A_3}}_1 - (t-g) -  \norm{B^\delta(C,u)\cap D_{A_3}}_1)^+} du \\
	%
    \leq{} &\int_{(r^\ast-2r)^+}^{r^\ast+2r} \abs*{ |B^\delta(C,u)\cap X_{A_3}|-\norm[big]{B^\delta(C,u)\cap D_{A_3}}_1 } du \\
		&\quad +\int_{(1-\eps)(r^\ast-2r)^+}^{(r^\ast-2r)^+} \abs[\big]{(\abs{X_{A_3}} - (t-g)-\abs{B^\delta(C,u)\cap X_{A_3})}^+-(\norm{D_{A_3}}_1 - (t-g)- \norm{B^\delta(C,u)\cap D_{A_3}}_1)^+} du \\
		&\quad +\int_{(r^\ast+2r)}^{(1+\eps)(r^\ast+2r)} \abs[\big]{(\abs{X_{A_3}} - (t-g)-\abs{B^\delta(C,u)\cap X_{A_3}})^+-(\norm{D_{A_3}}_1 - (t-g)- \norm{B^\delta(C,u)\cap D_{A_3}}_1)^+} du.
		\end{aligned}
	\end{equation}
	We can upper bound the second term on the right-hand side of \eqref{eqn:integral_split} as
	\begin{align*}
		&\int_{(1-\eps)(r^\ast-2r)^+}^{(r^\ast-2r)^+} \abs[\big]{(\abs{X_{A_3}}-(t-g)-\abs{B^\delta(C,u)\cap X_{A_3}})^+ - (\norm{D_{A_3}}_1 - (t-g)- \norm{B^\delta(C,u)\cap D_{A_3}}_1)^+} du \\
		\leq{} & \int_{(1-\eps)(r^\ast-2r)^+}^{(r^\ast-2r)^+} (\abs{X_{A_3}}-(t-g))^+ du + \int_{(1-\eps)(r^\ast-2r)^+}^{(r^\ast-2r)^+} (\norm{D_{A_3}}_1 - (t-g))^+ du \\
		\leq{} & (\abs{X_{A_3}}-(t-g))^+ \cdot \eps 
		(r^\ast-2r)^+  + (\norm{D_{A_3}}_1 - (t-g))^+ \cdot \eps (r^\ast-2r)^+\\
		\leq{} & \eps\cdot \cost[t-g](X_{A_3},C) + \eps\cdot \cost[t-g](D_{A_3},C).
	\end{align*}
	Similarly, the third term on the right-hand side of \eqref{eqn:integral_split} can be bounded as
	\begin{align*}
		&\int_{(r^\ast+2r)}^{(1+\eps)(r^\ast+2r)} \abs[\big]{(\abs{X_{A_3}} - (t-g)-\abs{B^\delta(C,u)\cap X_{A_3}})^+-(\norm{D_{A_3}}_1 - (t-g)- \norm{B^\delta(C,u)\cap D_{A_3}}_1)^+} du \\
		\leq{} & \int_{(r^\ast+2r)}^{(1+\eps)(r^\ast+2r)} (\abs{X_{A_3}}-(t-g))^+ du + \int_{(r^\ast+2r)}^{(1+\eps)(r^\ast+2r)} (\norm{D_{A_3}}_1 - (t-g))^+ du \\
		\leq{} & (\abs{X_{A_3}}-(t-g))^+ \cdot \eps (r^\ast + 2r)  + (\norm{D_{A_3}}_1 - (t-g))^+ \cdot \eps (r^\ast + 2r) \\
		\leq{} & (\abs{X_{A_3}}-(t-g))^+ \cdot \eps (r^\ast - 2r)^+  + (\norm{D_{A_3}}_1 - (t-g))^+ \cdot \eps (r^\ast - 2r)^+ + O(\eps n r) \\
		\leq{} & \eps\cdot \cost[t-g](X_{A_3},C) + \eps\cdot \cost[t-g](D_{A_3},C) + O(\eps n r).
	\end{align*}
	Therefore, we can proceed upper bounding \eqref{eqn:integral_split} as
	\begin{equation}\label{eqn:integral2}
	\begin{aligned}
		&\abs*{\cost[t-g,\delta](X_{A_3},C)-\cost[t-g,\delta](D_{A_3},C)} \\
		\leq{} & \int_{(r^\ast-2r)^+}^{r^\ast+2r} \abs*{ |B^\delta(C,u)\cap X_{A_3}|-\norm[big]{B^\delta(C,u)\cap D_{A_3}}_1 } du \\
		&\quad + 2\eps\cdot \cost[t-g](X_{A_3},C) + 2\eps\cdot \cost[t-g](D_{A_3},C) + O(\eps n r).
	\end{aligned}
	\end{equation}
	Plugging \eqref{eqn:integral2} into \eqref{eqn:cost_diff_doubling} yields that
	\begin{align*}
		        & \abs*{\cost[t-g,\delta](X_{A_3},C)-\cost[t-g,\delta](D_{A_3},C)} \\
		\leq{}  & \abs{\cost(X_{A_1},C)-\cost(D_{A_1},C)} + \int_{(r^\ast-2r)^+}^{r^\ast+2r} \abs*{ |B^\delta(C,u)\cap X_{A_3}|-\norm[big]{B^\delta(C,u)\cap D_{A_3}}_1 } du  \\
			    & \quad + O(\eps)(\cost[t-g](X_{A_3},C) + \cost[t-g](D_{A_3},C)) + O(\eps n r) \\
		\leq{}  & \abs{\cost(X_{A_1},C)-\cost(D_{A_1},C)} + \int_{(r^\ast-2r)^+}^{r^\ast+2r} \abs*{ |B^\delta(C,u)\cap X_{A_3}|-\norm[big]{B^\delta(C,u)\cap D_{A_3}}_1 } du \\
				& \quad + O(\eps)\left( \cost[t](X,C) + \cost[t](D, C) \right) + O(\eps n r).
	\end{align*}
	The proof of Lemma~\ref{lemma:diffrerence:doubling} is now complete.
\end{proof}

Lemma~\ref{lemma:main:doubling} now follows immediately.
\begin{proof}[Proof of Lemma~\ref{lemma:main:doubling}]
    We can combine Lemma~\ref{lemma:diffrerence:doubling}, Definitions~\ref{def_subset} and \ref{def:eps_smoothed_range} to obtain that
    \[
	    \abs{ \cost[t](X,C)-\cost[t](D,C)} \leq O(\eps)\cdot (\cost[t](X,C)+\cost[t](D,C))+O(\eps nr),
    \]
    whence it follows that
    \[
	    \abs{ \cost[t](X,C)-\cost[t](D,C)} \leq O(\eps)\cost[t](X,C) + O(\eps nr),
    \]
    which completes the proof of Lemma~\ref{lemma:main:doubling}.
\end{proof}

\subsection{Proof of Lemma~\ref{lem:alg:doubling}: Performance analysis of Algorithm~\ref{alg_coreset} (doubling)}
\label{sec:lemma_main:range:doubling}

In Algorithm~\ref{alg_coreset}, $D$ is always capacity-respecting with respect to $X$ by Lemma~\ref{alg_guarantee}. Thus, it suffices to show separately the $\eps$-range space approximation property and the indexed-subset cost approximation property.

\paragraph{Range space approximation.}
By Lemma~\ref{lemma_event}, it suffices to prove that $U$ is an $\eps$-range space approximation of $X$ w.r.t. $\delta$, which implies that $D$ is the $O(\eps)$-range space approximation w.r.t. $\delta$ when $\mathcal{E}$ occurs.

We first present two auxiliary lemmata regarding the $\eps$-smoothed distance function and the associated ball range space, both taken from \cite{huang2018epsilon}. To state the lemmata, we need some notations. Let $\Ballh[\delta](C,r)=\{ x\in H \mid \delta(x,C)\leq r  \}$ denote the union of $r$-radius balls centered at points in $C$ with respect to the distance function $\delta$. We define $\Ballsh[\delta]=\{\Ballsh[\delta](C,r) \mid C\in \binom{M}{k},r\geq 0\} $ to represent the set of essentially distinct ball range spaces. 

\begin{lemma}[Random $\eps$-smoothed distance function, {\cite[Eq.(10)]{huang2018epsilon}}] \label{lemma:randomdist}
    Suppose that $(M,\dist)$ is a doubling space of doubling dimension $d$. Let $0<\gamma<1$ be a constant. There exists a random $\eps$-smoothed distance function $\delta$ such that for every nonempty $H\subseteq X$, 
    \begin{equation}\label{eqn:doubling_ball_range_space_bound}
        \Pr_\delta \{ \abs{\Ballsh[\delta]}\leq T(\abs{H},\gamma) \}\geq 1-\gamma,
    \end{equation}
    where $T:\Z^+\times \R^+ \rightarrow \mathbb{R}$ is given by
    \begin{equation}\label{eqn:doubling_T}
       T(q,\gamma) = e^{O(k\cdot d\cdot \log\frac{1}{\eps})}\cdot \log^k{\frac{q}{\gamma}}\cdot q^{6k}.
    \end{equation}
\end{lemma}

\begin{lemma} [Range space approximation for random distance function, {\cite[Lemma 3.1]{huang2018epsilon}}] \label{lemma:randomapprox}
    Suppose that $\delta$ is a random $\eps$-smoothed distance function and $T:\Z^+ \times \R^+\to \R$ is a function such that \eqref{eqn:doubling_ball_range_space_bound} holds 
    for every nonempty $H\subseteq X$ and constant $0<\gamma<1$.
    Let $S$ be a collection of $s$ independent uniform samples from $X$. Then with probability at least $1-\tau$, $S$ is an $\alpha$-range space approximation of $X$ w.r.t. $\delta$, where
    \begin{equation}\label{eqn:doubling_range_space_alpha}
        \alpha = \sqrt{\frac{48( \log{T(2s,\frac{\tau}{4})} +\log{\frac{8}{\tau}})}{s}}.
    \end{equation}
\end{lemma}

We are now ready to establish the guarantee for our sample set $U$.
\begin{lemma} [Range space approximation for $U$]
\label{lemma:eps_range}
    With probability $1 - O(\frac{1}{\log(k/\eps)\log{mk}})$, the following inequality holds for every $J\subseteq [k_0]$, $C\in \binom{M}{k}$, and each $r>0$:
    \[
    \abs[\big]{\abs{ B^\delta(C,r)\cap X_J } - \norm{B^\delta(C,r)\cap  U_{J}}_1 } \leq \eps n.
    \]
\end{lemma}

\begin{proof}
We use the method from \cite[Lemma 3.1]{huang2018epsilon}. 
Combining Lemmata~\ref{lemma:randomdist} and~\ref{lemma:randomapprox}, we obtain that there exists a random $\eps$-smoothed distance function $\delta$ such that 
$U$ is an $\alpha$-range space approximation of $X$ w.r.t. $\delta$ with probability at least $1-\tau$, where $\alpha$ is the same as in \eqref{eqn:doubling_range_space_alpha} and $T$ in \eqref{eqn:doubling_T}.

Since $J\subseteq [k_0]$ and $k_0\leq k$, the number of distinct $X_J$ is at most $2^k$ and we thus take $\tau=\frac{1}{2^k}\cdot O(\frac{1}{\log{\frac{k}{\eps}\log{mk}}})$. It then follows that 
\begin{align*}
    48(\log{T(2s,\frac{\tau}{4})} + \log{\frac{8}{\tau}})
    &\lesssim k\cdot d \cdot \log\frac{1}{\eps} + k\log s + k \log\log\frac{1}{\tau} + \log\frac{1}{\tau} \\
    &\lesssim k\cdot d \cdot \log\frac{1}{\eps} + k\log s + k\log k + \log\left(\log\frac{k}{\eps}\log(mk)\right)\\
    &\lesssim \eps^2 s
\end{align*}
provided that 
\[
    s \gtrsim \frac{kd_{D}}{\eps^2} \log \frac{kd_{D}}{\eps} + \frac{1}{\eps^2}\log \log m.
\]
This implies that $\alpha \lesssim \eps$. By a union bound, we see that with probability $1-O(\frac{1}{\log(k/\eps)\log{mk}})$, for every $J\subseteq [k_0]$, $C\in \binom{M}{k}$ and each $r>0$, $U_J$ is an $O(\eps)$-range space approximation of $X_J$ w.r.t. $\delta$. Rescaling $\eps$ completes the proof.
\end{proof}


\paragraph{Indexed-subset cost approximation.} We also need to prove that Definition~\ref{def_subset} is satisfied for doubling metrics. The proof follows essentially the same approach as in Section~\ref{sec_proof_subset}, with the only difference being the net size bound in Lemma~\ref{lemma_netsize_main}. In fact, the net size is much smaller for doubling metrics, as shown by the following lemma.

\begin{lemma}[Bounding net size on doubling instances] \label{lemma_netsize_doubling}
    For every $i\in [T]$, $\alpha\in (0,1)$ and every subset $S\subseteq X $ with $\abs{S}\geq 2$, there exists an $\alpha$-net $N_{\alpha,i}^{S}$ for $V_{i}^S$ such that
    \[
        \log {\abs{N_{\alpha,i}^S}}\lesssim k d \log (k/\alpha).
    \]
    Furthermore, if $\alpha \geq 2^{-i}$, it is possible to ensure that $N_{\alpha,i}^S$ consists of piecewise constant vectors. That is, every $v\in N_{\alpha,i}^S$ satisfies that $v(x)$ is constant for all $x\in X_j$ for each $j\in [k_0]$.
\end{lemma}

\begin{proof}
    For each $i\in \{0,\dots,T\}$, let $\alpha_i=\alpha  2^{i+1} r$. Let $N_{i,j}$ be an $\alpha_i$-net of $B(a_j,2^{i+1}r)\setminus B(a_j,2^i r)$. It follows from the definition of doubling dimension that
    \[
    	\abs{N_{i,j}} \leq (\frac{2^{i+1} r}{\alpha_i})^{d} \leq (\frac{k}{\alpha})^{d}.
    \]
    Fix $J\in [k_0]$ and consider a distance vector $v_{i,J,C}^S$, where $C = \{c_1,\dots,c_k\}$. For each $j\in J$, there exists $\bar{c}_j\in N_{i,j}$ such that $\dist(c_j, \bar{c}_j)\leq \alpha_i$. Define $\bar{v}_{i,J,C}:S\to \R_{\geq 0}$ as
    \[
    	\bar{v}_{i,J,C}(x) = 
    	\begin{cases}
    		\dist(x,\bar{C}),	& v_{i,J,C}^S(x) > 0\\
    		0,				 	& v_{i,J,C}^S(x) = 0.
    	\end{cases}
    \]
    This implies that 
    \[
    	\abs{v_{i,J,C}^S(x) - \bar{v}_{i,J,C}^S(x)} \leq \max_{j} \dist(c_j, \bar{c}_j) \leq \alpha_i.
    \]
    for all $x\in S$. Let $N_{i,J}$ denote the set of all possible vectors $\bar{v}_{i,J,C}$, then $N_{i,J}$ is an $\alpha$-net of $V_{i,J}^S := \{v_{i,J,C}^S \mid C\in \binom{M}{k}\}$. The number of vectors in $N_{i,J}$ is at most the number of possible center sets $\bar{C}$ (since $\bar{C}$ determines $\bar{v}_{i,J,C}(x)$ for all $x\in S$), hence
    \[
    	\abs{N_{i,J}} \leq  \left(\frac{k}{\alpha}\right)^{d k}\leq \left(\frac{k}{\alpha}\right)^{d k}.
    \]
    Observe that that $N_{\alpha,i}^S := \bigcup_J N_{i,J}$ is an $\alpha$-net of $V_{i}^S = \bigcup_J V_{i,J}^S$. Since there are at most $2^{k_0}$ subsets $J$, it follows that
    \begin{gather*}
    	\abs{ N_{\alpha,i}^S } \leq 2^{k_0} \cdot \left(\frac{k}{\alpha}\right)^{d k} \leq 2^{k} \cdot \left(\frac{k}{\alpha}\right)^{d k}
    \shortintertext{and}
    	\log \abs{ N_{\alpha,i}^S }  \lesssim k + kd \log{\frac{k}{\alpha}} \lesssim k d \log\frac{k}{\alpha}. 
    \end{gather*}
    Next, we show that we can choose a net consisting of piecewise constant vectors when $\alpha\geq 2^{-i}$. For $u\in N_{\alpha,i}^S$, define a vector $\bar{u}$ as
    \[
        \bar{u}(x) = \begin{cases}
                    \dist(a_{\pi(x)}, \bar{C}),    & u(x) > 0 \\
                    0,                      & u(x) = 0.
                  \end{cases}
    \]
    Then for a vector $v_{i,J,C}^S$, there exists $u\in N_{\alpha, i}^S$ such that $\abs{v_{i,J,C}(x) - \bar{u}(x)} \leq \alpha_i$. It follows that 
    \[
        \abs{v_{i,J,C}(x) - \bar{u}(x)} \leq \abs{v_{i,J,C}(x) - u(x)} + \abs{u(x) - \bar{u}(x)} \leq \alpha_i + r \leq \alpha2^{i+2}r.
    \]
    If $v_{i,J,C}^S(x)=0$, we know that $u(x)=0$, then $\bar{u}(x)=0$ by construction.
    Therefore, $\bar{N}_{\alpha,i}^S=\{\bar{u}\mid u\in N_{\alpha,i}^S\}$ is an $\alpha$-net of $V_{i,J,C}^S$ and every vector in  $\bar{N}_{\alpha,i}^S$ is a constant in each cluster. Furthermore,
    \[
        \log \abs{\bar{N}_{\alpha,i}^S} \leq \log |N_{\alpha,i}^S|\lesssim kd\log (k/\alpha). \qedhere
    \]

\end{proof}


\section{Proof of Lemma~\ref{lemma_Euc}: Third Coreset for Euclidean instances}

\label{sec_euc}


Recall that $X_I = \bigcup_{i\in I} X_i$.
Similar to Section~\ref{sec:alg_guarantee_proof}, we let $D_i=D\cap X_i$ for each $i\in [k]$, and $U_I = U \cap X_I$ and $D_I = D \cap X_I$ for each $I\subseteq [k]$.

We need the following adapted version of indexed-subset cost approximation (Definition \ref{def_subset}).

\begin{definition}[Strong indexed-subset cost approximation]\label{def_strong_subset} $S$ is an $\eps$-strong indexed-subset cost approximation of $X$ if for every $J\subset [k]$ and $C\in \binom{\mathbb{R}^d}{k}$,
\begin{eqnarray}\label{inq_strong_subset}
|\cost(S_J,C)-\cost(X_J,C)|\leq \eps\cdot (\cost(X_J,C)+\cost(X,A)).
\end{eqnarray}
\end{definition}

Note that we replace the term \( \eps n r \) in Definition \ref{def_subset} with \( \eps \cost(X, A) \).  
Since \( X \subset B(A, r) \), it follows that \( \cost(X, A) \leq n r \), making this adapted definition more stringent.



The following are the two main technical lemmata for the proof of Lemma~\ref{lemma_Euc}. Their proofs are deferred to subsequent subsections.

\begin{lemma}[General sufficient condition for robust coreset]
\label{lem_euc_coreset_guarantee}
Let $X$ be an $(r,k)$-instance. Suppose $D\subset X$ is capacity-respecting (Definition~\ref{def_cap}) and is an $\eps$-strong indexed-subset cost approximation of $X$, then $D$ is an $(\eps,m,6mr + \eps\cost(X,A))$-robust coreset of $X$.
\end{lemma}

We note that Lemma \ref{lem_euc_coreset_guarantee} holds for any metric space, including Euclidean spaces, and may be of independent research interest.

\begin{lemma}[Performance analysis of Algorithm \ref{alg_coreset3}]
\label{lemma_Euc_tech}
    For every $(r,k)$-instance $X=X_1\cup\cdots\cup X_k$, with probability at least $0.9$, Algorithm~\ref{alg_coreset3} computes a weighted subset $D\subset X$ in $O(nk)$ time, such that $D$ is capacity-respecting and an $O(\eps)$-strong indexed-subset cost approximation of $X$.
\end{lemma}

Lemma~\ref{lemma_Euc} now follows immediately.

\begin{proof}[Proof of Lemma~\ref{lemma_Euc}]
By combining Lemmata~\ref{lem_euc_coreset_guarantee} and~\ref{lemma_Euc_tech}, we conclude that the weighted set $D$ returned by Algorithm~\ref{alg_coreset3} is an $(O(\eps),m,6mr + O(\eps)\cdot\cost(X,A))$-robust coreset with probability at least $0.9$. The size of $D$ is guaranteed by the algorithm, prescribed by Line~\ref{alg_coreset3_size}. Rescaling $\eps$ completes the proof of correctness. Finally, it is clear that the algorithm runs in $O(nk)$ time.
\end{proof}

\begin{remark}[Comparison with \cite{jiang2025coresets}]
\label{remark:comparison}
Recall that \cite{jiang2025coresets} provided two coreset sizes for Euclidean robust \kMedian: $O(km\eps^{-1}) + (\text{vanilla size})$ and $\tilde{O}(m\eps^{-2}) + (\text{Vanilla size})$. We compare our approach against both bounds separately below.

To achieve the first bound \(O(km\eps^{-1}) + \text{(vanilla size)}\), \cite{jiang2025coresets} first decomposes the dataset into \(k\) dense components along with \(O(km\eps^{-1})\) residual points.  
They build a vanilla coreset for the union of dense components and retain all residual points, yielding a total size of \(O(km\eps^{-1}) + \text{(vanilla size)}\).  
Their analysis must accommodate up to \(m\) outliers in each of the \(k\) dense components, effectively handling \(mk\) outliers. 
In contrast, we apply a single vanilla‐reduction to the entire \((r,k)\)-instance, which only needs to manage \(m\) outliers in total.  
This joint treatment reduces the number of residual points from \(O(km\eps^{-1})\) down to \(O(m\eps^{-1})\).

To achieve the second bound \(\tilde{O}(m\eps^{-2}) + \text{(vanilla size)}\), \cite{jiang2025coresets} first partitions the dataset into \(\Gamma = \tilde{O}(m\eps^{-2})\) radius‐bounded components, then builds a vanilla coreset for the \ProblemName{$(k + \Gamma)$‐Medians} problem, using specific geometric properties of the decomposition rather than generic vanilla coreset algorithms.
%
%
They then show that the vanilla coreset is robust by leveraging a specific sparsity condition satisfied by these components (see Lemma 4.3 and Definition 4.4 of~\cite{jiang2025coresets}).
In contrast, our decomposition produces at most 
$
\Gamma' = k + O\bigl(m\eps^{-1}\bigr)
$
radius‐bounded components---namely the \(k\) balls in the \((r,k)\)-instance plus \(O(m\eps^{-1})\) singleton points. However, these components do not satisfy the sparsity condition of \cite{jiang2025coresets}, so our coreset \(S\), produced by Algorithm~\ref{alg_coreset3}, is not a vanilla \ProblemName{$(k + \Gamma')$-Medians} coreset. This prevents the use of their vanilla reduction argument on $S$. 
Instead, we exploit the uniform‐radius structure of \((r,k)\)-instances to prove directly that \(S\) is a robust coreset. 
By reducing the number of radius‐bounded components by a factor of \(\eps^{-1}\), we shrink the residual component size from \(O(m\eps^{-2})\) down to \(O(m\eps^{-1})\).
\end{remark}

\subsection{Proof of Lemma~\ref{lem_euc_coreset_guarantee}: General sufficient condition}

First, we present the following structural lemma, which approximately decomposes \( \cost^{(t)}(X,C) \) and \( \cost^{(t)}(D,C) \) into their respective vanilla costs.

\begin{lemma}[Robust cost decomposition]
\label{proof_Euc_order} 
If $D$ is a capacity-respecting weighted subset of $X$, we have that
	for every $C\in \binom{\mathbb{R}^d}{k}$, and every $t\in [0,m]$, there exists an index $j_t^C\in [k]$, an index subset $I_t^C\subseteq [k]$, and a constant $\alpha_t^C\in [0,1]$ such that
	\begin{gather} 
		|\cost^{(t)}(X,C)- \alpha_t^C\cost(X_{j_t^C},C)-\cost(X_{I_t^C},C)|\leq 2mr \label{inq_order_X} \\
		|\cost^{(t)}(D,C)- \alpha_t^C\cost(D_{j_t^C},C)-\cost(D_{I_t^C},C)|\leq 2mr \label{inq_order_D}
	\end{gather}
\end{lemma}
\begin{proof}
	Fix $C\in \binom{\mathbb{R}^d}{k}$ and $t\in [0,m]$. We sort $\{a_1,\dots,a_k\}$ by their distances to $C$, denoted by
	$$
	\dist(a_{i_1},C)\geq \dist(a_{i_2},C)\geq\cdots\geq \dist(a_{i_k},C),
	$$ where ties are broken arbitrarily. Let $q\in [k]$ denote the unique integer such that $$\sum_{j=1}^{q-1} |X_{i_j}|<t\leq \sum_{j=1}^q |X_{i_j}|.$$ 
	Let $\alpha_t^C = (\sum_{j=1}^q |X_{i_j}|-t)/|X_{i_q}|$, $j_t^C=q$ and $I_t^C=\{i_{q+1},\dots,i_k\}$. It remains to prove \eqref{inq_order_X} and \eqref{inq_order_D} for these parameters.
	
	Let $L\subseteq X$ with $\|L\|_1=t$ denote the outlier set excluded by $\cost^{(t)}(X,C)$. For every $\beta\in (0,t]$ define the $\beta$-cut of $L$ (with respect to $C$) as, 
	$$
	\mathrm{Cut}(L,C,\beta)=\max_{R\subseteq L,\|R\|_1=\beta} \min_{x\in R} \dist(x,C).
	$$ 
	Note that a $\beta$-cut can be interpreted as the fractional extension of the $\beta$-th largest distance for integral $\beta$.
	
	We claim that
	\[
	\mathrm{Cut}(L,C,\beta)\leq \frac{\cost(X_{i_l},C)}{|X_{i_l}|}+2r \quad\text{when}\quad \sum_{j=1}^{l-1}|X_{i_j}|<\beta\leq \sum_{j=1}^{l}|X_{i_j}|.
	\]
	To see this, note we would otherwise have 
	\[
	\mathrm{Cut}(L,C,\beta)> \frac{\cost(X_{i_l},C)}{|X_{i_l}|}+2r>\dist(a_{i_l},C)+r\geq \max_{x\in \bigcup_{j=l}^k X_{i_j}} \dist(x,C).
	\]
	Thus, there exists an $R\subseteq L$ with $|R|=\beta$ such that $R$ does not intersect $\bigcup_{j=l}^k X_{i_j}$, so $|R|\leq \sum_{j=1}^{l-1} |X_{i_j}|<\beta$, which is a contradiction.
	
	It then follows that
	\begin{align*}
		\cost(L,C) 
		&\leq \int_{0}^t \mathrm{Cut}(L,C,\beta) d\beta\\
		&\leq \sum_{j=1}^{q-1} |X_{i_j}| \left(\frac{\cost(X_{i_j},C)}{|X_{i_j}|}+2r\right) + (1 - \alpha_t^C) |X_{i_q}| \left(\frac{\cost(X_{i_q},C)}{|X_{i_q}|}+2r\right)\\
		&\leq \cost\left(\bigcup_{j=1}^{q-1} X_{i_j}, C\right) + (1-\alpha_t^C)\cost(X_{i_q},C)+2mr,
	\end{align*}
	where the second inequality is due to the fact that $\sum_{j=1}^{q-1} |X_{i_j}|+(1-\alpha_t^C)|X_{i_q}|=t$. By the definition of $\cost^{(t)}(X,C)$, we also have
	\[
	\cost(L,C)\geq \sum_{i=1}^{q-1} \cost(X_{i_j},C)+(1-\alpha_t^C)\cost(X_{i_q},C).
	\]
	Since $\cost^{(t)}(X,C)=\cost(X,C)-\cost(L,C)$, it follows from the upper and lower bounds of $\cost(L,C)$ that
	\begin{align*}
		\cost^{(t)}(X,C) &\geq \alpha_t^C \cost(X_{j_t^C},C)+\cost(X_{I_t^C},C)-2mr \\ 
		\cost^{(t)}(X,C) &\leq \alpha_t^C \cost(X_{j_t^C},C)+\cost(X_{I_t^C},C),
	\end{align*}
	establishing \eqref{inq_order_X}. Since $D$ is capacity-respecting, \eqref{inq_order_D} can be proved in the same manner with the same set of parameters.
\end{proof} 

Lemma~\ref{lem_euc_coreset_guarantee} now follows easily from Lemma~\ref{proof_Euc_order}.
\begin{proof}[Proof of Lemma~\ref{lem_euc_coreset_guarantee}]
	By Lemma~\ref{proof_Euc_order}, the triangle inequality and the assumption that $D$ is an $\eps$-strong indexed-subset cost approximation of $X$, we have that for every $t\in [0,m]$ and $C\in \binom{\R^d}{k}$,
	\begin{align*}
		&\quad\ \abs{\cost^{(t)}(X,C)-\cost^{(t)}(D,C)}\\
		&\leq 2mr+2mr+\alpha_t^C\cdot |\cost(D_{j_t^C},C)-\cost(X_{j_t^C},C)|+|\cost(D_{I_t^C},C)-\cost(X_{I_t^C},C)|\\
		&\leq 4mr+\eps\cdot (\cost(X_{I_t^C},C)+\alpha_t^C\cost(X_{j_t^C},C))+\eps\cdot \cost(X,A)\\
		&\leq \eps\cdot \cost^{(m)}(X,C)+6mr+\eps\cdot \cost(X,A). \qedhere
	\end{align*}
\end{proof}


\subsection{Proof of Lemma~\ref{lemma_Euc_tech}: Performance analysis of Algorithm \ref{alg_coreset3}} \label{sec_Euc_analysis}

Recall that $(U, w_U)$ is a vanilla coreset constructed by the importance sampling procedure of \cite{bansal2024sensitivity} (Lines 1--8 of Algorithm \ref{alg_coreset3}).
Define two events
\begin{align*}
\mathcal{E}_0 &=``\text{$u_i\in (1\pm \eps)\cdot |X_i|$ for all $i\in [k]$''} \\
\mathcal{E}_1 &=``\text{$\cost(U\cap X_i,A)\in (1\pm \eps)\cdot \cost(X_i,A)$ for all $i\in [k]$''}
\end{align*}

The following lemma suggests that both events are likely to occur.

\begin{lemma}[Adaption of {\cite[Lemma~2.17]{bansal2024sensitivity}}]
\label{lm:event_prob}
When $s\gtrsim k\eps^{-2}\log(k\eps^{-1})$, it holds that $$
\mathbb{P}[\mathcal{E}_0\cup\mathcal{E}_1]\geq 1-\frac{\eps}{k}.
$$
\end{lemma}

The following lemma connects $U$ and $D$.

\begin{lemma}[Connection between $U$ and $D$] 
\label{lemma_Euc_UD}
If $\mathcal{E}_0$ happens and $U$ is an $\eps$-strong indexed-subset cost approximation of $X$, then $D$ is capacity-respecting and an $O(\eps)$-strong indexed-subset cost approximation of $X$.
\end{lemma}

\begin{proof}

By Line~\ref{alg3_scale} of Algorithm~\ref{alg_coreset3}, $D$ is capacity-respecting.

Recall that when $\mathcal{E}_0$ happens, we have $w_D(x)\in (1\pm x)\cdot w_U(x)$ for each $x\in U$. Thus, $\cost(D_I,C)\in (1\pm \eps)\cdot \cost(U_I,C)$ for each $I\subseteq [k]$ and $C\in \binom{\mathbb{R}^d}{k}$. 
It follows that
\begin{align*}
    \cost(D_I,C) - \cost(X_I,C) &\leq (1+\eps)(\cost(U_I,C) - \cost(X_I,C)) + \eps \cost(X_I,C)\\
    &\leq (1+\eps)(\cost(X_I,C) + \cost(X,A)) + \eps\cost(X_I,C) \\
    &\leq (1+2\eps)(\cost(X_I,C) + \cost(X,A))
\end{align*}
and similarly that $\cost(D_I,C) - \cost(X_I,C)\geq (1-2\eps)(\cost(X_I,C) + \cost(X,A))$. The proof is complete.
\end{proof}

Lemma~\ref{lemma_Euc_UD} implies that it suffices to prove that $U$ is an $\eps$-strong indexed-subset cost approximation of $X$, namely, (\ref{inq_strong_subset}) holds when $S = U$. See Lemma \ref{lemma_Euc_U_subset}.

Recall that for each $x\in X$, $\pi(x)$ is the unique index $i\in [k]$ such that $x\in X_i$. 
Also recall that for each $i\in [k]$, $\Delta_i = \frac{\cost(X_i,A)}{|X_i|}$ denote the average cost of $X_i$.
By the definition of $p(x)$ and $w_U(x)=\frac{1}{s\cdot p(x)}$, we have the following lemma.

\begin{lemma}[Fact B.1 of \cite{bansal2024sensitivity}]\label{lemma_Euc_weight} 

For every $x\in U$,
$$w_U(x)\leq \frac{4}{s}\cdot \min\bigg(\frac{\cost(X,A)}{\dist(x,A)},k|X_{\pi(x)}|,\frac{\cost(X,A)}{\Delta_{\pi(x)}},\frac{k\cost(X_{\pi(x)},A)}{\dist(x,A)}\bigg).$$
\end{lemma}

We need the following two-sided bounds for $\cost(U_i,C)$, which are derived from the triangle inequality.

\begin{lemma}[Two-sided bounds for $\cost(U_i,C)$]
\label{lemma_Euc_useful_triangle}
For each $i\in [k]$,
\begin{align}
    \cost(U_i,C) &\leq \cost(U_i,a_i)+\frac{u_i}{|X_i|}\cdot (\cost(X_i,a_i)+\cost(X_i,C)) \label{inq_Euc_tri}\\
                &\leq 4\cost(X,A)+4k\cdot \cost(X_i,a_i)+4k\cdot \cost(X_i,C), \label{inq_tri_worst}
\shortintertext{and}
    \cost(U_i,C) &\geq \frac{u_i}{|X_i|}(\cost(X,C)-\cost(X_i,a_i))-\cost(U_i,a_i).\label{inq_Euc_tri_reverse}
\end{align}
\end{lemma}
\begin{proof}
We consider the following process of sending $U_i$ to $C$: first, move all points in $U_i$ to $a_i$, incurring a cost of $\cost(U_i,a_i)$; next, uniformly distribute $u_i$ points to $x\in X_i$, incurring a cost of $\frac{u_i}{|X_i|}\cdot \cost(X_i,a_i)$; finally, send points from $X_i$ to $C$, with a cost of
$\frac{u_i}{|X_i|}\cdot \cost(X,C)$. Thus, \eqref{inq_Euc_tri} follows from the triangle inequality. Similarly, \eqref{inq_Euc_tri_reverse} follows from applying the triangle inequality to the reversed process.

Additionally, by Lemma~\ref{lemma_Euc_weight}, we know that 
\begin{gather*}
\cost(U_i,a_i)=\sum_{x\in U_i}w_U(x)\|x-a_i\|_2\leq \sum_{x\in U_i} \frac{4}{s}\cdot \cost(X,A)\leq 4\cost(X,A),
\shortintertext{and}
u_i=\sum_{x\in U_i} w_U(x)\leq \sum_{x\in U_i}\frac{4}{s}\cdot k |X_i|\leq 4k|X_i|.
\end{gather*}
Then \eqref{inq_tri_worst} follows.
\end{proof}

Similar to \cite{bansal2024sensitivity}, we define bands, levels, and interaction numbers. However, a major difference is that we define bands and levels for every subset $J\subseteq [k]$, rather than partitioning $[k]$ into subsets with the same band and level as in \cite{bansal2024sensitivity}. This modification is necessary since (\ref{inq_strong_subset}) requires approximation on every $J\subseteq[k]$.

\paragraph{Bands}
Let $T=\frac{\eps}{k}\cdot \cost(X,A)$, $b_{max}=\lceil\log \frac{k}{\eps}\rceil$. For $j\in [b_{max}]$, a subset $J\subseteq [k]$ is called a Band-$j$ subset if $\cost(X_i,A)\in (2^{j-1}T,2^jT]$ for all $i\in J$. In addition, we call $J$ a Band-$0$ subset if $\cost(X_i,A)\leq T$ for all $i\in J$. For $j=0,\dots,b_{max}$, let $\mathcal{B}_j=\{J\subseteq [k]\mid J\text{ is Band-}j\}$ denote the collection of all Band-$j$ subsets.

Note that the maximal subset $J\in \mathcal{B}_j$ contains all $i\in [k]$ such that $\cost(X_i,A)\in (2^{j-1}T,2^jT]$, and $\mathcal{B}_j$ is exactly the power set of this $J$.

\paragraph{Levels} Let $l_{max}=\lceil\log \eps^{-1} \rceil$.
For any subset $J\subseteq[k]$, any $k$-center $C\in \binom{\mathbb{R}^d}{k}$, and any $j\in [l_{max}]$, $C$ is called a level-$j$ center of $J$ if $\dist(a_i,C)\in (2^{j-1}\Delta_i,2^j \Delta_i]$ for all $i\in J$. 
In addition, $C$ is called a level-$0$ center of $J$ if $\dist(a_i,C)\leq \Delta_i$ for all $i\in J$. Let $L^J_j\subset \binom{\mathbb{R}^d}{k}$ denote the set of all level-$j$ centers of $J$. Furthermore, let $L^{J}_{-1}=\{C\in \binom{\mathbb{R}^d}{k}\mid \forall i\in J,\ \dist(a_i,C)>\eps^{-1} \Delta_i\}$.

Intuitively, $\frac{\cost(X_J, C)}{\cost(X_J, A)} \approx 2^j$ if $C\in L_j^J$.
If $C\in L_{-1}^J$, we note that $\frac{\cost(X_J, C)}{\cost(X_J, A)} \gtrsim \eps^{-1}$.

Next, we provide three lemmata to handle errors induced by a subset $J$ of clusters (Lemmata \ref{lemma_EUC_chaining_minorcase0}--\ref{lemma_Euc_chaining_main}).
We consider three cases according to the scale of $\cost(X_J, A)$ or $\cost(X_J, C)$, respectively.

\begin{lemma}[Handling error when $\cost(X_J, A)$ is small]
\label{lemma_EUC_chaining_minorcase0}
\begin{eqnarray*}
\E\sup_{J\in \mathcal{B}_0}\sup_{C\in \binom{\mathbb{R}^d}{k}} \bigg|\frac{\cost(U_J,C)-\cost(X_J,C)}{\cost(X_J,C)+\cost(X,A)}\bigg|\lesssim \eps.
\end{eqnarray*}
\end{lemma}
\begin{proof}
Fix $J\in \mathcal{B}_0$ and $C\in \binom{\mathbb{R}^d}{k}$.

If $\mathcal{E}_0\cup\mathcal{E}_1$ happens, by \eqref{inq_Euc_tri} in  Lemma~\ref{lemma_Euc_useful_triangle}, we have  for each $i\in J$,
\begin{align*}
\cost(U_i,C)&\leq \cost(U_i,a_i)+\frac{u_i}{|X_i|}\cdot (\cost(X_i,a_i)+\cost(X_i,C))\\
&\leq
(1+\eps)  \cost(X_i,a_i)+(1+\eps)(\cost(X_i,a_i)+\cost(X_i,C))\\
&\leq (1+\eps) \cost(X_i,C)+O(\frac{\eps}{k})\cdot \cost(X,A).
\end{align*}
Similarly, by \eqref{inq_Euc_tri_reverse} in Lemma~\ref{lemma_Euc_useful_triangle}, we have
$$
\cost(U_i,C)\geq (1-\eps) \cost(X_i,C)-O(\frac{\eps}{k}) \cdot \cost(X,A).
$$
Summing over $i\in J$, we obtain that when $\mathcal{E}_0\cup\mathcal{E}_1$ happens,
$$
\abs*{ \frac{\cost(U_J,C)-\cost(X_J,C)}{\cost(X_J,C)+\cost(X,A)} }\lesssim \eps.
$$
On the other hand, when $\mathcal{E}_0\cup\mathcal{E}_1$ does not happen, we know from \eqref{inq_tri_worst} that
$$
\abs*{ \frac{\cost(U_J,C)-\cost(X_J,C)}{\cost(X_J,C)+\cost(X,A)} }\lesssim k.
$$
Therefore,
\begin{align*}
 \E\sup_{J\in \mathcal{B}_0}\sup_{C\in \binom{\mathbb{R}^d}{k}} \bigg|\frac{\cost(U_J,C)-\cost(X_J,C)}{\cost(X_J,C)+\cost(X,A)}\bigg|
&\lesssim  \eps\cdot \mathbb{P}(\mathcal{E}_0\cup\mathcal{E}_1)+k\cdot (1-\mathbb{P}(\mathcal{E}_0\cup\mathcal{E}_1))\\
&\lesssim  \eps+k\cdot\frac{\eps}{k}\\
&\lesssim \eps. \qedhere
\end{align*}
\end{proof}

\begin{lemma}[Handling error when $\cost(X_J, C)$ is large] 
\label{lemma_Euc_chaining_minorcases1}
\begin{eqnarray*}
\E\sup_{J\subseteq [k]}\sup_{C\in L^J_{-1}} \bigg|\frac{\cost(U_J,C)-\cost(X_J,C)}{\cost(X_J,C)+\cost(X,A)}\bigg|\lesssim \eps.
\end{eqnarray*}
\end{lemma}

\begin{proof}
Fix $J\subseteq[k]$ and $C\in L_{-1}^J$. By definition, we know that for each $i\in J$, 
\[
    \dist(a_i,C)\geq \eps^{-1}\Delta_i,
\]
which implies that
\[
    \cost(X_i,C)\geq |X_i|\dist(a_i,C)-\cost(X_i,a_i)\geq (\eps^{-1}-1)|X_i|\cdot \Delta_i\gtrsim \eps^{-1}\cdot \cost(X_i,a_i).
\]
Hence, when $\mathcal{E}_0\cup\mathcal{E}_1$ happens, it follows from Lemma~\ref{lemma_Euc_useful_triangle} that for each $i\in J$,
\begin{align*}
\cost(U_i,C)&\leq \cost(U_i,a_i)+\frac{u_i}{|X_i|}\cdot (\cost(X_i,a_i)+\cost(X_i,C))\\
&\leq
(1+\eps)  \cost(X_i,a_i)+(1+\eps)(\cost(X_i,a_i)+\cost(X_i,C))\\
&\leq
(1+O(\eps)) \cost(X_i,C).
\end{align*}
Similarly, by \eqref{inq_Euc_tri_reverse}, we have that
$$
\cost(U_i,C)\geq (1-O(\eps))\cdot \cost(X_i,C).
$$
Summing over $i\in J$, we obtain that when $\mathcal{E}_0\cup\mathcal{E}_1$ happens,
$$
\bigg|\frac{\cost(U_J,C)-\cost(X_J,C)}{\cost(X_J,C)+\cost(X,A)}\bigg|\lesssim \eps.
$$
On the other hand, when $\mathcal{E}_0\cup\mathcal{E}_1$ does not happen, we know from \eqref{inq_tri_worst} that
$$
\bigg|\frac{\cost(U_J,C)-\cost(X_J,C)}{\cost(X_J,C)+\cost(X,A)}\bigg|\lesssim k.
$$
Therefore,
\begin{align*}
 \E\sup_{J\subseteq[k]}\sup_{C\in L^J_{-1}} \bigg|\frac{\cost(U_J,C)-\cost(X_J,C)}{\cost(X_J,C)+\cost(X,A)}\bigg|
&\lesssim \eps\cdot \mathbb{P}(\mathcal{E}_0\cup\mathcal{E}_1)+k\cdot (1-\mathbb{P}(\mathcal{E}_0\cup\mathcal{E}_1))\\
&\lesssim \eps+k\cdot\frac{\eps}{k}\\
&\lesssim \eps. \qedhere
\end{align*}
\end{proof}

The remaining case is for a medium scale of $\cost(X_J,C)$.
In this case, we do a more careful partition to center sets $C$ according to so-called \emph{interaction numbers}.
This approach follows from \cite{bansal2024sensitivity}, which is useful for reducing the coreset size.

\paragraph{Interaction Numbers} Fix a subset $J\subseteq [k]$ and a $k$-center $C\in \binom{\mathbb{R}^d}{k}$. For each $c\in C$ and $i\in J$, we say center $c$ interacts with $i$ if $\|a_i-c\|_2\in [32\Delta_i,16\dist(a_i,C)]$. Let $N(J,C)$ denote the number of interacting pairs over $(i,c)\in J\times C$. Clearly, $N(J,C)\leq k|J|$. 

Let $\psi_{max}=\lceil2\log k\rceil$. For each $\psi \in [\psi_{max}]$, let $\mathcal{N}^J_\psi=\{C\in \binom{\mathbb{R}^d}{k}\mid N(J,C)\in [2^{\psi-1},2^\psi)\}$. 
Intuitively, as $\psi$ increases, the complexity of interaction between $J$ and $C$ also increases.
This complexity affects both the net size and the variance in a chaining argument; as shown in Lemmata \ref{lemma_Euc_variance2} and \ref{lemma_Euc_netsize}, respectively.

Moreover, we let $\mathcal{N}^J_0=\{C\in \binom{\mathbb{R}^d}{k}\mid N(J,C)=0\}$. 
If $C\in \mathcal{N}^J_0$, we have $\dist(a_i, C) = \min_{c\in C} \dist(a_i, c) \notin [32\Delta_i,16\dist(a_i,C)]$.
Hence, we have $\dist(a_i, C)\leq 32\Delta_i$, which implies that the location of $C$ is ``close to'' $A$.

Let $L_{j,\psi}^J := L_j^J\cap N_\psi^J$ for each $j=0,\dots,l_{max}$ and $\psi=0,\dots,\psi_{max}$.
We provide the following lemma based on $L_{j,\psi}^J$.

\begin{lemma}[Handling error when $\cost(X_J, C)$ is medium] 
\label{lemma_Euc_chaining_main}
For each $j\in [b_{max}]$, $l\in \{0,\dots,l_{max}\}$, and $\psi \in \{0,\dots,\psi_{max}\}$,
\begin{eqnarray} \label{inq_Euc_main}
\E\sup_{J\in \mathcal{B}_j}\sup_{C\in L^J_{l,\psi}} \bigg|\frac{\cost(U_J,C)-\cost(X_J,C)}{\cost(X_J,C)+\cost(X,A)}\bigg|\lesssim \frac{\eps}{\log^5\frac{k}{\eps}}.
\end{eqnarray}
\end{lemma}

The proof of this lemma is mostly an adaptation of \cite{bansal2024sensitivity}.
We defer the proof to Section~\ref{sec_Euc_chaining_main} for completeness. 
Now we are ready to prove the following lemma for $U$.

\begin{lemma}[Properties of $U$] 
\label{lemma_Euc_U_subset}
With probability at least $1-\frac{1}{10\log \frac{k}{\eps}}$, $U$ is an $O(\eps)$-strong indexed-subset cost approximation of $X$.
\end{lemma}
\begin{proof} For every $C$, we can partition $J$ into $O(\log^2\frac{k}{\eps})$ subsets as 
\[
    \left\{ J_{j,l}^C \;\middle\vert\; j=0,\dots,b_{max},l=-1,\dots,l_{max} \right\}
\]
where $J_{j,l}^C$ is the maximal subset $J'\subseteq J$ such that $J'\in \mathcal{B}_j$ and $C\in L_l^{J'}$. 
By Lemmata~\ref{lemma_EUC_chaining_minorcase0},~\ref{lemma_Euc_chaining_minorcases1},~\ref{lemma_Euc_chaining_main}, and Markov's inequality, we know that with probability at least $1-\frac{1}{10\log \frac{k}{\eps}}$, it holds for every choice of $J,C,j,l$ that
\[
|\cost(U_{J_{j,l}^C},C)-\cost(X_{J_{j,l}^C},C)|\lesssim \frac{\eps}{\log^2\frac{k}{\eps}} (\cost(X_{J_{j,l}^C},C)+\cost(X,A)).
\]
Summing over all $j$ and $l$ yields that
$$
|\cost(U_J,C)-\cost(X_J,C)|\lesssim \eps \cdot (\cost(X_J,C)+\cost(X,A)),
$$
which is the definition of $O(\eps)$-strong indexed-subset cost approximation.
\end{proof}

Lemma~\ref{lemma_Euc_tech} is a direct corollary of Lemmata~\ref{lemma_Euc_U_subset} and~\ref{lemma_Euc_UD}.

\bibliographystyle{alphaurl}
\bibliography{ref}

@article{bhaskara2019greedy,
  title={Greedy sampling for approximate clustering in the presence of outliers},
  author={Bhaskara, Aditya and Vadgama, Sharvaree and Xu, Hong},
  journal={Advances in Neural Information Processing Systems},
  volume={32},
  year={2019}
}

@inproceedings{huang2025coresets,
  title={Coresets for Constrained Clustering: General Assignment Constraints and Improved Size Bounds},
  author={Huang, Lingxiao and Li, Jian and Lu, Pinyan and Wu, Xuan},
  booktitle={Proceedings of the 2025 Annual ACM-SIAM Symposium on Discrete Algorithms (SODA)},
  pages={4732--4782},
  year={2025},
  organization={SIAM}
}

@inproceedings{NN19,
  author    = {Shyam Narayanan and
               Jelani Nelson},
  title     = {Optimal terminal dimensionality reduction in {E}uclidean space},
  booktitle = {{STOC}},
  pages     = {1064--1069},
  publisher = {{ACM}},
  year      = {2019}
}

@article{cohen2022improved,
  title={Improved Coresets for {E}uclidean $ k $-Means},
  author={Cohen-Addad, Vincent and Green Larsen, Kasper and Saulpic, David and Schwiegelshohn, Chris and Sheikh-Omar, Omar Ali},
  journal={Advances in Neural Information Processing Systems},
  volume={35},
  pages={2679--2694},
  year={2022}
}

@inproceedings{Cohen2022Towards,
  author    = {Vincent Cohen{-}Addad and
               Kasper Green Larsen and
               David Saulpic and
               Chris Schwiegelshohn},
  title     = {Towards optimal lower bounds for $k$-median and $k$-means coresets},
  booktitle = {{STOC}},
  pages     = {1038--1051},
  publisher = {{ACM}},
  year      = {2022}
}

@inproceedings{cohen2021new,
  author    = {Vincent Cohen{-}Addad and
               David Saulpic and
               Chris Schwiegelshohn},
  title     = {A new coreset framework for clustering},
  booktitle = {{STOC}},
  pages     = {169--182},
  publisher = {{ACM}},
  year      = {2021}
}

@inproceedings{DBLP:conf/stoc/FeldmanL11,
  author    = {Dan Feldman and
               Michael Langberg},
  title     = {A unified framework for approximating and clustering data},
  booktitle = {{STOC}},
  pages     = {569--578},
  publisher = {{ACM}},
  year      = {2011},
  note      = {\url{https://arxiv.org/abs/1106.1379}}
}

@inproceedings{huang2022near,
  author    = {Lingxiao Huang and
               Shaofeng H.{-}C. Jiang and
               Jianing Lou and
               Xuan Wu},
  title     = {Near-optimal Coresets for Robust Clustering},
  booktitle = {Proceedings of ICLR 2023},
  year      = {2023}
}

@article{wang2021robust,
  title={Robust and Fully-Dynamic Coreset for Continuous-and-Bounded Learning (With Outliers) Problems},
  author={Wang, Zixiu and Guo, Yiwen and Ding, Hu},
  journal={Advances in Neural Information Processing Systems},
  volume={34},
  pages={14319--14331},
  year={2021}
}

@inproceedings{braverman2022power,
  author    = {Vladimir Braverman and Vincent Cohen-Addad and Shaofeng Jiang and Robert Krauthgamer and Chris Schwiegelshohn and Mads Bech Toftrup and Xuan Wu},
  title     = {The Power of Uniform Sampling for Coresets},
  booktitle = {62nd {IEEE} Annual Symposium on Foundations of Computer Science, {FOCS} 2022},
  pages     = {},
  publisher = {{IEEE} Computer Society},
  year      = {2022}
}

@inproceedings{HJV19,
  author    = {Lingxiao Huang and
               Shaofeng H.{-}C. Jiang and
               Nisheeth K. Vishnoi},
  title     = {Coresets for Clustering with Fairness Constraints},
  booktitle = {NeurIPS},
  pages     = {7587--7598},
  year      = {2019}
}

@inproceedings{huang2018epsilon,
  author    = {Lingxiao Huang and
               Shaofeng H.{-}C. Jiang and
               Jian Li and
               Xuan Wu},
  title     = {Epsilon-Coresets for Clustering (with Outliers) in Doubling Metrics},
  booktitle = {2018 IEEE 59th Annual Symposium on Foundations of Computer Science (FOCS)}, 
  pages     = {814--825},
  publisher = {{IEEE} Computer Society},
  year      = {2018}
}

@inproceedings{huang2020coresets,
  author    = {Lingxiao Huang and
               Nisheeth K. Vishnoi},
  title     = {Coresets for clustering in {Euclidean} spaces: importance sampling is
               nearly optimal},
  booktitle = {{STOC}},
  pages     = {1416--1429},
  publisher = {{ACM}},
  year      = {2020}
}

@inproceedings{langberg2010universal,
  author    = {Michael Langberg and
               Leonard J. Schulman},
  title     = {Universal epsilon-approximators for Integrals},
  booktitle = {{SODA}},
  pages     = {598--607},
  publisher = {{SIAM}},
  year      = {2010}
}

@inproceedings{bansal2024sensitivity,
  title={Sensitivity Sampling for $ k $-Means: Worst Case and Stability Optimal Coreset Bounds},
  author={Bansal, Nikhil and Cohen-Addad, Vincent and Prabhu, Milind and Saulpic, David and Schwiegelshohn, Chris},
  booktitle={2024 IEEE 65th Annual Symposium on Foundations of Computer Science (FOCS)},
  pages={1707--1723},
  year={2024},
  organization={IEEE}
}

@inproceedings{feldman2012data,
  title={Data reduction for weighted and outlier-resistant clustering},
  author={Feldman, Dan and Schulman, Leonard J},
  booktitle={Proceedings of the twenty-third annual ACM-SIAM symposium on Discrete Algorithms},
  pages={1343--1354},
  year={2012},
  organization={SIAM}
}

@inproceedings{BJKW21,
  author    = {Vladimir Braverman and
               Shaofeng H.{-}C. Jiang and
               Robert Krauthgamer and
               Xuan Wu},
  title     = {Coresets for Clustering in Excluded-minor Graphs and Beyond},
  booktitle = {Proceedings of the 2021 ACM-SIAM Symposium on Discrete Algorithms (SODA)},
  pages     = {2679--2696},
  publisher = {{SIAM}},
  year      = {2021},
}

@inproceedings{charikar2001algorithms,
  title={Algorithms for facility location problems with outliers},
  author={Charikar, Moses and Khuller, Samir and Mount, David M and Narasimhan, Giri},
  booktitle={SODA},
  volume={1},
  pages={642--651},
  year={2001},
  organization={Citeseer}
}

@article{DBLP:journals/dcg/Har-PeledK07,
  author    = {Sariel Har{-}Peled and
               Akash Kushal},
  title     = {Smaller Coresets for $k$-Median and $k$-Means Clustering},
  journal   = {Discret. Comput. Geom.},
  volume    = {37},
  number    = {1},
  pages     = {3--19},
  year      = {2007}
}

@article{Chen09,
author = {Chen, Ke},
title = {On Coresets for $k$-Median and $k$-Means Clustering in Metric and {E}uclidean Spaces and Their Applications},
journal = {SIAM Journal on Computing},
volume = {39},
number = {3},
pages = {923-947},
year = {2009}
}

@inproceedings{NEURIPS2021_90fd4f88,
 author = {Braverman, Vladimir and Jiang, Shaofeng and Krauthgamer, Robert and Wu, Xuan},
 booktitle = {Advances in Neural Information Processing Systems},
 editor = {M. Ranzato and A. Beygelzimer and Y. Dauphin and P.S. Liang and J. Wortman Vaughan},
 pages = {17360--17372},
 publisher = {Curran Associates, Inc.},
 title = {Coresets for Clustering with Missing Values},
 url = {https://proceedings.neurips.cc/paper/2021/file/90fd4f88f588ae64038134f1eeaa023f-Paper.pdf},
 volume = {34},
 year = {2021}
}

@inproceedings{Schmidt19,
  author    = {Melanie Schmidt and
               Chris Schwiegelshohn and
               Christian Sohler},
  title     = {Fair Coresets and Streaming Algorithms for Fair $k$-means},
  booktitle = {{WAOA}},
  series    = {Lecture Notes in Computer Science},
  volume    = {11926},
  pages     = {232--251},
  publisher = {Springer},
  year      = {2019}
}

@inproceedings{DBLP:conf/icalp/BandyapadhyayFS21,
  author    = {Sayan Bandyapadhyay and
               Fedor V. Fomin and
               Kirill Simonov},
  title     = {On Coresets for Fair Clustering in Metric and {E}uclidean Spaces and
               Their Applications},
  booktitle = {{ICALP}},
  series    = {LIPIcs},
  volume    = {198},
  pages     = {23:1--23:15},
  publisher = {Schloss Dagstuhl - Leibniz-Zentrum f{\"{u}}r Informatik},
  year      = {2021}
}

@inproceedings{Cohen25,
  author    = {Vincent Cohen-Addad
               and 
               Andrew Draganov
               and 
               Matteo Russo
               and David Saulpic},
  title     = {A Tight {VC}-Dimension Analysis of Clustering Coresets with Applications},
  booktitle = {Proceedings of 36th {SODA}},
  publisher = {{ACM/SIAM}},
  year      = {2025}
}

@article{fefferman2016testing,
  title={Testing the manifold hypothesis},
  author={Fefferman, Charles and Mitter, Sanjoy and Narayanan, Hariharan},
  journal={Journal of the American Mathematical Society},
  volume={29},
  number={4},
  pages={983--1049},
  year={2016}
}

@article{eisenstat2007vc,
  title={The {VC} dimension of $k$-fold union},
  author={Eisenstat, David and Angluin, Dana},
  journal={Information Processing Letters},
  volume={101},
  number={5},
  pages={181--184},
  year={2007},
  publisher={Elsevier}
}

@inproceedings{huang2024optimal,
  title={On optimal coreset construction for {E}uclidean $(k, z)$-clustering},
  author={Huang, Lingxiao and Li, Jian and Wu, Xuan},
  booktitle={Proceedings of the 56th Annual ACM Symposium on Theory of Computing},
  pages={1594--1604},
  year={2024}
}

@inproceedings{jiang2025coresets,
  author       = {Shaofeng H.{-}C. Jiang and
                  Jianing Lou},
  title        = {Coresets for Robust Clustering via Black-Box Reductions to Vanilla
                  Case},
  booktitle    = {{ICALP}},
  series       = {LIPIcs},
  volume       = {334},
  pages        = {101:1--101:18},
  publisher    = {Schloss Dagstuhl - Leibniz-Zentrum f{\"{u}}r Informatik},
  year         = {2025}
}

@book{LT91,
    AUTHOR = {Michel Ledoux and Michel Talagrand},
     TITLE = {Probability in {B}anach spaces},
 PUBLISHER = {Springer-Verlag},
   ADDRESS = {Berlin},
      YEAR = {1991},
}

@book{vershynin2018high,
  title={High-Dimensional Probability: An Introduction with Applications in Data Science},
  author={Vershynin, Roman},
  volume={47},
  year={2018},
  publisher={Cambridge University Press}
}

@InProceedings{cohen2022fixed,
  author =	{Cohen-Addad, Vincent and Li, Jason},
  title =	{On the Fixed-Parameter Tractability of Capacitated Clustering},
  booktitle =	{46th International Colloquium on Automata, Languages, and Programming (ICALP 2019)},
  pages =	{41:1--41:14},
  series =	{Leibniz International Proceedings in Informatics (LIPIcs)},
  ISBN =	{978-3-95977-109-2},
  ISSN =	{1868-8969},
  year =	{2019},
  volume =	{132},
  editor =	{Baier, Christel and Chatzigiannakis, Ioannis and Flocchini, Paola and Leonardi, Stefano},
  publisher =	{Schloss Dagstuhl -- Leibniz-Zentrum f{\"u}r Informatik},
  address =	{Dagstuhl, Germany},
  URL =		{https://drops.dagstuhl.de/entities/document/10.4230/LIPIcs.ICALP.2019.41},
  URN =		{urn:nbn:de:0030-drops-106171},
  doi =		{10.4230/LIPIcs.ICALP.2019.41},
}

@inproceedings{chen2008constant,
  title={A constant factor approximation algorithm for $k$-median clustering with outliers},
  author={Chen, Ke},
  booktitle={Proceedings of the nineteenth annual ACM-SIAM symposium on Discrete algorithms},
  pages={826--835},
  year={2008}
}

@inproceedings{krishnaswamy2018constant,
  title={Constant approximation for $k$-median and $k$-means with outliers via iterative rounding},
  author={Krishnaswamy, Ravishankar and Li, Shi and Sandeep, Sai},
  booktitle={Proceedings of the 50th annual ACM SIGACT symposium on theory of computing},
  pages={646--659},
  year={2018}
}

@inproceedings{feng2019improved,
  title={Improved algorithms for clustering with outliers},
  author={Feng, Qilong and Zhang, Zhen and Huang, Ziyun and Xu, Jinhui and Wang, Jianxin},
  booktitle={Proc. 30th International symposium on algorithms and computation (ISAAC 2019)},
  year={2019}
}

@inproceedings{cohen2017input,
  title={Input sparsity time low-rank approximation via ridge leverage score sampling},
  author={Cohen, Michael B and Musco, Cameron and Musco, Christopher},
  booktitle={Proceedings of the Twenty-Eighth Annual ACM-SIAM Symposium on Discrete Algorithms},
  pages={1758--1777},
  year={2017},
  organization={SIAM}
}

@inproceedings{chhaya2020coresets,
  title={On coresets for regularized regression},
  author={Chhaya, Rachit and Dasgupta, Anirban and Shit, Supratim},
  booktitle={International conference on machine learning},
  pages={1866--1876},
  year={2020},
  organization={PMLR}
}

@article{lucic2018training,
  title={Training gaussian mixture models at scale via coresets},
  author={Lucic, Mario and Faulkner, Matthew and Krause, Andreas and Feldman, Dan},
  journal={Journal of Machine Learning Research},
  volume={18},
  number={160},
  pages={1--25},
  year={2018}
}

@article{huang2022coresets,
  title={Coresets for wasserstein distributionally robust optimization problems},
  author={Huang, Ruomin and Huang, Jiawei and Liu, Wenjie and Ding, Hu},
  journal={Advances in Neural Information Processing Systems},
  volume={35},
  pages={26375--26388},
  year={2022}
}

@article{feldman2020turning,
  title={Turning big data into tiny data: Constant-size coresets for $k$-means, pca, and projective clustering},
  author={Feldman, Dan and Schmidt, Melanie and Sohler, Christian},
  journal={SIAM Journal on Computing},
  volume={49},
  number={3},
  pages={601--657},
  year={2020},
  publisher={SIAM}
}

@article{gupta2024structural,
  title={Structural iterative rounding for generalized $k$-median problems},
  author={Gupta, Anupam and Moseley, Benjamin and Zhou, Rudy},
  journal={Mathematical Programming},
  pages={1--54},
  year={2024},
  publisher={Springer}
}

@article{agrawal2023clustering,
  title={Clustering what matters: Optimal approximation for clustering with outliers},
  author={Agrawal, Akanksha and Inamdar, Tanmay and Saurabh, Saket and Xue, Jie},
  journal={Journal of Artificial Intelligence Research},
  volume={78},
  pages={143--166},
  year={2023}
}

@inproceedings{MMWY22,
  title={Active Linear Regression for  $\ell_p$ Norms and Beyond},
  author={Musco, Cameron and Musco, Christopher and Woodruff, David P and Yasuda, Taisuke},
  booktitle={Proceedings of the 63rd {IEEE} Annual Symposium on Foundations of Computer Science},
  pages={744--753},
  year={2022},
  organization={IEEE}
}

@inproceedings{MO24:sensitivity,
author = {Munteanu, Alexander and Omlor, Simon},
title = {Optimal bounds for $\ell_p$ sensitivity sampling via $\ell_2$ augmentation},
year = {2024},
publisher = {JMLR.org},
booktitle = {Proceedings of the 41st International Conference on Machine Learning},
articleno = {1494},
numpages = {28},
location = {Vienna, Austria},
series = {ICML'24}
}

@inproceedings{AP24,
author = {Alishahi, Meysam and Phillips, Jeff M.},
title = {No dimensional sampling coresets for classification},
year = {2024},
publisher = {JMLR.org},
booktitle = {Proceedings of the 41st International Conference on Machine Learning},
articleno = {43},
numpages = {42},
location = {Vienna, Austria},
series = {ICML'24}
}

@inproceedings{WY23:sensitivity,
author = {Woodruff, David P. and Yasuda, Taisuke},
title = {Sharper bounds for $\ell_p$ sensitivity sampling},
year = {2023},
publisher = {JMLR.org},
booktitle = {Proceedings of the 40th International Conference on Machine Learning},
articleno = {1550},
numpages = {35},
location = {Honolulu, Hawaii, USA},
series = {ICML'23}
}

@inproceedings{LT25,
    title={Near-optimal Active Regression of Single-Index Models},
    author={Yi Li and Wai Ming Tai},
    booktitle={The Thirteenth International Conference on Learning Representations},
    year={2025},
}

@InProceedings{gajjar:COLT,
  title = 	 {Agnostic Active Learning of Single Index Models with Linear Sample Complexity},
  author =       {Gajjar, Aarshvi and Tai, Wai Ming and Xingyu, Xu and Hegde, Chinmay and Musco, Christopher and Li, Yi},
  booktitle = 	 {Proceedings of Thirty Seventh Conference on Learning Theory},
  pages = 	 {1715--1754},
  year = 	 {2024},
  editor = 	 {Agrawal, Shipra and Roth, Aaron},
  volume = 	 {247},
  series = 	 {Proceedings of Machine Learning Research},
  month = 	 {30 Jun--03 Jul},
  publisher =    {PMLR},
  url = 	 {https://proceedings.mlr.press/v247/gajjar24a.html},
}

\appendix
\section{Size lower bound for robust coreset}
\label{sec:lb}

The following theorem gives a size lower bound of $\max\{m,\;Q\}$ for robust coresets.
Since $\max\{m,\;Q\} \ge (m+Q)/2$, this also implies an $\Omega(m+Q)$ lower bound.

\begin{theorem}[\bf{Coreset lower bound}]
\label{thm:lower}
Let $Q(k,\varepsilon)$ denote a lower bound on the coreset size for vanilla $\kMedian$ with $k$ centers and error $\varepsilon$.
For any $k,\varepsilon,m$, there exists a dataset $X$ such that every $(\varepsilon,m)$-robust coreset of $X$ must have size at least $\max\{m,\;Q(k,\varepsilon)\}$.
\end{theorem}

\begin{proof}
If $m\ge Q(k,\varepsilon)$, then by Theorem~4.1 of \cite{huang2022near} there exists a dataset $X$ of size $m+1$ such that every $(\varepsilon,m)$-robust coreset of $X$ has size at least $m$.
Although \cite{huang2022near} states the construction for $k=1$, it extends to any $k$ by placing all $k$ centers at the same location.

If $m\le Q(k,\varepsilon)$, then by the definition of $Q(k,\varepsilon)$ there exists a dataset $X$ such that every $\varepsilon$-coreset for vanilla $\kMedian$ on $X$ has size at least $Q(k,\varepsilon)$.
By Definition~\ref{def:rbcoreset}, any $(\varepsilon,m)$-robust coreset is, in particular, an $\varepsilon$-vanilla coreset by taking $t=0$.
Thus, any $(\varepsilon,m)$-robust coreset of $X$ has size at least $Q(k,\varepsilon)$.

Combining the two cases completes the proof.
\end{proof}
\section{Proof of Lemma~\ref{lemma_subset}: Indexed-subset cost approximation}
\label{sec_proof_subset_appendix}

Recall that our instance is an $(r,k_0)$-regular instance, given by $X=X_1\cup\cdots\cup X_{k_0}$, where $X_i\subset B(a_i,r)$ for some $a_i\in M$. For every $x\in X$, let $\pi(x)$ denote the unique index $j\in [k_0]$ such that $x\in X_j$. Then $\dist(x,a_{\pi(x)})\leq r$.

Let $T=\lceil \log(1/\eps)\rceil$. We define $A_0=[0,2r]$, and $A_i=(2^i r,2^{i+1}r]$ for each $i\in [T]$. Let $I_{i,C}=\{j\in [k_0]\mid \dist(a_j,C)\in A_i\}$ for each $i\in \{0,\dots,T\}$.

For every $i\in \{0,\dots,T\}$, $J\subseteq [k_0]$, $C\in \binom{M}{k}$, and every subset $S\subseteq X$, the distance vector $v_{i,J,C}^S:S\rightarrow \mathbb{R}_{\geq 0}$ is defined as
\[
v^S_{i,J,C}(x)=
\begin{cases}
    \dist(x,C), & \pi(x)\in J\cap I_{i,C} \\
    0, & \text{otherwise}.
\end{cases}
\]

Define the level-$i$ distance vector set as $V^S_{i}=\{v^S_{i,J,C}\mid J\subseteq [k_0],C\in \binom{M}{k}\}$. The following lemma establishes the relationship between indexed-subset cost approximation and $V_i^U$.

\begin{lemma}[Sufficient chaining error for Lemma \ref{lemma_subset}]
\label{lemma_relate} 
If for every $i\in \{0,\dots,T\}$, 
\begin{equation} \label{inq_chaining}
\sup_{J,C} \frac{n}{s}\cdot \abs*{ \frac{\sum_{x\in U} v^U_{i,J,C}(x)-\E \sum_{x\in U} v^U_{i,J,C}(x)}{(1+2^i|J_{i,C}|/k_0)\cdot nr} }
\leq \frac{\eps}{\log(1/\eps)},
\end{equation}
where $J_{i,C}=J\cap I_{i,C}$,
and event $\mathcal{E}$ (defined in \eqref{eq_event}) happens, then \eqref{inq_subset} holds. 
In other words, $U$ is an $\eps$-indexed-subset cost approximation of $X$.
\end{lemma}

\begin{proof}
Let $I_{far,C}=[k_0]\setminus \bigcup_{i=0}^T I_{i,C}$ and $J_{\far,C}=J\cap I_{\far,C}$. We observe that
\begin{gather*}
    \cost(X_J,C)=\sum_{i=0}^T \cost(X_{J_{i,C}},C)+\cost(X_{J_{\far,C}},C), 
\shortintertext{and similarly}
    \cost(U_J,C)=\sum_{i=0}^T \cost(U_{J_{i,C}},C)+\cost(U_{J_{\far,C}},C).
\end{gather*}
We shall show termwise approximations. First, we show that $$\cost(U_{J_{\far,C}},C)\in (1\pm O(\eps))\cdot \cost(X_{J_{\far,C}},C).$$ To see this, note that for each $x\in X_{J_{\far,C}}$, $\dist(x,C)\geq \eps^{-1} r$. Thus for every $j\in J_{\far,C}$, all points in $X_j$ have the same distance to $C$, up to a multiplicative $1+O(\eps)$ factor. When $\mathcal{E}$ happens, we have $w(U\cap X_j)\in (1\pm \eps) |X_j|$ for each $j\in [k_0]$, which implies $\cost(U_{J_{\far,C}},C)\in (1\pm O(\eps)) \cost(X_{J_{\far,C}},C)$.

Now, for $i\in \{0,\dots,T\}$, observe that
\[
    \cost(U_{J_{i,C}},C)=\frac{n}{s}\cdot \sum_{x\in U} v^U_{i,J,C}(x)
    \quad\text{and}\quad
    \cost(X_{J_{i,C}},C)=\E\frac{n}{s}\cdot \sum_{x\in U} v^U_{i,J,C}(x).
\]
It follows from the assumption \eqref{inq_chaining} that
\[
\abs*{ \cost(U_{J_{i,C}},C)-\cost(X_{J_{i,C}},C) }\leq \frac{\eps}{\log \eps^{-1}}\cdot (1+2^i|J_{i,C}|/k_0)\cdot nr.
\]
Note that
\begin{gather*}
    \cost(X_{J_{i,C}},C)\geq |J_{i,C}|\cdot \frac{n}{2 k_0}\cdot 2^i r,
\shortintertext{thus}
    |\cost(U_{J_{i,C}},C)-\cost(X_{J_{i,C}},C)|\lesssim \frac{\eps}{\log \eps^{-1}}\cdot (\cost(X_{J_{i,C}},C)+nr)).
\end{gather*}
Since there are $T+1=O(\log \eps^{-1})$ many $J_{i,C}$'s, by the triangle inequality,
we have that
$$
|\cost(U_J,C)-\cost(X_J,C)|\lesssim \eps \cost(X_J,C)+\eps n r.
$$
Rescaling $\eps$ completes the proof.
\end{proof}

A subset $N\subseteq V^S_i$ is called an $\alpha$-net of $V^S_i$ if for every $i,J,C$, there exists a vector $v\in N$ such that
\[
\begin{cases}
|v(x)-v^S_{i,J,C}(x)|\leq \alpha\cdot 2^{i+2}\cdot r & v^S_{i,J,C}(x)>0\\
0 & v^S_{i,J,C}(x)=0
\end{cases}
\]

The following lemma is essentially an adaptation and extension of Lemma~5.3 of \cite{huang2025coresets} and Lemma~3.6 of \cite{Cohen25}. We include its proof for completeness.

\begin{lemma}[Bounding net size on VC instances] 
\label{lemma_netsize}
For every $i\in \{0,\dots,T\}$, $\alpha\in (0,1)$ and every subset $S\subseteq X$ with $\abs{S}\geq 2$, there exists an $\alpha$-net $N^S_{\alpha,i}$ of $V^S_i$ such that
\[
\log \abs[big]{N^S_{\alpha,i}} \lesssim \alpha^{-1}d k\log k \cdot \log\abs{S}.
\]
Furthermore, if $\alpha\geq 2^{-i}$, it is possible to ensure that $N_{\alpha,i}^S$ consists of piecewise constant vectors. That is, every $v\in N_{\alpha,i}^S$ satisfies that $v(x)$ is constant for all $x\in X_j$ for each $j\in [k_0]$.
\end{lemma}

\begin{proof}

Let $\beta=\alpha 2^{i} r$. For a distance vector $v_{i,J,C}^S\in V_i^S$, we define its $\alpha$-rounding to be 
$$
\bar{v}_{i,J,C}^S(x) =  \left\lceil \frac{v_{i,J,C}^S(x)}{\beta} \right\rceil \cdot \beta.
$$
It is clear that $\bar{v}_{i,J,C}^S(x) \in \{0,\beta,\dots, \beta t\}$ for $t=\lceil 2/\alpha \rceil$.

We show that the collection of all $\alpha$-roundings
\[
    N_{\alpha,i}^S := \left\{\bar{v}_{i,J,C}^S\mid v_{i,J,C}^S\in V_i^S\right\}
\] 
is the desired net. 
The construction yields immediately that for every $x\in S$, $|\bar{v}_{i,J,C}^S(x)-v_{i,J,C}^S(x)|\leq \beta\leq \alpha\cdot 2^i\cdot r$ and $\bar{v}_{i,J,C}^S(x)=0$ if $v_{i,J,C}^S(x)=0$, which implies that $N_{\alpha,i}^S$ is indeed an $\alpha$-net of $V_i^S$. 

It remains to bound the size of $N_{\alpha,i}^S$. Consider each $u\in N_{\alpha,i}^S$, define for each $j\in \{0,\dots,t\}$ a ball $B_j^S=\{x\in S\mid u(x)\leq \beta \cdot j\}$. Since all coordinates of $u$ are fully determined by the sequence $B_0^S,\dots,B_t^S$, it suffices to bound the number of possible balls $B_j^S$. 
By construction, $B_0^S=(C\cap S)\cup (S\setminus X_{I_{i,C}})$, thus, the number of possible choices for $B_0^S$ is at most 
\[
    \binom{\abs{S}}{k}\cdot 2^k = |S|^{O(k)}.
\]
For $i\geq 1$, write $B_i^S$ as $B_i^S=(B(C, \beta\cdot i)\cap S)\cap X_{I_{i,C}}$. Since the VC dimension of $\mathrm{Balls}_k$ is at most $O(dk\log k)$~\cite{eisenstat2007vc}, the number of distinct choices for $B(C, \beta\cdot i)\cap S$ is
at most $|S|^{O(dk\log k)}$. Additionally, there are at most $2^k$ choices of $X_{I_{i,C}}$, leading to at most \[
    |S|^{O(dk\log k)}\cdot 2^k = |S|^{O(dk\log k)}
\]
distinct possibilities for each $B_i^S$, for each $i\in [t]$.
Consequently, the different number of combinations of $B_0^S,\dots,B_t^S$ can be bounded by 
\[
    |S|^{O(k)}\cdot (|S|^{O(dk\log k)})^{t}=|S|^{O(\alpha^{-1}dk\log k)},
\]
which is an upper bound on the size of $N_{\alpha,i}^S$.

When $\alpha\geq 2^{-i}$, we show that every $u\in N_{\alpha,i}^S$ can be made constant in each cluster. Define 
\[
\bar{u}(x)=
\begin{cases}
    \lceil \frac{\dist(a_{\pi(x)},C)}{\beta}\rceil\cdot \beta & u(x)>0\\
    0 & u(x)=0
\end{cases}
\]
Since $|u(x)-v_{i,J,C}^S(x)|\leq \beta$ for all $x\in S$, it follows that 
$|\bar{u}(x)-v_{i,J,C}^S(x)|\leq r + 2\beta\leq \alpha\cdot 2^{i+2}\cdot r$. 
If $v_{i,J,C}^S(x)=0$, we know that $u(x)=0$, then $\bar{u}(x)=0$ by construction.
Therefore, $\bar{N}_{\alpha,i}^S=\{\bar{u}\mid u\in N_{\alpha,i}^S\}$ is an $\alpha$-net of $V_{i,J,C}^S$ and every vector in  $\bar{N}_{\alpha,i}^S$ is a constant in each cluster. Furthermore,
\[
 \abs{\bar{N}_{\alpha,i}^S} \leq |N_{\alpha,i}^S|\leq |S|^{O(\alpha^{-1}dk\log k)}. \qedhere
\]
\end{proof}

It remains to prove the following lemma.

\begin{lemma}[Bounding chaining error for Lemma~\ref{lemma_relate}] 
\label{lemma_chaining}
The bound \eqref{inq_chaining} holds for every $i\in \{0,\dots,T\}$ with probability at least $1-O\big(\frac{1}{\log \frac{k}{\eps}\cdot \log (mk)}\big)$.
\end{lemma}

\begin{proof}
It suffices to prove that for every $i\in\{0,\dots,T\}$,
\begin{eqnarray} \label{lemma_chaining_maininq}
    \E\sup_{J,C} \frac{1}{s}\cdot \abs*{ \frac{\sum_{x\in U} v^U_{i,J,C}(x)-\E \sum_{x\in U} v^U_{i,J,C}(x)}{(1+2^i|J_{i,C}|/k_0)\cdot r} }
    \lesssim \frac{\eps}{\log^2 \frac{1}{\eps}\cdot \log \frac{k}{\eps}\cdot \log mk},
\end{eqnarray}
then use Markov's inequality and take a union bound.

To ease the notation, let $\gamma_{i,J,C} = 1/(1+2^i|J_{i,C}|/k_0)$. The standard symmetrization trick gives that
\begin{equation}\label{eqn:chaining_error_symmetrization}
\E\sup_{J,C} \frac{\gamma_{i,J,C}}{rs}  \abs*{ \sum_{x\in U} v^U_{i,J,C}(x)-\E \sum_{x\in U} v^U_{i,J,C}(x) } 
\lesssim
\E_U \E_g \frac{1}{s} \sup_{J,C} \frac{\gamma_{i,J,C}}{r} \abs*{ \sum_{x\in U} g_x v^U_{i,J,C}(x) }
\end{equation}
For each $h\geq 0$, let $N_{2^{-h},i}^U$ be a ($2^{-h}$)-net of $V_i^U$ satisfying that
\[
    \log \abs[big]{N_{2^{-h},i}^U} \lesssim 2^h dk \log k \log s.
\]
which is possible by Lemma~\ref{lemma_netsize}. Furthermore, when $h = i$, one can ensure that $N_{2^{-i},i}^U$ consists of piecewise constant vectors.

By the net property, there exists a vector $b_{h,J,C}^U\in N_{2^{-h},i}^U$ such that
\[
|b_{h,J,C}^U(x)-v_{i,J,C}^U(x)|\leq 2^{i+2-h} r,\quad \forall x \in U.
\] 
By writing $v_{i,J,C}^U(x)$ as the following telescoping sum
\[
v_{i,J,C}^U(x)=\sum_{h=i}^{\infty} \big(b_{h+1,J,C}^U(x)-b_{h,J,C}^U(x)\big)+b_{i,J,C}^U(x),
\]
we have
\begin{align}
&\quad\ \E_U \E_g \frac{1}{s} \sup_{J,C} \frac{\gamma_{i,J,C}}{r}\abs*{ \sum_{x\in U} g_x v^U_{i,J,C}(x) } \notag \\
&\leq \E_U \frac{1}{s} \sum_{h=i}^\infty \E_g \sup_{J,C} \frac{\gamma_{i,J,C}}{r} \abs*{\sum_{x\in U} g_x (b_{h+1,J,C}^U(x)-b_{h,J,C}^U(x))} 
+ \E_U\frac{1}{s}\E_g \sup_{J,C} \frac{\gamma_{i,J,C}}{r} \abs*{\sum_{x\in U} g_x \cdot b_{i,J,C}^U(x)} \label{eqn:chaining_error_aux}
\end{align}

\sloppy
We bound the two terms individually, starting with the first term. Consider the sum $\sum_{x\in U} g_x \big(b_{h+1,J,C}^U(x)-b_{h,J,C}^U(x)\big)$. 
Observe that $|b_{h+1,J,C}^U(x)-b_{h,J,C}^U(x)|\lesssim 2^{i-h} r$ and, by the construction of $V_i^U$, there are at most $\|U\cap X_{J_{i,C}} \|_0$ nonzero terms in the summation. Thus, the supremum in the first term is taken over $\abs[\big]{N^U_{2^{-h},i}}$ gaussians, each with a variance at most 
\[
    \gamma_{i,J,C}^2 \cdot 2^{2i-2h}\|U\cap X_{J_{i,C}}\|_0 
    \leq 2^{2i-2h} \sup_{J,C} \gamma_{i,J,C}^2 \|U\cap X_{J_{i,C}}\|_0.
\]
By the inequality for the maximum of Gaussians (Fact~\ref{fact_max_gauss}), we conclude that 
\begin{equation}\label{eqn:chaining_error_first_term_aux}
\begin{aligned}
&\quad\ \E_U \frac{1}{s} \sum_{h=i}^\infty \E_g \sup_{J,C} \frac{\gamma_{i,J,C}}{r} \abs*{\sum_{x\in U} g_x (b_{h+1,J,C}^U(x)-b_{h,J,C}^U(x))}   \\
&\lesssim  \frac{1}{s}\sum_{h=i}^{\infty} 2^{i-h/2} \cdot \sqrt{dk\log k \log s}\cdot \E_U \sup_{J,C} \gamma_{i,J,C}\|U\cap X_{J_{i,C}}\|_0^{1/2}.
\end{aligned}
\end{equation}


We have the following lemma to bound the expected supremum on the right-hand side.
\begin{lemma} \label{lemma_chaining_var1}
It holds that $$
\E_U \sup_{J,C} \gamma_{i,J,C}\|U\cap X_{J_{i,C}}\|_0^{1/2} \lesssim \sqrt{\frac{s}{2^i}}
$$
\end{lemma}
\begin{proof}
Recall the event $\mathcal{E}$ defined in \eqref{eq_event}. When $\mathcal{E}$ happens, we have $\norm{U\cap X_\ell}_1 \in (1\pm \eps)\cdot |X_\ell|$ for every $\ell\in [k_0]$, which implies that $\frac{n}{s}\cdot \|U\cap X_{J_{i,C}}\|_0 \in (1\pm\eps) \cdot |X_{J_{i,C}}|$. Since $X$ is an $(r,k_0)$-regular instance, we know that $|X_{J_{i,C}}|\leq \frac{2|J_{i,C}|}{k_0}\cdot n$. It follows that $\|U\cap X_{J_{i,C}}\|_0\leq \frac{4|J_{i,C}|}{k_0}\cdot s$
and
\[
\E_U \left[\sup_{J,C}\frac{\|U\cap X_{J_{i,C}}\|^{1/2}_0}{ 1+2^i|J_{i,C}|/k_0} \middle| \mathcal{E}\right]
\leq  \E \left[\sup_{J,C}\frac{\|U\cap X_{J_{i,C}}\|^{1/2}_0}{ 2\cdot (2^{i}|J_{i,C}|/k_0)^{1/2}} \middle| \mathcal{E} \right]\\
\lesssim \sqrt{\frac{s}{2^i}}
\]
where the first inequality is by $a^2+b^2\geq 2ab$ and the second inequality is by $\|U\cap X_{J_{i,C}}\|_0\leq \frac{4|J_{i,C}|}{k_0}\cdot s$.

On the other hand, it always holds that
\[
    \sup_{J,C}\frac{\|U\cap X_{J_{i,C}}\|^{1/2}_0}{ 1+2^i|J_{i,C}|/k_0} \leq k_0\sqrt{s}.
\]
Therefore,
\begin{multline*}
\E_U \sup_{J,C} \gamma_{i,J,C}\|U\cap X_{J_{i,C}}\|_0^{1/2} = \E_U \sup_{J,C}\frac{\|U\cap X_{J_{i,C}}\|^{1/2}_0}{ 1+2^i|J_{i,C}|/k_0}
\leq \sqrt{\frac{s}{2^i}}\cdot \mathbb{P}(\mathcal{E}) + k_0\sqrt{s}\cdot \mathbb{P}(\bar{\mathcal{E}})\\
\leq \sqrt{\frac{s}{2^i}} + k_0\sqrt{s}\cdot \frac{\eps}{sk}
\lesssim \sqrt{\frac{s}{2^i}}.
\end{multline*}
Here, the last inequality follows from the fact that $\frac{1}{2^{i/2}}\geq \frac{1}{2^i}\geq \frac{1}{2^T}\gtrsim \eps$.
\end{proof}

Hence, for the first term in \eqref{eqn:chaining_error_aux}, by plugging Lemma~\ref{lemma_chaining_var1} into \eqref{eqn:chaining_error_first_term_aux}, we obtain that
\begin{equation}\label{eqn:chaining_error_first_term_final}
\begin{aligned}
\E_U \frac{1}{s} \sum_{h=i}^\infty \E_g \sup_{J,C} \frac{\gamma_{i,J,C}}{r} \abs*{\sum_{x\in U} g_x (b_{h+1,J,C}^U(x)-b_{h,J,C}^U(x))}   
\lesssim{} &  \frac{1}{\sqrt{s}}\sum_{h=i}^{\infty} 2^{i/2-h/2} \cdot \sqrt{dk\log k \log s} \\
\lesssim{} &  \sqrt{\frac{dk\log k \log s}{s}}\\
\leq{} &  \frac{\eps}{\log^2 \eps^{-1}\cdot \log \frac{k}{\eps}\cdot \log (mk)},
\end{aligned}
\end{equation}
provided that
\[
    s \gtrsim \frac{kd}{\eps^2} \cdot \log^{O(1)}(\frac{kd}{\eps})\cdot  \log^3 m.
\]

Next, we bound the second term in \eqref{eqn:chaining_error_aux}, starting with
\[
\E_U\frac{1}{s}\E_g \sup_{J,C} \frac{\gamma_{i,J,C}}{r} \abs*{\sum_{x\in U} g_x \cdot b_{i,J,C}^U(x)}
 =
\E_U\frac{1}{s}\E_g \sup_{J,C} \abs*{ \frac{\sum_{x\in U} g_x \cdot b_{i,j,C}^U(x)}{(1+2^i|J_{i,C}|/k_0)\cdot r} }.
\]

If $J_{i,C}=\emptyset$, we have $b_{i,J,C}^U(x)=0$ for all $x\in U$ and the gaussian sum vanishes. Thus, it suffices to consider $J_{i,C}\neq\emptyset$ for the supremum. Then
\begin{align*}
\E_g \sup_{J,C} \abs*{ \frac{\sum_{x\in U} g_x \cdot b_{i,J,C}^U(x)}{(1+2^i|J_{i,C}|/k_0)\cdot r}
}
&\leq \E_g \sup_{J,C} \frac{\sum_{j\in J_{i,C}} \abs*{\sum_{x\in U\cap X_j} g_x \cdot b_{i,J,C}^U(x)}}{2^ir\cdot|J_{i,C}|/k_0} \\
&\leq \E_g\sup_{J,C}\max_{j\in J_{i,C}} \abs*{ \frac{\sum_{x\in U\cap X_j} g_x\cdot b_{i,J,C}^U(x)}{2^ir/k_0} }\\
&\leq \E_g \sup_{C}\max_{j\in I_{i,C}} \abs*{ \frac{\sum_{x\in U\cap X_j} g_x\cdot b_{i,J,C}^U(x)}{2^ir/k_0} }\\
&= \E_g \sup_{C}\max_{j\in I_{i,C}} \abs*{ \sum_{x\in U\cap X_j} g_x }\cdot \frac{k_0\cdot \norm{b_{i,J,C}^U}_{\ell^\infty(X_j)}}{2^ir}. 
\end{align*}
In the above, the third inequality is due to the definition that $J_{i,C}=J\cap I_{i,C}\subseteq I_{i,C}$, and the last equality the fact that $b_{i,J,C}^U(x)$ is a constant over each cluster.

Since $b_{i,J,C}^U(x)\leq v_{i,J,C}^U(x)+4r=\dist(x,C)+4r\lesssim 2^i r$, we can then bound
\begin{equation}\label{eqn:chaining_error_second_term_semifinal}
\begin{aligned}
\E_U\frac{1}{s}\E_g \sup_{J,C} \abs*{ \frac{\sum_{x\in U} g_x \cdot b_{i,J,C}^U(x)}{(1+2^i|J_{i,C}|/k_0)\cdot r} } 
&\leq  \E_U\frac{1}{s}\E_g \sup_{C}\max_{j\in I_{i,C}} \abs*{ \sum_{x\in U\cap X_j} g_x 
}\cdot \frac{k_0\cdot \norm{b_{i,J,C}^U}_{\ell^\infty(X_j)}}{2^ir} \\
&\lesssim \E_U\frac{k_0}{s}\cdot \E_g \sup_{C}\max_{j\in I_{i,C}}  \abs*{ \sum_{x\in U\cap X_j} g_x }\\
&\leq   \E_U\frac{k_0}{s}\cdot \E_g \max_{j\in [k_0]}  \abs*{ \sum_{x\in U\cap X_j} g_x }\\
&\lesssim \frac{k_0}{s}\cdot \sqrt{\log k_0} \cdot \E_U \max_{j\in [k_0]} \|U\cap X_j\|^{1/2}_0.
\end{aligned}
\end{equation}
where the last inequality follows from the observation that  $\max_{j\in [k_0]} \big|\sum_{x\in U\cap X_j} g_x \big|$ is the maximum of $k_0$ ganssians, each with a variance at most $\max_{j\in [k_0]} \|U\cap X_j\|_0$.

\begin{lemma} \label{lemma_chaining_var2}
$$
\E_U \max_{j\in [k_0]} \|U\cap X_j\|^{1/2}_0\lesssim \sqrt{\frac{s}{k_0}}.
$$
\end{lemma}

\begin{proof}
Identical to the proof of Lemma~\ref{lemma_chaining_var1}, we know that if $\mathcal{E}$ happens,  $\|U\cap X_j\|_0\leq \frac{4s}{k_0}$ for every $j\in [k_0]$. It follow from Lemma~\ref{lemma_event} that
\[
\E_U \max_{j\in [k_0]} \|U\cap X_j\|^{1/2}_0
\lesssim \sqrt{\frac{s}{k_0}}\cdot\mathbb{P}(\mathcal{E}) + \sqrt{s}\cdot \mathbb{P}(\bar{\mathcal{E}})
\lesssim \sqrt{\frac{s}{k_0}} + \sqrt{s}
\cdot \frac{\eps}{ks}
\lesssim \sqrt{\frac{s}{k_0}}. \qedhere
\]
\end{proof}

Therefore, the second term in \eqref{eqn:chaining_error_aux} can be upper bounded as
\begin{equation}\label{eqn:chaining_error_second_term_final}
\E_U\frac{1}{s}\E_g \sup_{J,C} \abs*{ \frac{\sum_{x\in U} g_x \cdot b_{i,J,C}^U(x)}{(1+2^i|J_{i,C}|/k_0)\cdot r} }
\lesssim \sqrt{\frac{k_0\log k_0}{s}}
\leq \frac{\eps}{\log^2 \frac{1}{\eps}\cdot \log \frac{k}{\eps} \cdot \log (mk)},
\end{equation}
provided that 
\[
    s\gtrsim \frac{k}{\eps^2}\cdot \log^{O(1)} \frac{k}{\eps}\cdot \log^2 m.
\]
The bound \eqref{lemma_chaining_maininq} follows immediately from combining \eqref{eqn:chaining_error_symmetrization}, \eqref{eqn:chaining_error_aux}, \eqref{eqn:chaining_error_first_term_final} and \eqref{eqn:chaining_error_second_term_final}. The proof of Lemma~\ref{lemma_chaining} is complete.
\end{proof}

The claim of indexed-subset cost approximation in Lemma~\ref{alg_guarantee} follows immediately by combining Lemmata~\ref{lemma_event}, \ref{lemma_relate}, \ref{lemma_chaining}, and rescaling $\eps$.

\begin{proof}[Proof of Lemma~\ref{alg_guarantee}: $D$ is a indexed-subset cost approximation]

\sloppy
We follow the notation of the proof that $D$ is a range space approximation at the end of Section~\ref{sec:alg_guarantee_proof} and we also condition on all three events: $\mathcal{E}$  in \eqref{eq_event}, the event of Lemma~\ref{lemma_range} and the event of Lemma~\ref{lemma_subset} all happen.

Note that \eqref{eq_weight} implies that for every $J\subseteq [k_0]$ and $C\in \binom{M}{k}$,
\begin{eqnarray} \label{inq_cost}
(1 - O(\epsilon)) \cost(U_J,C) \leq \cost(D_J,C)\leq (1 + O(\epsilon)) \cost(U_J,C).
\end{eqnarray}
Plugging \eqref{inq_cost} into \eqref{inq_subset} yields immediately 
\[
    (1 - O(\eps)) \cost(X_J,C) - O(\eps n r) \leq \cost(D_J,C)\leq (1 + O(\eps)) \cost(X_J,C) + O(\eps n r),
\]
which, by rescaling $\eps$, establishes that $D$ satisfies Definition \ref{def_subset}.
\end{proof}

\section{Proof of Lemma~\ref{lemma_Euc_chaining_main}: Handling Error for Medium $\cost(X_J, C)$} \label{sec_Euc_chaining_main}

In this section, we fix indices
$j\in [b_{max}]$, $l\in \{0,\dots,l_{max}\}$ and $\psi \in \{0,\dots,\psi_{max}\}$, and adapt the chaining method from \cite{bansal2024sensitivity}. Our key difference lies in the denominator in \eqref{inq_Euc_main}: we use $\cost(X_J,C)+\cost(X,A)$ instead of the larger $\cost(X,C)$ used in \cite{bansal2024sensitivity}. Since $\cost(X_J,C)+\cost(X,A)$ can be much smaller than $\cost(X,C)$ in the general case, we cannot expect a black-box reduction from their analysis. Indeed, extra care is needed to refine the variance bounds in order to achieve our more sensitive bound.

\paragraph{Distance Vectors} For a subset $J\in \mathcal{B}_j$, and a $k$-center $C\in L_{l,\psi}^J$, we define then cost vector $v_C^J:U\rightarrow \mathbb{R}_{\geq 0}$ as,
$$
v_C^J(x)=
\begin{cases}
\dist(x,C) & \text{if }x\in U_J\\
0 & \text{otherwise}.
\end{cases}
$$
The distance vector set is defined as $V^J=\{v_C^J\mid C\in L_{l,\psi}^J\}$. We observe that $\cost(U_J,C)=\sum_{x\in U} w_U(x)\cdot v_C^J(x)$ and $\cost(X_J,C)=\E \sum_{x\in U} w_U(x)\cdot v_C^J(x)$. The symmetrization trick implies that
\begin{align}
 \E\sup_{J\in \mathcal{B}_j,\,C\in L^J_{l,\psi}} \abs*{\frac{\cost(U_J,C)-\cost(X_J,C)}{\cost(X_J,C)+\cost(X,A)} } 
={} & \E\sup_{J,C} \abs*{ \frac{\sum_{x\in U} w_U(x) v_C^J(x)-\E \sum_{x\in U} w_U(x) v_C^J(x)}{\cost(X_J,C)+\cost(X,A)} } \notag\\
\lesssim{} & \label{inq_Euc_sym}\E\sup_{J,C}\bigg|\frac{\sum_{x\in U} g_x w_U(x)v_C^J(x)}{\cost(X_J,C)+\cost(X,A)}\bigg|,
\end{align}
where $\{g_x\}_{x\in U}$ are i.i.d.\ standard Gaussians. The rest of the proof focuses on bounding \eqref{inq_Euc_sym}.

\paragraph{Vector Nets} An $\alpha$-net $N_{\alpha}=N_{\alpha,J,j,l,\psi}$ for $V^J$ is a subset of $V^J$ such that for any $v_C^J\in V^J$, there exists a vector $q\in N_{\alpha}$ satisfying the following properties for all $x\in U$,
\begin{enumerate}
\item $|q(x)-v_C^J(x)|\leq \alpha \cdot (\dist(x,A)+\Delta_{\pi(x)})$ if $v_C^J(x)>0$;
\item $q(x)=0$ if $v_C^J(x)=0$.
\end{enumerate}

\begin{lemma}[Adaption of Lemma G.3 of \cite{bansal2024sensitivity}]\label{lemma_Euc_netsize}
For any $\alpha\in (0,1]$, there exists an $\alpha$-net\footnote{We remark that \cite{bansal2024sensitivity} constructs $\alpha$-net for $\alpha\in (0,\frac{1}{2}]$. However, for $\alpha\in(\frac{1}{2},1]$, we can simply let $N_{\alpha}:=N_{\frac{1}{2}}.$} of $V^J$ such that
$$
\log |N_{\alpha}|\lesssim \min(2^\psi+k\alpha^{-2},2^{2l}k\alpha^{-2})\cdot \log (k\alpha^{-1}\eps^{-1}).
$$
\end{lemma}

  \paragraph{Telescoping} Let $h_{max}=3\lceil \log \frac{k}{\eps}\rceil$ and for each $h\in \{0,\dots,h_{max}\}$, let $q_C^h\in N_{2^{-h}}$ denote the net vector of $v_C^J$. Write $v_C^J(X)$ as
\begin{equation} \label{eq_Euc_tel}
v_C^J(x) = \left(v_C^J(x)-q_C^{h_{max}}(x)\right) + \sum_{h=1}^{h_{max}} \left(q_C^{h}(x)-q_C^{h-1}(x)\right) + q_C^0(x).
\end{equation}

Substituting \eqref{eq_Euc_tel} into \eqref{inq_Euc_sym} and applying the triangle inequality yields that
\begin{align*}
\E\sup_{J,C}\bigg|\frac{\sum_{x\in U} g_x w_U(x)v_C^J(x)}{\cost(X_J,C)+\cost(X,A)}\bigg|
&\leq \sum_{h=1}^{h_{max}}\E\sup_{J,C}\bigg|\frac{\sum_{x\in U} g_x w_U(x)\cdot (q_C^h(x)-q_C^{h-1}(x))}{\cost(X_J,C)+\cost(X,A)}\bigg|\\
&\qquad +\E\sup_{J,C}\bigg|\frac{\sum_{x\in U} g_x w_U(x)\cdot q_C^0(x)}{\cost(X_J,C)+\cost(X,A)}\bigg|\\
&\qquad +\E\sup_{J,C}\bigg|\frac{\sum_{x\in U} g_x w_U(x)\cdot (v_C^J(x)-q_C^{h_{max}}(x))}{\cost(X_J,C)+\cost(X,A)}\bigg|.
\end{align*}
We shall bound each term on the right-hand side individually. We first handle the last term.

In the remainder of this section, we set $\nu = 5$. 

\begin{lemma}[Residual term] \label{lemma_euc_residual}
\begin{eqnarray*}
\E\sup_{J,C}\bigg|\frac{\sum_{x\in U} g_x w_U(x)\cdot (v_C^J(x)-q_C^{h_{max}}(x))}{\cost(X_J,C)+\cost(X,A)}\bigg|\lesssim \frac{\eps}{\log^\nu \frac{k}{\eps}}.
\end{eqnarray*}
\end{lemma}

\begin{proof}
Let $\beta \asymp \frac{\eps}{\log^\nu \frac{k}{\eps}}$. By our choice of $h_{max}$, 
we have  
$|v_C^J(x)-q_C^{h_{max}}(x)| \leq \beta^2 (\dist(x,A)+\Delta_{\pi(x)})$,
whence it follows from Lemma~\ref{lemma_Euc_weight} that
$$
\frac{w_U(x)\cdot |v_C^J(x)-q_C^{h_{max}}(x)|}{\cost(X_J,C)+\cost(X,A)}\lesssim \frac{\beta^2 \cost(X,A)}{s\cost(X,A)}= \frac{\beta^2}{s}.
$$
Then, by Cauchy-Schwarz inequality, 
\[
\E\sup_{J,C}\bigg|\frac{\sum_{x\in U} g_x w_U(x)\cdot (v_C^J(x)-q_C^{h_{max}}(x))}{\cost(X_J,C)+\cost(X,A)}\bigg|\leq \E \|g\|_2 \cdot \frac{\beta}{\sqrt{s}}\leq \beta. \qedhere
\]
\end{proof}

The other two terms can be bounded as in the following two lemmata, whose proofs are deferred to Sections \ref{sec_Euc_chaining_coarse} and \ref{sec_Euc_chaining_main_main}, respectively.

\begin{lemma}[Initial telescoping term] \label{lemma_Euc_chaining_coarse}
\begin{eqnarray*}
\E\sup_{J,C}\bigg|\frac{\sum_{x\in U} g_x w_U(x)\cdot q_C^0(x)}{\cost(X_J,C)+\cost(X,A)}\bigg|\lesssim \frac{\eps}{\log^\nu \frac{k}{\eps}}.
\end{eqnarray*}
\end{lemma}

\begin{lemma}[Main telescoping terms] \label{lemma_Euc_chaining_main_main}
For each $h\in [h_{max}]$,
$$
\E\sup_{J,C}\bigg|\frac{\sum_{x\in U} g_x w_U(x)\cdot (q_C^h(x)-q_C^{h-1}(x))}{\cost(X_J,C)+\cost(X,A)}\bigg|\lesssim \frac{\eps}{\log^{\nu+1}\frac{k}{\eps}}.
$$
\end{lemma}

Lemma~\ref{lemma_Euc_chaining_main} follows immediately from combining the three preceding lemmata.

\subsection{Proof of Lemma~\ref{lemma_Euc_chaining_coarse}: Initial telescoping term}

\label{sec_Euc_chaining_coarse}

For each $i\in J$, 
we round $\dist(a_i,c)$ to an integer multiple of $\Delta_i$ by setting
\[
    d_i = \left\lceil\frac{\dist(a_i,C)}{\Delta_{i}}\right\rceil\cdot \Delta_i.
\]
Thus, $|\dist(a_i,C)-d_i|\leq \Delta_i$ and $d_i\in \{0,\Delta_i,\dots,2^l\Delta_i\}$. Notably, there are at most $(2^l+1)^{|J|}\leq 2^{O(kl)}<|N_{1/2}|$ distinct tuples of $(d_i)_{i\in J}$. 

By the definition of the net and the triangle inequality, we have that 
\begin{align}
&\E\sup_{J,C}\bigg|\frac{\sum_{x\in U} g_x w_U(x)\cdot q_C^0(x)}{\cost(X_J,C)+\cost(X,A)}\bigg| \notag \\
\lesssim{} & \E\sup_{J,C}\bigg|\frac{\sum_{x\in U_J} g_x w_U(x)\cdot d_{\pi(x)}}{\cost(X_J,C)+\cost(X,A)}\bigg|
+\E\sup_{J,C}\bigg|\frac{\sum_{x\in U_J} g_x w_U(x)\cdot (q_C^0(x)-d_{\pi(x)})}{\cost(X_J,C)+\cost(X,A)}\bigg| \label{eqn:chaining_euc_initial_aux}
\end{align}
We bound the second term first. By the definition of the net, $|q_0^C(x)-d_{\pi(x)}|\lesssim \cost(x,A)+\Delta_{\pi(x)}$. 
Also note that the number of distinct vectors $q_0^C(x) - d_\pi(x)$ is at most 
$2^k|N_{\frac{1}{2}}|\cdot |N_{\frac{1}{2}}| = 2^k |N_{1/2}|^2$. Hence, both the variance and number of gaussians in the second term are of the same order as those in
$$
\E\sup_{J,C}\bigg|\frac{\sum_{x\in U} g_x w_U(x)\cdot (q_C^1(x)-q_C^{0}(x))}{\cost(X_J,C)+\cost(X,A)}\bigg|.
$$
An identical argument as in Section~\ref{sec_Euc_chaining_main_main} gives, for $s\gtrsim k\eps^{-2}\log^{2\nu+1}\frac{k}{\eps}$, that
\begin{equation}\label{eqn:chaining_euc_initial_second_term_final}
\E\sup_{J,C}\bigg|\frac{\sum_{x\in U} g_x w_U(x)\cdot (q_C^0(x)-d_{\pi(x)})}{\cost(X_J,C)+\cost(X,A)}\bigg|\lesssim \frac{\eps}{\log^{\nu}\frac{k}{\eps}}.
\end{equation}

\paragraph{Remark} We remark that discretizing $\dist(a_i,C)$ into $d_i$ is essential. If one were to set $d_i=\dist(a_i,C)$, the number of Gaussians over which a supremum is taken could become unbounded. This detail was not fully addressed in \cite{bansal2024sensitivity}. Specifically, while the authors define in Section 2.5 that $u^{S,0}(x):=\dist(a_{\pi(x)},S)$, they directly bound in Section 2.6 the number of pairs $(u^{S,0},u^{S,1})$ by $|N_{1/2}|^2$ without further discussion.

It remains to control the first term of \eqref{eqn:chaining_euc_initial_aux}. Note that
\begin{equation}\label{eqn:chaining_euc_initial_first_aux}
	\begin{aligned}
		     \E\sup_{J,C} \abs*{\frac{\sum_{x\in U} g_x w_U(x)\cdot d_{\pi(x)}}{\cost(X_J,C)+\cost(X,A)}} 
		&= \E\sup_{J,C}\abs*{\frac{\sum_{i\in J}\sum_{x\in U_i} g_x w_U(x)\cdot d_i}{\cost(X_J,C)+\cost(X,A)}}\\
		&\leq \E\sup_{i\in J}\sup_{C\in L_{l,\psi}^j} \abs*{\frac{\sum_{x\in U_i} g_x w_U(x)\cdot d_i}{\cost(X_i,C)+\cost(X_i,A)}}\\
		&= \E\sup_{i\in J}\sup_{C\in L_{l,\psi}^j} \abs*{\sum_{x\in U_i} g_x w_U(x)} \cdot \frac{ d_i}{\cost(X_i,C)+\cost(X_i,A)}\\
		&\lesssim \E\sup_{i\in J} \abs*{\sum_{x\in U_i} g_x \cdot \frac{w_U(x)}{|X_i|}},
	\end{aligned}
\end{equation}
where in the last inequality, we used the triangle inequality to obtain that $$|X_i|d_i\lesssim \cost(X_i,C)+\cost(X_i,A).$$
Let $G_{U,i} = \sum_{x\in U_i} g_x \cdot \frac{w_U(x)}{|X_i|}$. We have
\[
    \E_{U,g} \sup_{i\in J} \; \abs{G_{U,i}} = \E_U \left[ \E_g \sup_{i\in J} \; \abs{G_{U,i}} \;\middle\vert\; U\right] \leq \sqrt{\log k}\E_U \sqrt{\Var[G_{U,i} \mid U]}
\]
since the supremum is taken over at most $k$ Gaussians. Next we bound $\E_U \sqrt{\Var[G\mid U]}$, for which we need the following lemma.
\begin{lemma}[Variance bound]
\label{lemma_Euc_minorcase_1}
We have that
\[
\Var[G_{U,i}\mid U]\leq \frac{k\|U_i\|}{s|X_i|}
\qquad\text{and}\qquad
\E\sup_{i\in [k]} \sqrt\frac{\|U_i\|_1}{|X_i|}\lesssim 1.
\]
\end{lemma}
\begin{proof}
Lemma~\ref{lemma_Euc_weight} implies that $w_U(x)\leq\frac{k|X_i|}{m}$ for each $x\in U_i$, then
\[
\Var[G\mid U]=\frac{1}{|X_i|^2}\sum_{x\in U_i} w_U(x)^2\leq \frac{k}{m|X_i|}\cdot \sum_{x\in U_i} w_U(x)=\frac{k\|U_i\|_1}{s|X_i|}. \qedhere
\]

By Lemma~\ref{lemma_Euc_weight}, $\|U_i\|_1=\sum_{x\in U_i} w_U(x)\leq \sum_{x\in U_i}\frac{k|X_i|}{m}\leq k|X_i|$ and so it always holds that $\frac{\|U_i\|_1}{|X_i|}\leq k$. Meanwhile, when $\mathcal{E}_0$ happens, we have $\|U_i\|_1\leq (1+\eps)\cdot |X_i|$. Thus,
\[
\E\sup_{i\in [k]} \sqrt\frac{\|U_i\|_1}{|X_i|} \leq \sqrt{1+\eps}\cdot \mathbb{P}[\mathcal{E}_0] + \sqrt{k}\cdot (1-\mathbb{P}[\mathcal{E}_0])
\lesssim 1+\sqrt{k}\cdot \frac{\eps}{k}
\lesssim 1. \qedhere
\]
\end{proof}

By Lemma~\ref{lemma_Euc_minorcase_1}, we obtain that
$$
\E\sup_{i\in J}\bigg|\sum_{x\in U_i} g_x \cdot \frac{w_U(x)}{|X_i|}\bigg|\leq \sqrt\frac{k\log k}{s}\lesssim\frac{\eps}{\log^\nu\frac{k}{\eps}},
$$
provided that $s\gtrsim k\eps^{-2}\log^{2\nu+1}\frac{k}{\eps}$, which, together with \eqref{eqn:chaining_euc_initial_first_aux}, implies that
\begin{equation}\label{eqn:chaining_euc_initial_first_term_final}
	\E\sup_{J,C} \abs*{\frac{\sum_{x\in U} g_x w_U(x)\cdot d_{\pi(x)}}{\cost(X_J,C)+\cost(X,A)}} \lesssim \frac{\eps}{\log^\nu\frac{k}{\eps}}.
\end{equation}
Combining \eqref{eqn:chaining_euc_initial_aux}, \eqref{eqn:chaining_euc_initial_first_term_final} and \eqref{eqn:chaining_euc_initial_second_term_final} completes the proof of Lemma~\ref{lemma_Euc_chaining_coarse}.

\subsection{Proof of Lemma~\ref{lemma_Euc_chaining_main_main}: Main telescoping terms} \label{sec_Euc_chaining_main_main}

Fix $h\in [h_{max}]$. Let 
\[
    G(J,C)=\frac{\sum_{x\in U} g_x w_U(x)\cdot (q_C^h(x)-q_C^{h-1}(x))}{\cost(X_J,C)+\cost(X,A)}.
\]
For a fixed $U$, $G_{J,C}$ is a Gaussian variable and its variance is bounded by
\begin{align}
\mathrm{Var}[G(J,C)\mid U]
&= \frac{\sum_{x\in U} w_U(x)^2 (q_C^h(x)-q_C^{h-1}(x))^2}{(\cost(X_J,C)+\cost(X,A))^2} \notag\\
&\lesssim \frac{2^{-2h}\cdot \sum_{x\in U_J}  w_U(x)^2 (\dist(x,A)^2+\Delta_{\pi(x)}^2)}{(\cost(X_J,C)+\cost(X,A))^2} \label{inq_chaining_main_main}
\end{align}

Next, we bound $\mathrm{Var}[G(J,C) \mid U]$ in two different ways.

\begin{lemma}[First variance bound] 
\label{lemma_Euc_variance1}
\begin{gather}
\mathrm{Var}[G(J,C)\mid U]\leq \frac{2^{-2h}}{s}\cdot \frac{\cost(U_J,A)+\sum_{i\in J}u_i\Delta_i}{\cost(X_J,C)+\cost(X,A)} \notag
\shortintertext{and}
    \E_U\bigg[\sup_{J,C} \bigg(\frac{\cost(U_J,A)+\sum_{i\in J}u_i\Delta_i}{\cost(X_J,C)+\cost(X,A)}\bigg)^{1/2}\bigg]\lesssim 2^{-l/2}. \label{inq_Euc_chaining_var1}
\end{gather}
\end{lemma}

\begin{proof}
Lemma~\ref{lemma_Euc_weight} implies that $w_U(x)\dist(x,A)\leq \frac{\cost(X,A)}{s}$ and $w_U(x)\Delta_{\pi(x)}\leq \frac{\cost(X,A)}{s}$.
Then,
\begin{align*}
\sum_{x\in U_J}  w_U(x)^2 (\dist(x,A)^2+\Delta_{\pi(x)}^2)
&\leq \frac{\cost(X,A)}{s}\sum_{x\in U_J} w_U(x)(\dist(x,A)+\Delta_{\pi(x)})\\
&= \frac{\cost(X,A)}{s}\cdot (\cost(U_J,A)+\sum_{i\in J}u_i\cdot \Delta_i)
\end{align*}
Plugging the inequality above into \eqref{inq_chaining_main_main} completes the proof for the first inequality.

Next, we prove the second inequality.
we show that it always holds
$$
\frac{\cost(U_J,A)+\sum_{i\in J}u_i\Delta_i}{\cost(X_J,C)+\cost(X,A)}\leq 8.
$$
To see this, note that by Lemma~\ref{lemma_Euc_weight}, 
\begin{gather*}
\cost(U_J,A)=\sum_{x\in U_J} w_U(x)\cdot \dist(x,A)\leq \sum_{x\in U_J}\frac{4\cost(X,A)}{s}\leq 4\cost(X,A)
\shortintertext{and}
\sum_{i\in J}u_i\Delta_i=\sum_{x\in U_J} w_U(x)\Delta_{\pi(x)}\leq \sum_{x\in U_J} \frac{4\cost(X,A)}{s}=4\cost(X,A).
\end{gather*}
Thus \eqref{inq_Euc_chaining_var1} always holds for $l \leq 2$.

For $l\geq 3$, we note that by the triangle inequality and $C\in L^j_{l,\psi}$, 
\begin{equation}\label{eqn:euc_main_main_aux1}
    \cost(X_J,C)\geq (2^{l-1}-1)\cost(X_J,A).
\end{equation}
When $\mathcal{E}_0\cup\mathcal{E}_1$ happens, we have that 
\begin{gather}
    \cost(U_J,A) \leq (1+\eps)\cost(X_J,A) \label{eqn:euc_main_main_aux2}\\
    \sum_{i\in J} u_i\Delta_i \leq (1+\eps)\sum_{i\in J} |X_i|\Delta_i=(1+\eps)\cdot \cost(X_J,A). \label{eqn:euc_main_main_aux3}
\end{gather}
Thus,
\begin{align*}
    & \E_U\left[\sup_{J,C} \bigg(\frac{\cost(U_J,A)+\sum_{i\in J}u_i\Delta_i}{\cost(X_J,C)+\cost(X,A)}\bigg)^{1/2} \;\middle\vert\; \mathcal{E}_0\cup\mathcal{E}_1\right] \\
    \leq{} & \E_U\left[\sup_{J} \bigg(\frac{2(1+\eps)\cost(X_J,A)}{(2^{l-1}-1)\cost(X_J,A)+\cost(X,A)} \bigg)^{1/2} \;\middle\vert\; \mathcal{E}_0\cup\mathcal{E}_1 \right] \\
    \leq{} & \E_U\left[ \bigg(\frac{2(1+\eps)\cost(X,A)}{(2^{l-1}-1)\cost(X,A)+\cost(X,A)} \bigg)^{1/2} \middle\vert\; \mathcal{E}_0\cup\mathcal{E}_1 \right] \\
    \lesssim{} & 2^{-l/2}.
\end{align*}
Therefore,
\begin{multline*}
 \E_U \bigg[\sup_{J,C} \bigg(\frac{\cost(U_J,A)+\sum_{i\in J}u_i\Delta_i}{\cost(X_J,C)+\cost(X,A)}\bigg)^{1/2}\bigg] 
\lesssim 2^{-l/2}\cdot \mathbb{P}[\mathcal{E}_0\cup\mathcal{E}_1]+2\cdot (1-\mathbb{P}[\mathcal{E}_0\cup\mathcal{E}_1])\\
\lesssim 2^{-l/2}+\frac{\eps}{k}
\lesssim 2^{-l/2},
\end{multline*}
where the last inequality is due to the fact that $2^{-l/2}\geq 2^{-l_{max}/2}\gtrsim \eps^{1/2}\geq \frac{\eps}{k}$.
This completes the proof of Lemma \ref{lemma_Euc_variance1}.
\end{proof}

\begin{lemma}[Second variance bound] 
\label{lemma_Euc_variance2}
\begin{eqnarray*}
\mathrm{Var}[G(J,C)\mid U]\leq 2^{-2h-\psi}\cdot \frac{k^2}{s}\cdot \frac{(\cost(U_J,A)+\sum_{i\in J} u_i\Delta_i)\cdot \cost(X_J,A)}{(\cost(X_J,C)+\cost(X,A))^2}
\end{eqnarray*}
and
\begin{eqnarray} \label{inq_Euc_chaining_var2}
\E_U\bigg[\sup_{J,C} \bigg(\frac{(\cost(U_J,A)+\sum_{i\in J} u_i\Delta_i)\cdot \cost(X_J,A)}{(\cost(X_J,C)+\cost(X,A))^2}\bigg)^{1/2}\bigg]\lesssim 2^{-l}.
\end{eqnarray}
\end{lemma}

\begin{proof}
By the last bound in Lemma~\ref{lemma_Euc_weight}, we have that 
\begin{align*}
\sum_{x\in U_J} w_U(x)^2\dist(x,A)^2
=\sum_{i\in J}\sum_{x\in U}w_U(x)^2\dist(x,A)^2
&\leq \sum_{i\in J}\sum_{x\in U} w_U(x)\dist(x,A)\cdot \frac{k\cost(X_i,A)}{s}\\
&\lesssim\frac{k}{s\cdot |J|}\cdot \cost(U_J,A)\cdot \cost(X_J,A)\\
&\lesssim \frac{k^2}{s\cdot 2^\psi} \cdot \cost(U_J,A)\cdot \cost(X_J,A).
\end{align*}
Here, for the second inequality, we used the definition of $J\in \mathcal{B}_j$, which means that $\cost(X_i,A)\lesssim \frac{\cost(X_J,A)}{|J|}$ for all $i\in J$, and for the last inequality, we used the fact that $k|J|\gtrsim 2^\psi$ by the definition of $L^J_{l,\psi}$. 

Similarly, by the second bound in Lemma~\ref{lemma_Euc_weight}, we have
\begin{align*}
\sum_{x\in U_J} w_U(x)^2\Delta_{\pi(x)}^2
= \sum_{i\in J}\sum_{x\in U_i} w_U(x)^2\Delta_i^2
&\leq \sum_{i\in J}\sum_{x\in U_i} w_U(x)\Delta_i\cdot \frac{k\cost(X_i,A)}{s}\\
&\lesssim \frac{k}{s\cdot |J|}\cdot \sum_{i\in J} u_i\Delta_i\cdot \cost(X_J,A)\\
&\lesssim \frac{k^2}{s\cdot 2^\psi} \cdot\sum_{i\in J} u_i\Delta_i\cdot \cost(X_J,A).
\end{align*}
Plugging the above inequalities into \eqref{inq_chaining_main_main} completes the proof for the first inequality.

Next, we prove the second inequality.
Similar to the proof of Lemma~\ref{lemma_Euc_variance1}, it always holds that
\begin{gather*}
\cost(U_J,A)+\sum_{i\in J} u_i\Delta_i\leq 2\cost(X,A),
\shortintertext{whence it follows that}
\sup_{J,C} \bigg(\frac{(\cost(U_J,A)+\sum_{i\in J} u_i\Delta_i)\cdot \cost(X_J,A)}{(\cost(X_J,C)+\cost(X,A))^2}\bigg)^{1/2}\leq \sqrt{2}
\end{gather*}
and \eqref{inq_Euc_chaining_var2} holds for $l\leq 2$.

For $l\geq 3$, it follows from \eqref{eqn:euc_main_main_aux1}, \eqref{eqn:euc_main_main_aux2} and \eqref{eqn:euc_main_main_aux3} that
\begin{align*}
& \E_U\left[\sup_{J,C} \bigg(\frac{(\cost(U_J,A)+\sum_{i\in J} u_i\Delta_i)\cdot \cost(X_J,A)}{(\cost(X_J,C)+\cost(X,A))^2}\bigg)^{1/2}\middle\vert \mathcal{E}_0\cup \mathcal{E}_1\right]\\
\lesssim{} & \E_U \left[\sup_J \frac{(2(1+\eps))^{1/2} \cost(X_J,A)}{(2^{l-1}-1)\cost(X_J,A) + \cost(X,A)} \middle\vert \mathcal{E}_0\cup \mathcal{E}_1\right]\\
\leq{} & \E_U \left[\frac{(2(1+\eps))^{1/2} \cost(X,A)}{(2^{l-1}-1)\cost(X,A) + \cost(X,A)} \middle\vert \mathcal{E}_0\cup \mathcal{E}_1\right]\\
\lesssim{} & 2^{-l}.
\end{align*}
Therefore,
\begin{align*}
&\quad\ \E_U\left[\sup_{J,C} \bigg(\frac{(\cost(U_J,A)+\sum_{i\in J} u_i\Delta_i)\cdot \cost(X_J,A)}{(\cost(X_J,C)+\cost(X,A))^2}\bigg)^{1/2}\right]\\
&\lesssim 2^{-l}\cdot \mathbb{P}[\mathcal{E}_0\cup\mathcal{E}_1]+2\cdot (1-\mathbb{P}[\mathcal{E}_0\cup\mathcal{E}_1])\\
&\lesssim 2^{-l}+\frac{\eps}{k}\\
&\lesssim 2^{-l},
\end{align*}
where the last inequality is due to the fact that $2^{-l}\geq 2^{-l_{max}}\gtrsim \eps\geq \frac{\eps}{k}$.
This completes the proof of Lemma \ref{lemma_Euc_variance2}.
\end{proof}

We are now ready to bound $\E\sup_{J,C} |G(J,C)|$. 
Note that $\E\sup_{J,C} |G(J,C)|$ is a maximum over $2^k\cdot |N_{2^{-h}}|^2$ many Gaussians. It follows from Fact~\ref{fact_max_gauss} that
\[
\E\sup_{J,C} |G(J,C)|
\lesssim \E_U \E_g \left[\sup_{J,C}G(J,C) \;\middle\vert\; U\right]
\lesssim \E_U \sup_{J,C} \sqrt{\mathrm{Var}[G(J,C)\mid U]} \cdot  \sqrt{k+\log |N_{2^{-h}}|}.
\]
Now we discuss the following three cases. While we do not present them as mutually exclusive, together they cover all possible scenarios.

\begin{enumerate}[label=(\roman*)]
    \item $2^\psi\leq k\cdot 2^{2h}$. By the first bound in Lemma~\ref{lemma_Euc_netsize} and Lemma~\ref{lemma_Euc_variance1}, we obtain that
    \[
        \E\sup_{J,C} |G(J,C)|
        \lesssim \frac{2^{-h-l/2}}{\sqrt{s}}\cdot \sqrt{k \log (k/\eps)}\cdot 2^{h} 
        \leq \sqrt{\frac{ k\log (k/\eps)}{s}}
        \lesssim \frac{\eps}{\log^{\nu+1}\frac{k}{\eps}},
    \]
    provided that $s\gtrsim k\eps^{-2}\log^{2\nu+3}\frac{k}{\eps}$.

    \item $2^\psi > k$ and $k^{1/3} \geq \eps^{-1}$. By the second bound in Lemma~\ref{lemma_Euc_netsize} and Lemma~\ref{lemma_Euc_variance1}, we obtain that,
    \begin{align}
        \E\sup_{J,C} |G(J,C)|
        \lesssim \frac{2^{-h-l/2}}{\sqrt{s}}\cdot \sqrt{k \log (k/\eps)}\cdot 2^{h+l} 
        &\leq \sqrt{\frac{2^l k\log (k/\eps)}{s}} \notag \\
        &\lesssim \sqrt\frac{2^{l_{max}}\cdot k\log (k/\eps)}{s}  \notag\\
        &\lesssim \sqrt\frac{k\eps^{-1}\log (k/\eps)}{s} \label{inq_Euc_tradeoff1}\\
        &\lesssim \frac{\eps}{\log^{\nu+1} \frac{k}{\eps}} \notag,
    \end{align}
    provided that $s\gtrsim k\eps^{-3}\log^{2\nu+3}\frac{k}{\eps}$, which holds when $k^{1/3} \geq \eps^{-1}$.

    \item $k^{1/3} < \epsilon^{-1}$. By the first bound of Lemma~\ref{lemma_Euc_netsize} and Lemma~\ref{lemma_Euc_variance2}, we obtain that
    \begin{equation}\label{inq_Euc_tradeoff2}
        \E\sup_{J,C} |G(J,C)|  \lesssim \frac{k\cdot 2^{-h-l-\psi/2}}{\sqrt{s}}\cdot 2^{h+\psi/2}\cdot \sqrt{\log (k/\eps)}  
        \lesssim \sqrt{\frac{2^{-2l} k^2\log (k/\eps)}{s}}.
    \end{equation}
    Combining \eqref{inq_Euc_tradeoff1} and \eqref{inq_Euc_tradeoff2} and noticing that $\min(2^l k, 2^{-2l} k^2)\lesssim k^{4/3}$, we obtain 
    \[
        \E\sup_{J,C} |G(J,C)| \lesssim  \sqrt{\frac{\min(2^lk,2^{-2l} k^2)\cdot \log (k/\eps)}{s}}
        \lesssim \sqrt\frac{k^{4/3}\log (k/\eps)}{s}.
        \lesssim \frac{\eps}{\log^{\nu+1}\frac{k}{\eps}},
    \]
    provided that $s\gtrsim k^{4/3}\cdot \eps^{-2}\log^{2\nu+3}\frac{k}{\eps}$, which holds when $k^{1/3} < \eps^{-1}$.
\end{enumerate}

The proof of Lemma~\ref{lemma_Euc_chaining_main_main} is now complete.

\section{Extension to General $z\geq 1$ for VC and Doubling Instances}
\label{sec:kzC}

The following theorem extends both Theorem \ref{thm_gen} and Theorem \ref{thm_doubling} to general robust \kzC for constant $z\geq 1$.

\begin{theorem}[General robust \kzC for VC and doubling instances] 
\label{thm:kzC} 
Let $(M,\dist)$ be a metric space with VC dimension $d_{\VC}\geq 1$ and doubling dimension $d_D\geq 1$. 
Let $d = \min\{d_{\VC}, d_D\}$.
There exists an algorithm that given a dataset $X\subseteq M$ of size $n\geq 1$, constructs an $(\eps,m)$-robust coreset of 
$X$ for robust \kzC with size $O(m)+\tilde{O}(k\cdot d\cdot \eps^{-2z} )$ in $O(nk)$ time.
\end{theorem}

We first introduce the following lemmas for $(r,k)$-instances to prove this theorem.
The first one is a generalization of Lemma \ref{lemma_mainring}, whose proof can be found in Appendix \ref{sec:proof_1}.
Note that the coreset size extends to \( \tilde{O}(k d \eps^{-2z}) \cdot \log m \).

\begin{lemma}[First coreset for $(r,k)$-instances] 
\label{lemma_mainring_kzC}
Let $(M,\dist)$ be a metric space with VC dimension $d_{\mathrm{VC}}$ and doubling dimension $d_D$.
Let $d = \min\{d_{\mathrm{VC}}, d_D\}$.
There is a randomized algorithm which, given an $(r,k)$-instance $X$, computes with probability at least $1- O(\frac{1}{\log (mk)})
	$  an $(\eps,m,\eps r^z |X| )$-robust coreset with size $\tilde{O}(kd\eps^{-2z})\cdot \log m$ of $X$ for robust \kzC.
\end{lemma}

The second one is a generalization of Lemma \ref{lemma_innergroup}, whose proof can be found in Appendix \ref{sec:proof_2}.
Note that the additive error term extends to \( mkr^z/\eps^{z-1} \).

\begin{lemma}[Second coreset for $(r,k)$-instances] 
\label{lemma_innergroup_kzC} 
Suppose there exists a randomized algorithm $\mathcal{B}$ that on every \kzC instance, with probability at least $0.9$, computes a vanilla $\eps$-coreset with size $Q(\eps)$.
Given an $(r,k)$-instance $X$, with probability at least $0.9$, Algorithm~\ref{alg_coreset2} computes an $(\eps,m,mkr^z/\eps^{z-1})$-robust coreset with size $k+Q(\eps)$ of $X$ for robust \kMedian.
\end{lemma}

We also need the following lemma that generalizes Lemma \ref{thm_decomp}.

\begin{lemma}[Point set decomposition for general $z$] 
\label{thm_decomp_kzC}
	There is an $O(nk)$ time algorithm $\mathcal{A}$ that, given any constant-factor approximation $C^\ast$ for robust \kzC on $X$, decomposes $X$ as
	\[
		X=F\cup R\cup G,
	\]
	where
	\begin{itemize}[itemsep=0pt,topsep=0pt]
	\item $F$ is a finite subset of $X$ with $|F| = O(m+\eps^{-1})$, 
	\item $R$ is the union $R=R_1\cup\cdots\cup R_l$ with $l = O(\log (zmk/\eps))$, each $R_i$ being an $(r_i,k)$-instance satisfying $\sum_{i\in [l]} r_i^z|R_i| = O(\cost_z^{(m)}(X,C^\ast))$, 
	\item $G$ is an $(\eps \cdot (\frac{\cost_z^{(m)}(X,C^\ast)}{mk})^{1/z}, k)$-instance.
	\end{itemize}
\end{lemma}

\begin{proof}
    The proof is similar to the proof of Lemma~\ref{thm_decomp}. 
    Let $L$ be the set of $m+\ceil{\eps^{-1}}$ farthest point to $C^\ast$ in $X$ and $Y = X\setminus L$. We also define $r=(\frac{\cost_z(Y,C^\ast)}{\abs{Y}})^{1/z}$. For each $i\in [k]$, let $Y_i=\{y\in Y\mid i=\argmin_{i\in [k]}\dist(y,c_i^\ast)  \}$ denote the set of inliers whose closest center is $c_i^\ast$. Observe that $Y_1, \dots, Y_k$ form a partition of $Y$.

    \emph{Balls} around $C^\ast$ are defined as $B_i=\{y\in Y\mid \dist^z(y,C^\ast)\leq 2^i\cdot \eps r^z  \}$. \emph{Rings} are defined as $R_0=B_0$ and $R_i=B_i\setminus B_{i-1}$ for $i=1,2,\dots,T$, where $T$ is the largest number such that $R_T\neq\emptyset$, which implies that $T\leq \log{\frac{\abs{Y}}{\eps}}$.

    Now we define the decomposition. Let
    \begin{itemize}[itemsep=0pt]
        \item $F=L$;
        \item $R=\bigcup_{j=T-s+1}^T R_j$ for $s=\min\{\ceil{z + \log (mk)},T\}$;
        \item $G=\bigcup_{j=0}^{T-s} R_j$.
    \end{itemize}

    By definition we have $\abs{F}=O(1/\eps+m )$. For each $j\in \{T-s+1,\dots,T\}$, $R_j$ is contained in the ball $B_j$, so $R_j$ is an $(r_j,k)$-instance with $r_j=(2^j\cdot \eps r^z)^{1/z}$. Hence we have

    \[
    \sum_{j=T-s+1}^T r_j^z \abs{R_j} =\sum_{j=T-s+1}^T 2^j  \eps r^z   \abs{R_j} \leq \sum_{j=T-s+1}^T 2\cost_z(R_j,C^\ast) \leq 2\cost_z(Y,C^\ast) \leq 2\cost[m]_z(X,C^\ast).
    \]

    It remains to prove that $G$ is an $(\eps (\frac{\cost[m]_z(X,C^\ast)}{mk})^{1/z},k) $-instance. Let $G_i=G\cap  Y_i$. It suffices to prove that $\diam(G_i)\leq \eps \cdot  (\frac{\cost[m]_z(X,C^\ast)}{mk})^{1/z}$.

    Since $X=Y\cup L$ and $L=O(m+\eps^{-1})$, which implies that $\cost[m]_z(X,C^\ast)$ aggregates at least $\eps^{-1}$ points outside $B_T$ and thus
    \[
        \eps^{z} \cdot \cost[m]_z(X,C^\ast)\geq \cost[m]_z(X,C^\ast)/\ceil{\eps^{z}} \geq 2^T \eps r^z. 
    \]   
    It follows that for each $i\in [k]$,
	\[
	   \diam(G_i) \leq 2\left(2^{T-s} \eps r^z\right)^{\frac{1}{z}}\leq 2 \left(\frac{2^T \eps r^z}{2^z mk}\right)^{\frac{1}{z}}\leq 2\left(\frac{\eps^z \cost[m]_z(X,C^\ast)}{2^zmk} \right)^{\frac{1}{z}} 
       = \eps\left( \frac{\cost[m]_z(X,C^\ast)}{mk} \right)^{\frac{1}{z}}.
	\]
    This completes the proof.
\end{proof}

We are ready to prove Theorem \ref{thm:kzC}, following a near-identical proof of Theorems~\ref{thm_gen} and~\ref{thm_doubling} in Section~\ref{sec:proof_thm}.

\begin{proof}[Proof of Theorem \ref{thm:kzC}]
    We first find a constant-factor approximation $C^\ast$ of $X$, i,e, $\cost[m]_z(X,C^\ast)\lesssim \min_{C\in \binom{M}{k}} \cost[m]_z(X,C) $. Then we apply Lemma~\ref{thm_decomp_kzC} to decompose $X$ into $X=F\cup R\cup G$, where $R=R_1\cup R_2\cup \dots \cup R_l$ for $l=O(\log mk)$ and each $R_i$ is an $(r_i, k)$-instance, $G$ is an $(\eps \cdot (\frac{\cost[m]_z(X,C^\ast)}{mk})^{1/z}, k)$-instance. Then we can construct $D$ as follows.

    \begin{itemize}[itemsep=0pt,topsep=0pt]
        \item We add $F$ into $D$ with unit weight for each $x\in F$.
        \item For each $R_i$ ($i\in [l]$), we apply Lemma~\ref{lemma_mainring_kzC} to construct an $(O(\eps),m,\eps r_i^z\cdot \abs{R_i} )$-robust coreset $D_i$ for $R_i$ and add $D_i$ to $D$.
        \item For $G$, we apply Lemma~\ref{lemma_innergroup_kzC} to construct an $(\eps,m,\eps\cdot \cost[m]_z(X,C^\ast))$-coreset $D_{\inn}$ and add $D_{\inn}$ to $D$.
    \end{itemize}

    We first bound the size of $D$. Note that $\norm{D}_0=\norm{F}_0+\sum_{i=1}^l\norm{D_i}_0+\norm{D_{\inn}}_0$, and $\norm{F}_0\leq O(m + \eps^{-1})$ by Lemma~\ref{thm_decomp_kzC}, $\norm{R_i}_0=\tilde{O}(kd\eps^{-2z})\cdot \log m$ by Lemma~\ref{lemma_mainring_kzC} and $\norm{D_{\inn}}_0=k+Q(\eps)=k+\tilde{O}(kd\eps^{-2})$ by Lemma~\ref{lemma_innergroup_kzC}. Thus we have the coreset size $\norm{D}_0=O(m)+\tilde{O}(kd\eps^{-2z})\log^3{m}=O(m)+ \tilde{O}(kd\eps^{-2z})$.

    Next we show that $D$ is an $(O(\eps),m)$-coreset of $X$ and the proof will be complete by rescaling $\eps$. Recall that $F$ is an $(\eps,m,0)$-coreset for $F$, each $D_i$ is an $(\eps,m,r_i^z \abs{R_i})$-coreset for $R_i$, and $D_{\inn}$ is an $(\eps,m,\eps\cdot \cost[m]_z(X,C^\ast)$-coreset for $G$. Since $\sum_{i=1}^lr_i^z \abs{R_i}=O(\cost[m]_z(X,C^\ast)) $, by Fact~\ref{fact_merge} and $\cost[m]_z(X,C^\ast)=O(\cost[m]_z(X,C))$, we see that $D$ is an $(\eps,m,O(\eps \cdot \cost[m]_z(X,C^\ast)))$-coreset for $X$, and is an $(O(\eps), m)$-coreset for $X$.
\end{proof}

\subsection{Proof of Lemma \ref{lemma_mainring_kzC}: First coreset for $(r,k)$-instances}
\label{sec:proof_1}

As preparation, we recall the following well-known generalized triangle inequality.

\begin{lemma}[\bf{Generalized triangle inequality~\cite[Lemma 2.1]{braverman2022power}}]
    \label{lm:triangle}
    Let $a,b,c\in X$ and $z\geq 1$.
    For every $t > 0$, the following inequalities hold:
    \begin{gather*}
    d^z(a,b) \leq (1+t)^{z-1} d^z(a,c) + \Bigl(1+\frac{1}{t}\Bigr)^{z-1} d^z(b,c),
    \shortintertext{and}
    \left| d^z(a,c) - d^z(b,c) \right| \leq t\cdot d^z(a,c) + \Bigl(1 + \frac{2z}{t}\Bigr)^{z-1}\cdot d^z(a,b).
    \end{gather*}
\end{lemma}

We first provide the following theorem that generalizes Theorem \ref{thm_main}.

\begin{theorem}[$(r,k_0)$-regular instances for general $z\geq 1$]
\label{thm_main_kzC} 
Let $(M,\dist)$ be a metric space with VC dimension $d_{\mathrm{VC}}$ and doubling dimension $d_D$.
Let $d = \min\{d_{\mathrm{VC}}, d_D\}$.
Assume that $X\subseteq M$ is an $(r,k_0)$-regular instance with $|X|=n$.
With probability $1 - O\big(\frac{1}{\log(k/\eps)\cdot \log (mk)}\big)$, Algorithm~\ref{alg_coreset} returns an $(\eps,m,O(\eps n r^z))$-robust coreset of $X$ for robust \kMedian with size $\frac{kd}{\eps^{2z}}\cdot \log^{O(1)} \frac{kd}{\eps}\cdot \log^2 m$.
\end{theorem}

Similar to Section \ref{sec_mainring}, we can prove Lemma \ref{lemma_mainring_kzC} via this theorem.

\begin{proof}[Proof of Lemma \ref{lemma_mainring_kzC}]

Let $X=X_1\cup X_2\cup \cdots \cup X_k$ be an $(r,k)$-regular instance where $X_i\subset B(a_i,r)$ and suppose that $\abs{X}=n$. Let $P_0=\{i\in [k]\mid 0<\abs{X_i} <\eps^z n/k\}$, $P_j=\{i\in [k]\mid 2^{j-1}\eps^z n/k \leq \abs{X_i} \leq 2^j \eps^z n/k    \}$ for all $j=1,2, \dots ,\ceil{\log(k/\eps^z)}$,  hence $P_j$'s form a partition of $[k]$.

Now we can construct the coreset $D$ for $X$. 

Let $Y_j=\bigcup_{i\in P_j} X_i$ for each $j=0, \dots, \ceil{\log(k/\eps^z)}$. By definition of $Y_j$ is an $(r,\abs{Y_j})$-regular instance for each $j\geq 1$. Applying Theorem~\ref{thm_main_kzC} to $Y_j$, we can construct an $(\eps,m,\abs{Y_j}\cdot \eps r^z)$-robust coreset $D_j$ for $Y_j$ with probability at least $1-O(\frac{1}{\log(k/\eps)\log{mk}})$ with $\abs{D_j}=\tilde{O}(\frac{kd}{\eps^z})\cdot \log^2{m}$. We include all $D_j$ in coreset $D$.

Next, it suffices to handle $Y_0$. We initialize $D_0 = \emptyset$, then for each $i\in P_0$, select an arbitrary point $x_i\in X_i$ and add $(x_i,\abs{X_i})$ to $D_0$. We add $D_0$ to $D$.

For every $i\in P_0$ and every $x \in X_i$, $C\in \binom{M}{k}$, we have from Lemma~\ref{lm:triangle} that
\begin{align*}
    \abs{\dist^z(x,C)-\dist^z(x_i,C)} \leq \eps \dist^z(x,C) + (\frac{3z}{\eps})^{z-1} \dist^z(x_i,x) 
    \leq \eps \dist^z(x,C) + (\frac{3z}{\eps})^{z-1}(2r)^z.
\end{align*}
It then follows that $D_0$ is an $(\eps,m,(\eps\cdot (3z)^{z-1}\cdot 2^{z}))$-robust coreset of $Y_0$.

By mergeablity of robust coreset, we conclude that $D$ is an $(\eps,m, (6z)^z\eps n r^z)$-robust coreset of $X$ for robust \kzC. The size of $D$ satisfies that
\[
    \norm{D}_0=k+\ceil{\log{\frac{k}{\eps^z}}}\cdot \tilde{O}(\frac{kd}{\eps^{2z}})\cdot \log ^2 m =\tilde{O}(\frac{kd}{\eps^{2z}})\cdot \log ^2 m.
\]

The overall failure probability is most $\ceil{\log{\frac{k}{\eps^z}}}\cdot O(\frac{1}{\log(k/\eps)\cdot \log(mk)}) =O(\frac{1}{\log {(mk)}})$. 
\end{proof}

It remains to prove Theorem \ref{thm_main_kzC}.
The algorithm is almost identical to Algorithm \ref{alg_coreset} except that the sample size is changed to $s = \frac{kd}{\eps^{2z}}\cdot \log^{O(1)} \frac{kd}{\eps}\cdot \log^2 m$.
Since the sample size is at least $\tilde{O}(\frac{kd}{\eps^2})$, the following properties still holds:
\begin{itemize}[itemsep=0pt]
\item The output $D$ is capacity-respecting.
\item $\mathbb{P}[\text{$\norm{U\cap X_i}_1 \in (1\pm \eps)\cdot n_i$ for all $i\in [k_0]$}]\geq 1 - O(\frac{\epsilon}{sk\log m})$, i.e., Lemma \ref{lemma_event} holds.
\end{itemize}

We use the definition of range space approximation as given in Definition \ref{def:eps_smoothed_range}.  
Additionally, the definition of index-subset cost approximation remains the same as in Definition \ref{def_subset}, except that \( \cost \) is replaced with \( \cost_z \).  
With these definitions, we establish the following lemma, which generalizes Lemma \ref{lemma:main:doubling}.
Note that the original distance function \( \dist \) is a 0-smoothed distance function.  
As a result, the following lemma also extends Lemma \ref{lemma_main_tech}.

\begin{lemma}[Sufficient condition for general $z\geq 1$] \label{lemma:main:kzC}
    Let $X$ be an $(r,k_0)$-regular instance of size $n$ and $D$ is a weighted subset of $X$. 
    If $D$ is simultaneously capacity-respecting, an $\eps^z$-range approximation with respect to an $\frac{\eps}{10z}$-smoothed distance function $\delta$, and an $\eps$-indexed-subset cost approximation, then $D$ is an $(O(\eps), m, O(\eps n r^z))$-robust coreset of $X$.
\end{lemma}

\begin{proof}
We suppose that $D$ is capacity-respecting, an $\eps^z$-range approximation with respect to an $\frac{\eps}{10z}$-smoothed distance function $\delta$, and an $\eps$-indexed-subset cost approximation of $X$. 

Fix a center set $C\in \binom{M}{k}$.
Recall that we define $r^\ast := \inf \{r>0\;\big|\; |B(C,r)\cap X|\geq n-t\}$. 
Note that for any center set \( C \), the outlier sets within \( X \) (or \( D \)) are identical with respect to both \(\cost^{(t)}(X,C)\) and \(\cost_z^{(t)}(X,C)\) (or similarly, \(\cost^{(t)}(D,C)\) and \(\cost_z^{(t)}(D,C)\)).
Thus, we can again partition \([k_0]\) into three sets \(A_1, A_2, A_3\) according to $r^\ast$, each retaining the same geometric properties as stated in Lemma \ref{proof_fact}.
We first have the following lemma that generalizes Lemma \ref{lemma:diffrerence:doubling}.

\begin{lemma}[Cost difference decomposition for $\delta$ for general $z\geq 1$]
\label{lemma:diffrerence:doubling_kzC}
If $D$ is capacity-respecting, it holds that 
\begin{equation*}
	\begin{aligned}
	    \abs{ \cost[t]_z(X,C)-\cost[t]_z(D,C) } &\leq \abs{ \cost_z(X_{A_1},C)-\cost_z(D_{A_1},C) } \\
	    &\quad + \int_{(r^\ast-2r)^+}^{r^\ast+2r} \abs*{ |B^\delta(C,u)\cap X_{A_3}|-\norm[\big]{B^\delta(C,u)\cap D_{A_3}}_1 } z u^{z-1} du \\
	    &\quad +\eps 2^{O(z)}\cost^{(t)}_z(X,C) +\eps 2^{O(z)}\cost[t]_z(D,C) + \eps 2^{O(z)} n r^z.
	\end{aligned}
\end{equation*}
\end{lemma}

\begin{proof}
Let $g=\abs{ X_{A_2}\cap B(C,r^\ast) }$, then $g\leq t$.
Since $\delta$ is an $\frac{\eps}{10z}$-smoothed distance function, we know that for any $x\in X$,
\[
(1-\eps) \dist^z(x,C)\leq (1-\frac{\eps}{10z})^z \dist^z(x,C) \leq \delta^z(x,C) \leq (1+\frac{\eps}{10z})^z \dist^z(x,C) \leq (1+\eps) \dist^z(x,C).
\]
Then using the same argument for Lemma \ref{lemma:diffrerence:doubling}, we obtain the following inequality:
\begin{align*}
\abs{ \cost[t]_z(X,C)-\cost[t]_z(D,C) } &\leq \abs{ \cost_z(X_{A_1},C)-\cost_z(D_{A_1},C) } \\
& \quad + (1+\eps)\abs{ \cost[t-g, \delta]_z(X_{A_3},C)-\cost[t-g, \delta]_z(D_{A_3},C) }.
\end{align*}
Thus, it suffices to prove the following inequality
\begin{align} \label{eq:cost_kzC}
\begin{aligned}
& \quad \abs{ \cost[t-g, \delta]_z(X_{A_3},C)-\cost[t-g, \delta]_z(D_{A_3},C) } \\
\leq & \quad \int_{(r^\ast-2r)^+}^{r^\ast+2r} \abs*{ |B^\delta(C,u)\cap X_{A_3}|-\norm[\big]{B^\delta(C,u)\cap D_{A_3}}_1 } z u^{z-1} du \\
	    &\quad +\eps 2^{O(z)} \cost^{(t)}_z(X,C) +\eps  2^{O(z)} \cost[t]_z(D,C) + \eps 2^{O(z)} n r^z.
\end{aligned}
\end{align}
To prove this, we first note that
	\begin{align*}
    	\cost[t-g,\delta]_z(X_{A_3},C)&= \int_{0}^{\infty} (\abs{X_{A_3}} - (t-g)-\vert B^\delta(C,u)\cap X_{A_3} \vert )^+ z u^{z-1} du, \\
	    \cost[t-g,\delta]_z (D_{A_3},C )&=\int_0^\infty (\norm{D_{A_3}}_1 - (t-g) - \norm{B^\delta(C,u)\cap D_{A_3}}_1)^+ z u^{z-1} du.
	\end{align*}
By a similar argument for Lemma \ref{lemma:diffrerence:doubling}, we obtain that
\begin{align}\label{eq:kzC_0}
\begin{aligned}
& \quad \abs{ \cost[t-g, \delta]_z(X_{A_3},C)-\cost[t-g, \delta]_z(D_{A_3},C) } \\
= & \quad \int_{(1-\frac{\eps}{10z})(r^\ast-2r)^+}^{(1+\frac{\eps}{10z})(r^\ast+2r)} \abs*{ |B^\delta(C,u)\cap X_{A_3}|-\norm[\big]{B^\delta(C,u)\cap D_{A_3}}_1 } z u^{z-1} du \\
\leq & \quad \int_{(r^\ast-2r)^+}^{r^\ast+2r} \abs*{ |B^\delta(C,u)\cap X_{A_3}|-\norm[\big]{B^\delta(C,u)\cap D_{A_3}}_1 } z u^{z-1} du \\
	    &\quad + \left((\abs{X_{A_3}}-(t-g))^+ + (\norm{D_{A_3}}_1 - (t-g))^+ \right) \left((1+\frac{\eps}{10z})^z(r^\ast +2r)^z - (r^\ast +2r)^z \right) \\
        & \quad + \left((\abs{X_{A_3}}-(t-g))^+ + (\norm{D_{A_3}}_1 - (t-g))^+ \right) \left((r^\ast -2r)^z - (1-\frac{\eps}{10z})^z(r^\ast - 2r)^z \right)
\end{aligned}
\end{align}

Note that 
\begin{align*}
\cost^{(t)}_z(X,C) \geq \cost[t-g, \delta]_z(X_{A_3},C) \geq & \ (\abs{X_{A_3}}-(t-g))^+ \cdot ((r^\ast - 2r)^+)^z \\
\cost^{(t)}_z(D,C) \geq \cost[t-g, \delta]_z(D_{A_3},C) \geq & \ \norm{D_{A_3}}_1 - (t-g))^+ \cdot ((r^\ast - 2r)^+)^z.
\end{align*}
Thus, given Inequality \eqref{eq:kzC_0}, it suffices to establish the following inequalities to prove Inequality \eqref{eq:cost_kzC}:
\begin{align*}
(1+\frac{\eps}{10z})^z(r^\ast +2r)^z - (r^\ast + 2r)^z \leq & \ \eps 2^{O(z)} (r^\ast - 2r)^z + \eps 2^{O(z)} r^z \\
(r^\ast -2r)^z - (1-\frac{\eps}{10z})^z(r^\ast - 2r)^z \leq & \ \eps 2^{O(z)} ((r^\ast - 2r)^+)^z + \eps 2^{O(z)} r^z.
\end{align*}
The following inequality follows from 
\[
(1+\frac{\eps}{10z})^z(r^\ast +2r)^z - (r^\ast +2r)^z \leq \eps (r^\ast +2r)^z = \eps (r^\ast-2r + 4r)^z 
\leq  \eps 2^{O(z)} (r^\ast - 2r)^z + \eps 2^{O(z)} r^z,
\]
and the second inequality can be proved similarly.
Thus, we complete the proof of Lemma \ref{lemma:diffrerence:doubling_kzC}.
\end{proof}

Now we go back to prove Lemma \ref{lemma:main:kzC}.
By the definition of range space approximation and indexed-subset cost approximation, we obtain that
    \[
	    \abs{ \cost[t]_z(X,C)-\cost[t]_z(D,C)} \leq \eps 2^{O(z)}\cdot (\cost[t]_z(X,C)+\cost[t]_z(D,C))+ \eps 2^{O(z)} nr^z,
    \]
    whence it follows that
    \[
	    \abs{ \cost[t]_z(X,C)-\cost[t]_z(D,C)} \leq \eps 2^{O(z)} \cost[t](X,C) + \eps 2^{O(z)} nr^z,
    \]
    which completes the proof of Lemma~\ref{lemma:main:kzC} after rescaling \(\eps\) by a factor of \(1/2^{O(z)}\).
\end{proof}

Similarly to the analysis in Section \ref{sec:lemma_main:range:doubling}, it is sufficient to show that with high probability, \( U \) is an \( \eps^z \)-range space approximation and an \( \eps \)-indexed-subset cost approximation.  
With this, Theorem \ref{thm_main_kzC} follows directly as a corollary of Lemma \ref{lemma:main:kzC}.
We consider two cases: $d = d_{\VC}$ and $d = d_D$.

\paragraph{VC instances: $d = d_{\VC}$.}
For the range space approximation, the key is to establish a generalization of Lemma \ref{lemma_range}, specifically, that with probability at least \( 1 - O\big(\frac{1}{\log \frac{k}{\eps}\cdot \log(mk)}\big) \), \( U \) is an \( \eps^z \)-range space approximation of \( X \).  
This result follows directly since the coreset size is increased to \( \tilde{O}\big(\frac{k}{\eps^{2z}}\big) \), and the remainder of the proof mirrors that of Section \ref{sec:alg_guarantee_proof}, substituting \( \eps \) with \( \eps^z \).

Hence, it remains to prove Lemma \ref{lemma_subset} (replacing $\cost$ by $\cost_z$), i.e., with probability $1-O\big(\frac{1}{\log(k/\eps)\cdot \log (mk)}\big)$, $U$ is an $\eps$-indexed-subset cost approximation.
Similarly, the proof employs a chaining argument developed in recent works \cite{fefferman2016testing,Cohen2022Towards,Cohen25}. 
Below, we only highlight the key technical differences from these works.

For the net construction, we continue to use the notation of the distance vector \( v^S_{i,J,C} \) defined in Equation~\eqref{eq:vS} by replacing $\dist$ with $\dist^z$.  
The only modification is that we set \( T = \lceil z\log(16z/\eps) \rceil \), define \( A_0 = [0, 2^{1/z} r] \), and \( A_i = (2^{i/z} r, 2^{(i+1)/z} r] \) for each \( i \in [T] \).  
We still define $I_{i,C}=\{j\in [k_0]\mid \dist(a_j,C)\in A_i\}$ for each $i\in \{0,\dots,T\}$.
This adjustment ensures that for every $j\notin \bigcup_{i=0}^T I_{i,C}$, we have $\dist(a_j, C) > \frac{16z r}{\eps}$.
This implies that for any $x,y\in B(a_j, r)$, $\dist^z(x, C) \in (1\pm \eps)\dist^z(y, C)$.
Thus, such a ball only induces a multiplicative error, which can be safely ignored in the chaining argument.
Thus, we only need to consider the following collections of distance vectors: for $0\leq i \leq T$, define the level-$i$ distance vector set $V^S_{i}=\{v^S_{i,J,C}\mid J\subseteq [k_0], C\in \binom{M}{k}\}$.
Consequently, Lemma \ref{lemma_relate_main} still holds, and it suffices to prove Inequality \eqref{inq_chaining_main}.

To this end, we adapt the notion of nets as follows: A subset $N\subseteq V^S_i$ is called an $\alpha$-net of $V^S_i$ if for every $i,J,C$, there exists a vector $v\in N$ such that
\[
\begin{cases}
|v(x)-v^S_{i,J,C}(x)|\leq \alpha\cdot 2^{i+2}\cdot (r\cdot  (v^S_{i,J,C}(x))^{\frac{z-1}{z}} + r^z) & v^S_{i,J,C}(x)>0\\
0 & v^S_{i,J,C}(x)=0
\end{cases}
\]
The key is to bound the net size by the following lemma, which generalizes Lemma~\ref{lemma_netsize_main}.

\begin{lemma}[Bounding net size on VC instances for general $z\geq 1$] 
\label{lemma_netsize_main_kzC}
For every $i\in \{0,\dots,T\}$, $\alpha > 0$ and every subset $S\subseteq X$ with $\abs{S}\geq 2$, there exists an $\alpha$-net $N^S_{\alpha,i}$ of $V^S_i$ such that
\[
\log \abs[big]{N^S_{\alpha,i}} \lesssim \alpha^{-1} d_{\VC}\cdot  k\log k \cdot \log\abs{S}.
\]
Furthermore, if $\alpha\geq 2^{-i}\cdot z2^{z-1}$, it is possible to ensure that $N_{\alpha,i}^S$ consists of piecewise constant vectors. That is, every $v\in N_{\alpha,i}^S$ satisfies that $v(x)$ is constant for all $x\in X_j$ for each $j\in [k_0]$.
\end{lemma}

\begin{proof}

The net construction is similar to the construction in Lemma~\ref{lemma_netsize}. 
The only difference is that we let $\beta = \alpha 
2^{i} r^z$. 
By definition it is clear that $\tilde{v}_{iJ,C}^S(x) \in \{0,\beta,...,t\beta  \}$ for $t=\ceil{2 /\alpha}$. 
By the same argument, we can still construct an $\alpha$-net $N_{\alpha,i}^S$ for $V_i^S$ with size
\[
\abs{N_{\alpha,i}^S} \leq \abs{S}^{O(\alpha^{-1} kd \log k)} \leq \abs{S}^{O(\alpha^{-1} kd \log k)}.
\]
%

%
Next, we prove that if $\alpha\geq 2^{-i}\cdot z2^{z-2}$, it is possible to ensure that $N_{\alpha,i}^S$ consists of piecewise constant vectors. 
Define for $u\in N_{\alpha,i}^S$
\[
\tilde{u}(x)=
\begin{cases}
    \ceil{\frac{\dist^z(\alpha_{\pi(x)},C)}{\beta}}\cdot \beta \ &u(x)>0 \\
    0 \ &u(x)=0
\end{cases}
\]
For every $v_{i,J,C}^S\in V_i^S$ there exists $u\in N_{\alpha,i}^S$ such that $\abs{u(x)-v_{i,J,C}^S(x)}\leq \beta $ for all $x\in S$. Then 
\begin{align*}
    \abs{v_{i,J,C}^S(x) - \tilde{u}(x)} &\leq \abs{\dist^z(x,C) - \dist^z(a_{\pi(x)},C)} + \abs{\dist^z(a_{\pi(x)},C) - \tilde{u}(x)} \\
    &\leq \abs{\dist^z(x,C)-\dist^z(a_{\pi(x)},C)} + \beta.
\end{align*}
\sloppy
Now we bound $\abs{\dist^z(x,C) - \dist^z(a_{\pi(x)},C)}$. By the mean-value theorem, there exists $\xi\in (\dist(x,C), \dist(a_{\pi(x)},C)$ such that
\[
	\dist^z(x,C) - \dist^z(a_{\pi(x)},C) = z\cdot (\dist(x,C) - \dist(a_{\pi(x)},C)) \cdot \xi^{z-1},
\]
hence it follows that
\begin{equation}\label{eqn:dist^z_to_center_difference_kzC}
	\abs{\dist^z(x,C)-\dist^z(a_{\pi(x)},C)}
    \leq zr (\dist(x,C) + r)^{z-1}\\
    \leq z2^{z-1}(r\dist^{z-1}(x,C) + r^z).
\end{equation}
Thus, by our assumption on $\alpha$,
\[
	\abs{v_{i,J,C}^S(x) - \tilde{u}(x)} \leq \alpha 2^{i+2} (r (v_{i,J,C}^S(x))^{\frac{z-1}{z}} + r^z ).
\]
The remainder of the proof follows exactly as in the proof of Lemma~\ref{lemma_netsize}.
\end{proof}

We combine this net size with a chaining argument to obtain the following inequality, which is analogous to Inequality \eqref{inq_chaining_main}.  
\begin{equation*} 
\sup_{J,C} \frac{n}{s}\cdot \abs*{ \frac{\sum_{x\in U} v^U_{i,J,C}(x)-\E \sum_{x\in U} v^U_{i,J,C}(x)}{(1+2^i|J_{i,C}|/k_0)\cdot nr^z} }
\leq \frac{\eps}{\log(z/\eps)},
\end{equation*}
This can be shown using a similar approach as in the proof of Lemma~\ref{lemma_chaining}. 
The bound for the expected value on the right-hand side of \eqref{eqn:chaining_error_first_term_aux} (with a modified expression) becomes \( \sqrt{s/2^{i/z}} \) (cf.\ \cite[Section 4.2]{Cohen25}) while the bound for the expected value on the rightmost-hand of \eqref{eqn:chaining_error_second_term_semifinal} remains unchanged. 
In the end, 
we can conclude that $|D| = |U| = \tilde{O}( k d \eps^{-z-1})$ suffices to ensure that $U$ is an $\eps$-indexed-subset cost approximation.

\paragraph{Doubling instances: $d = d_D$.}
The argument is similar to the VC instances.
We can prove Lemma \ref{lemma:eps_range} with the error term $\eps^z$ for the range space approximation.
Again, it follows directly since the coreset size is increased to \( \tilde{O}\big(\frac{k}{\eps^{2z}}\big) \).
For the indexed-subset cost approximation, it suffices to bound the net size.
To this end, we give the following lemma, which generalizes Lemma \ref{lemma_netsize_doubling}.

\begin{lemma}[Bounding net size on doubling instances for general $z\geq 1$] \label{lemma_netsize_doubling_kzC}
    For every $i\in [T]$, $\alpha\in (0,1)$ and every subset $S\subseteq X $ with $\abs{S}\geq 2$, there exists an $\alpha$-net $N_{\alpha,i}^{S}$ for $V_{i}^S$ such that
    \[
        \log {\abs{N_{\alpha,i}^S}}\lesssim k d \log (zk/\alpha).
    \]
    Furthermore, if $\alpha \geq 2^{-i} \cdot 6^{z-1}$, it is possible to ensure that $N_{\alpha,i}^S$ consists of piecewise constant vectors. That is, every $v\in N_{\alpha,i}^S$ satisfies that $v(x)$ is constant for all $x\in X_j$ for each $j\in [k_0]$.
\end{lemma}

\begin{proof}
For each $i\in [T]$, let $\alpha_i = \Gamma^{-1}\alpha 2^{i/z} r$, where $\Gamma = 6z$. Let $N_{i,j}$ be an $\alpha_i$-net for $B(a_j, 2^{(i+1)/z} r)\setminus B(a_j, 2^{i/z} r)$. It follows by the property of doubling dimension that
\[
    \abs{N_{i,j}}\leq 
    \left(k \cdot \frac{\Gamma\cdot 2^{(i+1)/z}\cdot r}{\alpha \cdot 2^{i/z}\cdot r} \right)^d
    \leq \left(\frac{2\Gamma k}{\alpha}\right)^{d}. 
\]
Fix $J\subseteq [k_0]$ and consider a distance vector $v_{i,J,C}^S\in V_i^S$ where $C=\{c_1,\dots,c_k\}$. For each $j\in J$, there exists $\tilde{c}_j \in N_{i,j}$ such that $\dist(c_j,\tilde{c}_j) \leq \alpha_i$. Define $\tilde{v}_{i,J,C}:S\rightarrow \mathbb{R}_{\geq 0}$ as
\[
    \tilde{v}_{i,J,C}^S(x)= 
    \begin{cases}
        \dist^z(x,\tilde{C}), \ &v_{i,J,C}^S(x) >0 \\
         0                      \ &v_{i,J,C}^S(x) =0 
    \end{cases}
\]
By Lemma~\ref{lm:triangle}, it follows that
\begin{align*}
    \abs{v_{i,J,C}^S(x) - \tilde{v}_{i,J,C}^S(x)}&\leq \abs{\dist^z(x,c_j)- \dist^z(x, \tilde{c}_j) }  \\
    &\leq t\cdot \dist^z(x,c_j) + \left(1 + \frac{2z}{t}\right)^{z-1}  \dist^z(c_j,\tilde{c}_j) \\
    &\leq t\cdot \dist^z(x,c_j) + \left(1 + \frac{2z}{t}\right)^{z-1}\alpha_i^z. 
\end{align*}
Let $t = \beta\alpha_i/\dist(x,c_j)$, where $\beta = 1$ when $z\leq 3/2$ and $\beta = (z-1)^{1/z}(3z)^{1-1/z}$ when $z>3/2$. 
If $t > z$, then $1+2z/t \leq 3$ and it follows that
\begin{align*}
    \abs{v_{i,J,C}^S(x) - \tilde{v}_{i,J,C}^S(x)} 
    &\leq \beta \alpha_i (\dist(x,c_j))^{z-1} + 3^{z-1} \alpha_i^z \\
    &\leq \frac{\beta}{\Gamma} \alpha 2^{i/z} r(v_{i,J,C}^S(x))^{\frac{z-1}{z}} + \frac{(3\alpha)^{z-1}}{\Gamma^z} \alpha 2^{i} r^z \\
    &\leq \alpha\cdot 2^{i}\cdot (r\cdot (v_{i,J,C}^S(x))^{\frac{z-1}{z}} + r^z).
\end{align*}
Otherwise, $1+2z/t \leq 3z/t$ and it follows that
\[    
    \abs{v_{i,J,C}^S(x) - \tilde{v}_{i,J,C}^S(x)} 
    \leq \gamma \alpha_i (\dist(x,c_j))^{z-1}  
    \leq \alpha\cdot 2^{i/z}\cdot r\cdot (v_{i,J,C}^S(x))^{\frac{z-1}{z}},
\]
where $\gamma = 1+(3z)^{z-1}$ when $z\leq 3/2$ and $\gamma = \frac{z}{z-1}(3z)^{1-1/z}$ when $z > 3/2$; in either case, $\gamma\leq 3z$.

We can then proceed as in the proof of Lemma~\ref{lemma_netsize_doubling} to define $N_{\alpha,i}^S$ and obtain that 

\[
    \log{\abs{N_{\alpha,i}^S}}\lesssim kd\log{\frac{\Gamma k}{\alpha}} \lesssim kd\log{\frac{z k}{\alpha}}.
\]
For the second part of the lemma, we use the same argument as in the proof of Lemma~\ref{lemma_netsize_doubling} with the following modification. For $u\in N_{\alpha,i}^S$, we define $\tilde{u}$ as
    \[
        \bar{u}(x) = \begin{cases}
                    \dist^z(a_{\pi(x)}, \bar{C}),  & \bar{u}(x) > 0 \\
                    0,                      & \bar{u}(x) = 0.
                  \end{cases}
    \]
    Then for a vector $v_{i,J,C}^S$, suppose that it is closest to a net point $u\in N_{\alpha, i}^S$. We can then bound in a similar manner to before
    \begin{align*}
        \abs{v_{i,J,C}^S(x) - \bar{u}(x)} &\leq t\dist^z(x,c_j) + \left(1 + \frac{2z}{t}\right)^{z-1}(\alpha_i + r)^{z}.
    \end{align*}
    From similar calculations to the above, we have either
    \begin{align*}
        \abs{v_{i,J,C}^S(x) - \bar{u}(x)} 
        &\leq \beta (\alpha_i+r) (\dist(x,c_j))^{z-1} + 3^{z-1} (\alpha_i+r)^z \\
        &\leq (3z)\left(\frac{1}{\Gamma} \alpha 2^{i/z} + 1\right)r(\dist(x,c_j))^{z-1} + 3^{z-1}2^{z-1}(\alpha_i^z + r^z) \\
        &\leq \alpha 2^{i+2} r(\dist(x,c_j))^{z-1} + \left(\alpha 2^i + 6^{z-1}\right)r^z \\
        &\leq \alpha 2^{i+2} \left( r(\dist(x,c_j))^{z-1} + r^z\right)
    \end{align*}
    or
    \[
        \abs{v_{i,J,C}^S(x) - \bar{u}(x)}
        \leq (3z)(\alpha_i + r) (\dist(x,c_j))^{z-1}
        \leq \alpha 2^{i+2} r(\dist(x,c_j))^{z-1}.
    \]
    Thus, we always have
    \[
        \abs{v_{i,J,C}^S(x) - \bar{u}(x)}  \leq \alpha 2^{i+2} \left( r(\dist(x,c_j))^{z-1} + r^z\right).
    \]
    The remainder of the proof is the identical to that of Lemma~\ref{lemma_netsize_doubling}.
\end{proof}

Hence, similar to VC instances, we can prove that $|D| = |U| = \tilde{O}( k d \eps^{-z-1})$ suffices to ensure that $U$ is an $\eps$-indexed-subset cost approximation.

\begin{remark}[Further size improvement]
\label{remark:open}
Note that \( \tilde{O}( k d \eps^{-z-1}) \) samples suffice for the indexed-subset cost approximation.  
Therefore, the main bottleneck lies in the range space approximation.  
If we could further reduce the required sample size for range space approximation to \( \tilde{O}( k d \eps^{-z-1}) \), we would be able to improve the robust coreset size further, aligning it with the vanilla case.
\end{remark}

\subsection{Proof of Lemma \ref{lemma_mainring_kzC}: Second coreset for $(r,k)$-instances}
\label{sec:proof_2}

The algorithm is identical to Algorithm \ref{alg_coreset2}, except that \(\mathcal{B}\) outputs a coreset for vanilla \((k,z)\)-clustering.  
The proof extends that of Lemma \ref{lemma_mainring}.  

Let $D_H=\{(x_i,|H_i|)\mid i\in [k]\}$ denote the set of $k$ weighted points added into $D$ in the for loop of Algorithm \ref{alg_coreset2}. 
We prove that $D=D_H\cup D'$ is an $(\eps,m,O(mkr^z/\eps^{z-1}))$-robust coreset of $X=H\cup X'$.
Fix a center set $C\in \binom{m}{k}$.

We first argue, by losing an additive error of $\eps \cost_z^{(t)}(X, C) + O(mkr^z/\eps^{z-1})$, we can assume that the outliers in $\cost[t]_z(H\cup X', C)$ are always chosen from $H$ and the outliers in $\cost[t]_z(D_H\cup D', C)$ are always chosen from $D_H$. 
We use the same notation \( J^\ast \), \( J \), \( J' \), \( E^\ast \), \( E \), and \( E' \) as introduced in the proof of Lemma \ref{lemma_mainring}.
Note that $J', H\setminus J \subseteq X\setminus J^\ast$ and $D', D_H\setminus E \subseteq D\setminus E^\ast$.
Also, we have
	\begin{align*}
		& \ \cost[t]_z(D,C) - \cost[t]_z(X,C) \\
        \leq & \ \cost_z(D\setminus E^\ast,C) - \cost[t]_z(X,C)\\
        \leq & \ \abs{\cost_z(D_H\setminus E, C) - \cost_z(H\setminus J,C)} + \abs{\cost_z(D', C) - \cost_z(X',C)} \\
        &\hspace{2cm} + \abs{ \cost_z(E',C) - \cost_z(E^\ast\setminus D_H,C) } +\abs{ \cost_z(J',C) - \cost_z(J^\ast\setminus H,C) }.
	\end{align*}
We shall bound the four terms separately.
The main difference from the case of $z = 1$ is that we need to use the generalized triangle inequality (Lemma \ref{lm:triangle}).

	For the first term, 
	\begin{align*}
	\abs{\cost_z(D_H\setminus E, C) - \cost_z(H\setminus J, C)}
	\leq & \ \sum_{i\in [k]} \abs{(|H_i|-t_i)\dist^z(x_i, C) - \cost_z(H_i\setminus J, C)} \\
    \leq & \ \sum_{i\in [k]} \sum_{x\in H_i\setminus J} \eps \dist^z(x, C) + \frac{(3z)^{z-1}}{\eps^{z-1}} \dist^z(x_i, x) \\
	\leq & \ \eps \sum_{i\in [k]} \cost_z(H_i\setminus J, C) + z^{O(z)} m k r^z/\eps^{z-1} \\
     \leq & \ \eps \cost[t]_z(X, C) + z^{O(z)} m k r^z/\eps^{z-1},
	\end{align*}
	where the second inequality follows from Lemma \ref{lm:triangle}, the third inequality follows from the fact that \(\sum_i |H_i| \leq mk\), due to the construction of \(H\), and the last inequality follows from the inclusion \(H\setminus J \subseteq X\setminus J^\ast\) and the definition of \(J^\ast\).

    For the second term, since $D'$ is a vanilla $\eps$-coreset of $X'$, we have
	\begin{align*}
		\abs{\cost_z(D', C) - \cost_z(X',C)} 
		&\leq \eps \cost_z(X', C) \\
		&\leq \eps \cost_z(X\setminus J, C) \\
		&\leq \eps \cost_z(X\setminus J^\ast, C) + \eps\abs{ \cost(J, C) - \cost(J^\ast, C) } \\
		&\leq \eps \cost[t]_z(X, C) + \eps\abs{ \cost_z(J',C) - \cost_z(J^\ast\setminus H,C) }.
	\end{align*}

    For the third term,
        \begin{align*}
	\abs{ \cost_z(E',C) - \cost_z(E^\ast\setminus D_H,C) }
	\leq & \ \sum_{i\in [k]} \abs{(t_i - |E^\ast\cap D_H|)\dist^z(x_i, C) - \cost_z(E^\ast_i \setminus D_H, C)} \\
    \leq & \ \sum_{i\in [k]} \sum_{x\in E^\ast_i \setminus D_H} \eps \dist(x_i, C) + \frac{(3z)^{z-1}}{\eps^{z-1}} \dist^z(x_i, x) \\
	\leq & \ \eps \cost_z(E',C) + z^{O(z)} m r^z/\eps^{z-1} \\
     \leq & \ \eps \cost_z(D\setminus E^\ast, C) + z^{O(z)} m r^z/\eps^{z-1},
	\end{align*}
        where the second inequality follows from Lemma \ref{lm:triangle}, the third inequality follows from the fact that \(|E'| \leq m\), and the last inequality follows from the inclusion \(E' \subseteq D\setminus J^\ast\).

    For the last term, 
        \begin{align*}
	\abs{ \cost_z(J',C) - \cost_z(J^\ast\setminus H,C) }
	\leq & \ \sum_{i\in [k]} \abs{(t_i - |J^\ast\cap H|)\dist^z(x_i, C) - \cost_z(E^\ast_i \setminus D_H, C)} \\
    \leq & \ \sum_{i\in [k]} \sum_{x\in J^\ast\cap H} \eps \dist(x_i, C) + \frac{(3z)^{z-1}}{\eps^{z-1}} \dist^z(x_i, x) \\
	\leq & \ \eps \cost_z(J',C) + z^{O(z)} m r^z/\eps^{z-1} \\
     \leq & \ \eps \cost[t]_z(X, C) + z^{O(z)} m r^z/\eps^{z-1},
        \end{align*}
        where the second inequality follows from Lemma \ref{lm:triangle}, the third inequality follows from the fact that \(|J'| \leq m\), and the last inequality follows from the inclusion \(J' \subseteq X\setminus J^\ast\) and the definition of $J^\ast$.
    
	Therefore,
	\begin{align*}
		\cost[t]_z(D,C) - \cost[t]_z(X,C) 
        \leq & \ \cost_z(D\setminus E^\ast,C) - \cost[t]_z(X,C)\\
		\leq & \ \eps \cost^{(t)}(X, C) + \eps \cost[t]_z(X, C) + z^{O(z)} m k r^z/\eps^{z-1} \\
        & \ + \eps \abs{ \cost_z(J',C) - \cost_z(J^\ast\setminus H,C) } + \eps \cost_z(D\setminus E^\ast, C) \\
		\leq & \ 3\eps \cost^{(t)}(X, C) + \eps \cost_z(D\setminus E^\ast, C) + z^{O(z)} m k r^z/\eps^{z-1}. 
	\end{align*}
        It implies that
        \begin{align*}
        (1 - \eps) \cost[t]_z(D,C) \leq & \ (1 - \eps) \cost_z(D\setminus E^\ast, C) \\
        \leq & \ (1 + 3\eps) \cost^{(t)}(X, C) + z^{O(z)} m k r^z/\eps^{z-1}.
        \end{align*}
        Similarly, we can prove that
        \[
        (1 - \eps) \cost[t]_z(X,C) \leq (1 + 3\eps) \cost^{(t)}(D, C) + z^{O(z)} m k r^z/\eps^{z-1}.
        \]
        Thus, $D$ is an $(\eps,m,O(mk r^z/\eps^{z-1}))$-robust coreset of $X$, which completes the proof of Lemma \ref{lemma_mainring_kzC}.

\section{Extension to General $z\geq 1$ for Euclidean Instances}
\label{sec:Euc_kzC}

The following theorem extends Theorem \ref{thm_Euc} to general robust \kzC for constant $z\geq 1$.

\begin{theorem}[Euclidean spaces for general $z\geq 1$] \label{thm_Euc_kzC}
Let $(M,\dist)=(\mathbb{R}^d,\|\cdot\|_2)$. 
There exists an algorithm that, given a dataset $X\subset \mathbb{R}^d$ of size $n\geq 1$, constructs an $(\eps,m)$-robust coreset of 
$X$ for the robust \kzC with size $O(m \eps^{-z})+\tilde{O}(\min\{k^{\frac{2z+2}{z+2}}\eps^{-2}, k\eps^{-z-2}\})$ in $O(nk)$ time.
\end{theorem}

The algorithm we used here is nearly identical to that provided in Section \ref{proof_Euc}, with the following extensions:
\begin{itemize}[itemsep=0pt]
\item Regarding Lemma \ref{thm_euc_decomp}, we let $F$ be the collection of $m + z^{O(z)} m\eps^{-z}$ points whose distances to $A$ are the largest.
Let $G = X\setminus F$ be the collection of remaining points;
\item The sample size of $D$ for $G$ is updated to \( s \leftarrow \min(k^{(2z+2)/(z+2)} \eps^{-2},\; k \eps^{-z-2}) \cdot \log^{O(1)}(k/\eps) \);  
\item For constructing $D\subseteq G$ by Algorithm \ref{alg_coreset3}, the cost function \( \cost \) is replaced by \( \cost_z \);  
\item \( \Delta_i \) in Algorithm \ref{alg_coreset3} is redefined as \( \Delta_i := \frac{\cost_z(X_i, A)}{|X_i|} \).
\end{itemize}

The core of the analysis is to establish the following lemmata, which are analogous to Lemmata \ref{proof_Euc_order} and \ref{lemma_Euc_tech}.  
Their proofs are provided in the subsequent sections.

\begin{lemma}[Robust cost decomposition for general $z\geq 1$]
\label{lemma_Euc_order_kzC} 
If $D$ is a capacity-respecting weighted subset of $X$, we have that for every $C\in \binom{\mathbb{R}^d}{k}$, and every $t\in [0,m]$, there exists an index $j_t^C\in [k]$, an index subset $I_t^C\subseteq [k]$, and a constant $\alpha_t^C\in [0,1]$ such that
	\begin{gather*} 
		|\cost^{(t)}_z(X,C)- \alpha_t^C\cost_z(X_{j_t^C},C)-\cost_z(X_{I_t^C},C)|\leq \eps \cost^{(t)}_z(X,C) +  z^{O(z)} mr^z/\eps^{z-1}  \\
		|\cost^{(t)}_z(D,C)- \alpha_t^C\cost_z(D_{j_t^C},C)-\cost_z(D_{I_t^C},C)|\leq \eps \cost^{(t)}_z(D,C) +  z^{O(z)} mr^z/\eps^{z-1}
	\end{gather*}
\end{lemma}

\begin{proof}

    The proof is almost identical to that of Lemma \ref{proof_Euc_order}. By the generalized triangle inequality,
\[
\abs*{\frac{\cost_z(X_{i_l},C)}{\abs{X_{i_l}}} - \dist^z(a_{i_l},C) } \leq \eps\frac{\cost_z(X_{i_l},C)}{\abs{X_{i_l}}} + \frac{(3z)^{z-1}}{\eps^{z-1}} r^{z}.
\]
Suppose that $x\in X_{i_j}$ for some $j\geq l$. Then, again by the generalized triangle inequality,
\[
	\abs{\dist^z(a_j, C) - \dist^z(x, C)} \leq \eps \dist^z(a_j, C) + \frac{(3z)^{z-1}}{\eps^{z-1}} r^{z}.
\]
For notational simplicity, let $\Delta_\eps = ((3z)^{z-1}/\eps^{z-1}) r^{z}$. Then
\begin{multline*}
		 (1+\eps)^2\frac{\cost_z(X_{i_l},C)}{\abs{X_{i_l}}} + (2+\eps)\Delta_\eps
	=	 (1+\eps)\left[ (1+\eps)\frac{\cost_z(X_{i_l},C)}{\abs{X_{i_l}}} + \Delta_\eps \right] + \Delta_\eps \\
	\geq (1+\eps)\dist^z(a_{i_l},C) + \Delta_\eps
	\geq (1+\eps)\dist^z(a_{i_j},C) + \Delta_\eps
	\geq \dist^z(x,C).
\end{multline*}
This implies that
\[
	\operatorname{Cut}(L,C,\beta) \leq \frac{\cost_z(X_{i_l},C)}{\abs{X_{i_l}}} + \left(3\eps\frac{\cost_z(X_{i_l},C)}{\abs{X_{i_l}}} + (2+\eps)\Delta_\eps\right).
\]
The claimed result will then follow from the same argument with a rescaling of $\eps$.
\end{proof}

\begin{lemma}[Strong indexed-subset cost approximation for general $z\geq 1$]
\label{lm:strong_kzC}
With probability at least 0.9, the output $D$ is an $\eps$-strong indexed-subset cost approximation of $X$.
Here, we replace $\cost$ with $\cost_z$ in Definition \ref{def_strong_subset}.
\end{lemma}

Now we are ready to prove Theorem \ref{thm_Euc_kzC}.

\begin{proof}[Proof of Theorem \ref{thm_Euc_kzC}]


%
It remains to prove the correctness of the algorithm. Let $\Gamma = z^{O(z)}$, the same coefficient in the error bound of Lemma~\ref{lemma_Euc_order_kzC}.
Observing that for each point \( x \in G \),  
\[
d(x, C^\ast) \leq \left( \frac{\cost[m](X, C)}{\Gamma \cdot m \eps^{-z}} \right)^{1/z}.
\]  
This implies that \( G \) forms an \( \left( \left( \frac{\cost^{m}(X, C^\ast)}{\Gamma m \eps^{-z}} \right)^{1/z},\; k \right) \)-instance.  

By Lemma~\ref{lemma_Euc_order_kzC} and the assumption that $D$ is an $\eps$-strong indexed-subset cost approximation of $X$ (which is centered at $A$), we have that for every $t\in [0,m]$ and $C\in \binom{\mathbb{R}^d}{k}$, 
\begin{align*}
            &\abs*{\cost[t]_z(X,C)-\cost[t]_z(D,C)} \\
    \leq{}  & \eps\cdot \cost[t]_z(X,C)+\eps\cdot \cost[t]_z(D,C)+2\Gamma mr^z/\eps^{z-1}\\
            &\quad + \alpha_t^C\abs{\cost_z(X_{j_t^C},C)-\cost_z(D_{j_t^C},C)} + \abs{\cost_z(X_{I_t^C,C})-\cost_z(D_{I_t^C,C})} \\
    \leq{}  & \eps\cdot \cost[t]_z(X,C)+\eps\cdot\cost[t]_z(D,C) + 2\Gamma mr^z/\eps^{z-1}\\
            & \quad + \eps(\cost_z(X_{I_t^C},C) + \alpha_{t}^C\cdot \cost_z(X_{j_t^C},C)) + \eps\cost_z(X,A) \\
    \leq{}  & 2\eps\cost[t]_z(X,C) + \eps\cost[t]_z(D,C) + 2\Gamma mr^z/\eps^{z-1} + \eps\cost_z(X,A).
\end{align*}
This implies that
\[
    \abs*{\cost[t]_z(X,C)-\cost[t]_z(D,C)} \leq 6\eps\cost[t]_z(X,C) + 4\Gamma mr^z/\eps^{z-1} + 2\eps\cost_z(X,A).
\]
and thus $D$ is a $(6\eps, m , 4\Gamma\cdot mr^z/\eps^{z-1} + 2\eps\cost_z(X,A))$-robust coreset of $X$, analogous to Lemma \ref{lem_euc_coreset_guarantee}. 

The coreset $D$ is constructed as follows. 

First, we can use the variation of Lemma~\ref{thm_euc_decomp} with the extensions above in Theorem~\ref{thm_decomp_kzC} to compute a decomposition $X=F\cup G$ and add $F$ identically to $D$. Then we use the extension of Algorithm~\ref{alg_coreset3} to compute a weighted subset $D_0$, which is a $(6\eps, m , 4\Gamma\cdot mr^z/\eps^{z-1} + 2\eps\cost_z(G,C^\ast))$-robust coreset of $G$, and add $D_0$ to $D$. 
Since \( G \) is a \( \Big( \left( \frac{\cost[m](X, C^\ast)}{\Gamma m \eps^{-z}} \right)^{1/z},\; k \Big) \)-instance, $D_0$ is an $(6\eps,m,6\eps\cost[m]_z(X,C^\ast))$-robust coreset of $G$. 
We know that $D=D_0\cup F$ is an $(O(\eps), m, O(\eps) \cost[m]_z(X,C^\ast))$-robust coreset of $X=F\cup G$. 

Here $C^\ast$ is a constant approximation of $X$ for robust \kzC with $m$ outliers, so $D$ is an $(O(\eps),m)$-robust coreset of $X$. Rescaling $\eps$ completes the proof.
\end{proof}

\subsection{Proof of Lemma \ref{lm:strong_kzC}: Strong indexed-subset cost approximation}

Recall that $(U, w_U)$ is a vanilla coreset constructed by the importance sampling procedure of \cite{bansal2024sensitivity}.
We similarly define $\mathcal{E}_0$ and $\mathcal{E}_1$ by replacing $\cost$ with $\cost_z$, and modify Lemma~\ref{lm:event_prob} so that $\mathbb{P}(\mathcal{E}_0\cup\mathcal{E}_1)\geq 1 - \eps^z/k$ for $s\gtrsim k\eps^{-2}\log(k\eps^{-1})$, provided that the hidden constant in $s$ is chosen large enough and depends on $z$.
With this setup, the key is to prove the following lemma.

\begin{lemma}[Properties of $U$ for general $z\geq 1$] 
\label{lemma_Euc_U_subset_kzC}
With probability at least $1-\frac{1}{10\log(k/\eps)}$, $U$ is an $O(\eps)$-strong indexed-subset cost approximation of $X$.
\end{lemma}

Similar to Section~\ref{sec_Euc_analysis}, we define bands and levels.

\paragraph{Bands} Now, let $T = (\eps^z/k) \cost_z(X,A)$ and $b_{max}=\lceil \log k + z\log \frac{1}{\eps}\rceil$. For $j\in [b_{max}]$, a subset $J\subseteq [k]$ is called a Band-$j$ subset if $\cost_z(X_i,A)\in (2^{j-1}T,2^jT]$ for all $i\in J$. In addition, we call $J$ a Band-$0$ subset if $\cost_z(X_i,A)\leq T$ for all $i\in J$. For $j=0,\dots,b_{max}$, let $\mathcal{B}_j=\{J\subseteq [k]\mid J\text{ is Band-}j\}$ denote the collection of all Band-$j$ subsets.

\paragraph{Levels} Let $l_{max}=\lceil z\log \eps^{-1}\rceil$. For any subset $J\subseteq[k]$, any $k$-center $C\in \binom{\mathbb{R}^d}{k}$, and any $j\in [l_{max}]$, $C$ is called a level-$j$ center of $J$ if $\dist(a_i,C)\in (2^{(j-1)/z}\Delta_i^{1/z}, 2^{j/z} \Delta_i^{1/z}]$ for all $i\in J$. 
In addition, $C$ is called a level-$0$ center of $J$ if $\dist(a_i,C)\leq \Delta_i^{1/z}$ for all $i\in J$. Let $L^J_j\subset \binom{\mathbb{R}^d}{k}$ denote the set of all level-$j$ centers of $J$. Furthermore, let $L^{J}_{-1}=\{C\in \binom{\R^d}{k}\mid \forall i\in J,\ \dist(a_i,C) > \eps^{-1} \Delta_i^{1/z}\}$.

It suffices to focus on the most interesting cases $j\in [b_{max}],l\in [l_{max}]$, the cases $j=0$ or $l=-1$ can be handled almost identically to those in Section~\ref{sec_Euc_analysis}.

\paragraph{Cost Vector} For a subset $J\in \mathcal{B}_j$, and a $k$-center $C\in L_{l}^J$, we define then cost vector $v_C^J:U\rightarrow \mathbb{R}_{\geq 0}$ as,
$$
v_C^J(x)=
\begin{cases}
\dist^z(x,C) & \text{if }x\in U_J\\
0 & \text{otherwise}.
\end{cases}
$$

The distance vector set is defined as $V^J=\{v_C^J\mid C\in L_{l}^J\}$. Again, by symmetrization trick, it remains to bound the following error term.

$$
\E\sup_{J,C}\bigg|\frac{\sum_{x\in U} g_x w_U(x)v_C^J(x)}{\cost(X_J,C)+\cost(X,A)}\bigg|.
$$

The construction of $\alpha$-net in \cite{bansal2024sensitivity} relies heavily on the geometric structure of $k$-median and $k$-means. For general $z\geq 1$, we instead adapt the $\alpha$-net from \cite{huang2024optimal}. In this approach, center sets are no longer partitioned based on similar interaction numbers. In the following, we fix $j\in [b_{max}]$ and $l\in [l_{max}]$.

\paragraph{Net Construction} An $\alpha$-net $N_{\alpha}=N_{\alpha,J,j,l}$ for $V^J$ is a subset of $V^J$ such that for any $v_C^J\in V^J$, there exists a vector $q\in N_{\alpha}$ satisfying the following properties for all $x\in U$,
\begin{enumerate}[itemsep = 0pt]
\item $|q(x)-v_C^J(x)|\leq \alpha \cdot \operatorname{err}(x,C)$ if $v_C^J(x)>0$;
\item $q(x)=0$ if $v_C^J(x)=0$,
\end{enumerate}
where
\[
\operatorname{err}(x,C) := \left(\sqrt{\dist^z(x,C)\dist^z(x,A)} + \dist^z(x,A)\right)\cdot \sqrt{\frac{\cost_z(X_J,C)+\cost_z(X,A)}{\cost_z(X,A)}}.
\]

\begin{lemma}[Adaptation of Lemmata B.2 and B.11 of \cite{huang2024optimal}] \label{Lemma_Euc_net_z}
For $\alpha\in (0,1]$, there exists an $\alpha$-net of $V^J$ satisfying the following metric entropy bound:
\[
\log |N_{\alpha}|\lesssim \alpha^{-2} \cdot \min\{k^\frac{2z+2}{z+2},k\eps^{-z}\}\cdot \log^2(k\eps^{-1})\cdot \log^2(\eps^{-1}) \cdot \log (\alpha^{-1}\eps^{-1}).
\]
\end{lemma}

\begin{proof}
The adaptation consists of two steps. First, we partition $X$ into $q = O(\log(1/\eps)\log(k/\eps))$ subsets and obtain an $\alpha$-net for each subset using the net constructed in \cite[Lemma B.2]{huang2024optimal}. 
We remark that \cite[Lemma B.2]{huang2024optimal} only provides bounds for the first term $\tilde{O}(k^\frac{2z+2}{z+2} \alpha^{-2})$.
However, we note that this term comes from \cite[Lemma B.11]{huang2024optimal} and serves as the upper bound of the term
\[
T = \tilde{O}(\min\{ 2^{\beta z} k, k^2 2^{-2\beta} \} \cdot \alpha^{-2}).
\]
We note that their range of $\beta$ is $0\leq \beta \leq 2 + \log (z \eps^{-1})$.
Thus, $\tilde{O} ( k\eps^{-z} \alpha^{-2})$ can also serve as an upper bound of $T$, which results in the second bound in the lemma.
Overall, the above results in the metric entropy multiplied by $q$, giving a bound of the same order as the claimed result. 
Second, the $\alpha$-net in \cite{huang2024optimal} is designed to satisfy $\abs{q(x) - v_C^J(x) - \dist(a_{\pi(x)},C)} \leq \alpha\cdot \operatorname{err}(x,C)$. To remove the $\dist(a_{\pi(x)},C)$ term, we use another net to approximate it, leading to an additive lower-order term in the metric entropy bound.
\end{proof}

\paragraph{Telescoping} Let $h_{max}=\lceil (z+3)\log \frac{k}{\eps}\rceil$ and, for each $h\in \{0,\dots,h_{max}\}$, let $q_C^h\in N_{2^{-h}}$ denote the net vector of $v_C^J$. We rewrite $v_C^J(X)$ as a telescoping sum in exactly the same manner as \eqref{eq_Euc_tel}.
%
Similarly, it remains to control the initial term, main terms, and the residual individually. Let $\nu=5$ in the following.

\begin{lemma}[Residual term]
\begin{eqnarray*}
\E\sup_{J,C}\bigg|\frac{\sum_{x\in U} g_x w_U(x)\cdot (v_C^J(x)-q_C^{h_{max}}(x))}{\cost(X_J,C)+\cost(X,A)}\bigg|\lesssim \frac{\eps}{\log^\nu \frac{k}{\eps}}.
\end{eqnarray*}
\end{lemma}

\begin{proof}
Let $\beta \asymp \eps/\log^{\nu}(k/\eps)$.
By the choice of $h_{max}$, we have that 
\begin{align*}
|v_C^J(x)-q_C^{h_{max}}|
&\lesssim (\eps/k)^{z+3} \cdot (\sqrt{\dist^z(x,C)\dist^z(x,A)}+\dist^z(x,A))\cdot \eps^{-z/2}\\
&\lesssim (\eps/k)^{z/2+3} \cdot (\sqrt{\dist^z(x,C)\dist^z(x,A)}+\dist^z(x,A)) \\
&\lesssim (\eps/k)^{z/2+3} \cdot \eps^{-z/2} \dist^z(x,A) \\
&\lesssim \beta^2 \dist^z(x,A).
\end{align*}
Thus, 
\[
\frac{w_U(x) \abs{v_C^J(x)-q_C^{h_{max}}}}{\cost_z(X_J,C)+\cost(X,A)}\lesssim \frac{\beta^2}{s}\cdot \frac{\cost_z(X,A)}{\cost_z(X_J,C)+\cost_z(X,A)} \leq \frac{\beta^2}{s}.
\]
Similarly to the proof of Lemma~\ref{lemma_euc_residual}, by Cauchy-Schwarz inequality, 
\[
\E\sup_{J,C}\bigg|\frac{\sum_{x\in U} g_x w_U(x)\cdot (v_C^J(x)-q_C^{h_{max}}(x))}{\cost(X_J,C)+\cost(X,A)}\bigg|\lesssim \E \|g\|_2 \cdot \frac{\beta}{\sqrt{s}}\leq \beta. \qedhere
\]
\end{proof}

\begin{lemma}[Main telescoping terms] \label{lemma_Euc_chaining_z_main_main}
For each $h\in [h_{max}]$,
$$
\E\sup_{J,C}\bigg|\frac{\sum_{x\in U} g_x w_U(x)\cdot (q_C^h(x)-q_C^{h-1}(x))}{\cost(X_J,C)+\cost(X,A)}\bigg|\lesssim \frac{\eps}{\log^{\nu+1}\frac{k}{\eps}}.
$$
\end{lemma}

Let $G(J,C)=\frac{\sum_{x\in U} g_x w_U(x)\cdot (q_C^h(x)-q_C^{h-1}(x))}{\cost(X_J,C)+\cost(X,A)}$, we have the following bound on the variance of $G(J,C)$. Lemma~\ref{lemma_Euc_chaining_z_main_main} can be proved by combining Lemma~\ref{Lemma_Euc_net_z} and Lemma~\ref{lemma_Euc_z_variance1}.

\begin{lemma}[Variance bound of main telescoping terms]
\label{lemma_Euc_z_variance1}
\begin{gather}
\Var[G(J,C)\mid U]\leq \frac{2^{-2h}}{s}\cdot \frac{\cost_z(U_J,C)+\cost_z(U_J,A)}{\cost_z(X_J,C)+\cost_z(X,A)} \notag
\shortintertext{and}
    \E_U\bigg[\sup_{J,C} \bigg(\frac{\cost_z(U_J,C)+\cost_z(U_J,A)}{\cost_z(X_J,C)+\cost_z(X,A)}\bigg)^{1/2}\bigg]\lesssim 1. \label{inq_Euc_z_chaining_var1}
\end{gather}
\end{lemma}
\begin{proof}

By plugging the error term of $\alpha$-net, we have that,
\begin{align*}
&\Var[G(J,C)\mid U]\\
\lesssim{} & \frac{2^{-2h}\sum_{x\in U_J} w_U(x)^2 (\dist^z(x,C)\dist^z(x,A)+\dist^z(x,A)^2)}{(\cost_z(X_J,C)+\cost_z(X,A))^2}\cdot \frac{\cost_z(X_J,C)+\cost_z(X,A)}{\cost_z(X,A)}\\
={} & \frac{2^{-2h}\sum_{x\in U_J} w_U(x)^2 \cdot \dist^z(x,A)\cdot (\dist^z(x,C)+\dist^z(x,A))}{(\cost_z(X_J,C)+\cost_z(X,A))\cdot \cost_z(X,A)}.
\end{align*}
By Lemma~\ref{lemma_Euc_weight}, we have $w_U(x)\dist^z(x,A)\leq \frac{4}{s}\cdot \cost_z(X,A)$. It follows that
\[
\Var[G(J,C)\mid U]\lesssim \frac{2^{-2h}}{s}\cdot \frac{\cost(U_J,C)+\cost(U_J,A)}{\cost(X_J,C)+\cost(X,A)}.
\]
Next, we prove \eqref{inq_Euc_z_chaining_var1}. By an identical argument to the proof of Lemma~\ref{lemma_chaining_var1}, it always holds that
\[
\frac{\cost_z(U_J,A)}{\cost_z(X_J,C)+\cost_z(X,A)}\leq 4.
\] 
Thus, it suffices to prove that
\[
\E_U\bigg[\sup_{J,C} \bigg(\frac{\cost_z(U_J,C)}{\cost_z(X_J,C)+\cost_z(X,A)}\bigg)^{1/2}\bigg]\lesssim 1.
\]

When $\mathcal{E}_0\cup \mathcal{E}_1$ holds, we have that, by the generalized triangle inequality (Lemma~\ref{lm:triangle}),
$$
 \cost_z(U_J,C)\lesssim \cost_z(X_J,C)+\cost_z(X_J,A)+\cost_z(U_J,A).
$$
On the other hand, it always holds that
\begin{align*}
\cost_z(U_J,C) = \sum_{x\in U_J} w_U(x)\dist^z(x,C) &\lesssim 2^{l_{max}} \sum_{x\in U_J}w_U(x)\dist^z(x,A)\\
&\lesssim 2^{l_{max}} \sum_{x\in U_J} \frac{\cost_z(X,A)}{s}\\
&\lesssim 2^{l_{max}} \cdot \cost_z(X,A)\\
&\lesssim \eps^{-z} \cdot \cost_z(X,A).
\end{align*}
Thus, 
\begin{align*}
\E_U\bigg[\sup_{J,C} \bigg(\frac{\cost_z(U_J,C)}{\cost_z(X_J,C)+\cost_z(X,A)}\bigg)^{1/2}\bigg]&\lesssim  1\cdot \mathbb{P}[\mathcal{E}_0\cup \mathcal{E}_1] +\eps^{-z} \cdot (1-\mathbb{P}[\mathcal{E}_0\cup \mathcal{E}_1] )\\
&\lesssim  1+\eps^{-z}\cdot \frac{\eps^z}{k}\\
&\lesssim 1. \qedhere
\end{align*}
\end{proof}

\begin{lemma}[Initial telescoping term] \label{lemma_Euc_chaining_z_coarse}
\begin{eqnarray*}
\E\sup_{J,C}\bigg|\frac{\sum_{x\in U} g_x w_U(x)\cdot q_C^0(x)}{\cost(X_J,C)+\cost(X,A)}\bigg|\lesssim \frac{\eps}{\log^\nu \frac{k}{\eps}}.
\end{eqnarray*}
\end{lemma}

\begin{proof}
Similar to the proof of Lemma~\ref{lemma_Euc_chaining_coarse}, we define a rounding term of $d^z(a_i,C)$, denoted by
\[
d_i = \big\lceil \frac{\dist^z(a_i,C)}{\Delta_i} \big\rceil\cdot \Delta_i.
\]
Similarly, the number of distinct tuples of $(d_i)_{i\in J}$ can be bounded by $2^{O(kl_{max})}\leq |N_{1/2}|$. Now we consider
\begin{equation}\label{term_Euc_z_inital}
\E\sup_{J,C}\bigg|\frac{\sum_{x\in U_J} g_x w_U(x)\cdot (q_C^0(x)-d_{\pi(x)})}{\cost(X_J,C)+\cost(X,A)}\bigg|.
\end{equation}

We have $|q_C^0(x) - d_{\pi(x)}|\lesssim \operatorname{err}(x,C) + \dist^z(x,A) + \Delta_{\pi(x)}$. Moreover, the number of distinct vectors induced by $(q_C^0(x)-d_{\pi(x)})_{x\in U}$ can be bounded by $2^k|N_{1/2}|^2$. Hence, (\ref{term_Euc_z_inital}) has essentially the same order variance and logarithm of net size as
$$
\E\sup_{J,C}\bigg|\frac{\sum_{x\in U_J} g_x w_U(x)\cdot (q_C^0(x)-q_C^1(x))}{\cost(X_J,C)+\cost(X,A)}\bigg|,
$$
thus can be controlled in the same manner as in the proof of Lemma~\ref{lemma_Euc_chaining_z_main_main}. 
It remains to control 
\begin{equation}\label{term_euc_z_initial_2}
\E\sup_{J,C}\bigg|\frac{\sum_{x\in U_J} g_x w_U(x)\cdot d_{\pi(x)}}{\cost(X_J,C)+\cost(X,A)}\bigg|,
\end{equation} 
which has an identical form as the first term of (\ref{eqn:chaining_euc_initial_aux}). Thus, (\ref{term_euc_z_initial_2}) can be controlled identically.
\end{proof}

\end{document}